\documentclass[11pt]{article}

\usepackage{palatino}
\usepackage{hyperref}

\usepackage{tikz}
\usepackage{tikzpeople}
\usetikzlibrary{tikzmark, shapes, fit,backgrounds}
\usetikzlibrary{matrix, decorations.pathreplacing, angles, quotes}
\usetikzlibrary{patterns, calc, positioning}
\usetikzlibrary{arrows, arrows.meta}
\usetikzlibrary{shapes.multipart}
\usetikzlibrary{fit}
\usetikzlibrary{shapes.geometric, positioning}
\usetikzlibrary{decorations.pathmorphing}
\tikzstyle{block} = [rectangle, draw, 
    minimum width=4em, text centered, rounded corners, minimum height=1em]
\tikzstyle{line} = [draw, -latex]
\tikzset{meter/.append style={draw, inner sep=10, rectangle, font=\vphantom{A}, minimum width=30, scale=.7, path picture={\draw[black] ([shift={(.1,.3)}]path picture bounding box.south west) to[bend left=50] ([shift={(-.1,.3)}]path picture bounding box.south east);\draw[black,-{Latex[scale=.5]}] ([shift={(0,.1)}]path picture bounding box.south) -- ([shift={(.3,-.1)}]path picture bounding box.north);}}}
\tikzset{snake it/.style={decorate, decoration=snake}}

\def\ppsmatrix#1{\begin{psmallmatrix}#1\end{psmallmatrix}}
\usetikzlibrary{calc}
\usepackage{bm}
\usepackage[utf8]{inputenc}
\usepackage{mathtools}
\usepackage{graphicx}
\usepackage{algorithm}
\usepackage{cite}
\usepackage{amsmath,amssymb,amsfonts, amsthm,mathrsfs}
\usepackage{algorithmic}
\usepackage{textcomp}
\usepackage{xcolor}
\usepackage{enumitem}
\usepackage{braket}
\usepackage{diagbox}
\usetikzlibrary{patterns}
\usepackage{mathdots}

\newtheorem{lemma}{Lemma}
\newtheorem{theorem}{Theorem}

\newtheorem{remark}{Remark}
\newtheorem{corollary}{Corollary}
\newtheorem{definition}{Definition}
 
\newcommand\rk{\normalfont{\mbox{rk}}}
\newcommand{\tr}{\normalfont{\mbox{tr}}}
\newcommand{\dec}{\mbox{\footnotesize dec}}

\newcommand{\fullent}{\normalfont{\mbox{\tiny fullent}}}
\newcommand{\unent}{\normalfont{\mbox{\tiny unent}}}
 
\newcommand{\biject}{\leftrightarrow}

\allowdisplaybreaks

\title{The Capacity of Classical Summation over a Quantum MAC with Arbitrarily Distributed Inputs and  Entanglements}

\author{Yuhang Yao, Syed A. Jafar\\
{\small Center for Pervasive Communications and Computing (CPCC)}\\
{\small University of California Irvine, Irvine, CA 92697}\\
{\small \it Email: \{yuhangy5, syed\}@uci.edu}
\thanks{Presented in part at IEEE GLOBECOM 2023 \cite{Yao_Jafar_Sum_MAC_GC}.}
}

\date{}      

\begin{document}
\maketitle

\begin{abstract}
The $\Sigma$-QMAC problem is introduced, involving $S$ servers, $K$  classical ($\mathbb{F}_d$) data streams, and $T$ independent quantum systems. Data stream ${\sf W}_k, k\in[K]$ is replicated at a subset of servers $\mathcal{W}(k)\subset[S]$, and quantum system $\mathcal{Q}_t, t\in[T]$ is distributed among a subset of servers $\mathcal{E}(t)\subset[S]$ such that  Server $s\in\mathcal{E}(t)$ receives  subsystem $\mathcal{Q}_{t,s}$ of $\mathcal{Q}_t=(\mathcal{Q}_{t,s})_{s\in\mathcal{E}(t)}$. Servers   manipulate their quantum subsystems according to their  data and send the subsystems to a receiver. The  total download cost is $\sum_{t\in[T]}\sum_{s\in\mathcal{E}(t)}\log_d|\mathcal{Q}_{t,s}|$ qudits, where $|\mathcal{Q}|$ is the dimension of $\mathcal{Q}$. The states and measurements of $(\mathcal{Q}_t)_{t\in[T]}$ are required to be separable across $t\in[T]$ throughout, but for each $t\in[T]$, the  \emph{subsystems} of $\mathcal{Q}_{t}$ can be prepared initially in an arbitrary (independent of data) entangled state, manipulated arbitrarily by the respective servers, and measured jointly by the receiver. From the  measurements, the receiver must  recover the sum of all data streams. Rate  is defined as the number of dits ($\mathbb{F}_d$ symbols) of the desired sum computed per qudit of download. The capacity of $\Sigma$-QMAC, i.e., the supremum of achievable rates is characterized for arbitrary data and entanglement distributions $\mathcal{W}, \mathcal{E}$. 
For example, in the symmetric setting with $K=\scalebox{0.8}{$\binom{S}{\alpha}$}$ data-streams, each replicated among a distinct $\alpha$-subset of $[S]$, and $T=\scalebox{0.8}{$\binom{S}{\beta}$}$  quantum systems, each distributed among a distinct $\beta$-subset of $[S]$, the capacity of the $\Sigma$-QMAC is $\frac{1}{\beta T}\sum_{\gamma = (\alpha+\beta-S)^+}^{\min(\alpha,\beta)} \min(\beta,2\gamma) \cdot \binom{\alpha}{\gamma}\cdot \binom{S-\alpha}{\beta-\gamma}$.
Coding based on the $N$-sum box abstraction is optimal in every case. Notably, for every $S\neq 3$ there exists an instance of the $\Sigma$-QMAC where  $S$-party entanglement is necessary to achieve the fully entangled capacity.
\end{abstract}

\section{Introduction}
Entanglement is arguably the most counter-intuitive aspect of quantum systems. Quantum entanglement enables  correlations that are classically impossible. These correlations can be exploited for improvements in the efficiency of  quantum communication and computation networks. Understanding the fundamental limits of quantum entanglement phenomena is therefore essential to gauge the potential of the much-anticipated quantum internet of the future \cite{Caleffi_tutorial1, caleffi_tutorial2, caleffi_survey}. However, even quantifying the amount of entanglement is highly non-trivial, especially when the entanglement is distributed among \emph{many} parties. Unlike bipartite entanglement which is relatively well understood --- a fundamental  understanding of multi-party\footnote{We refer to entanglement among $N$-parties as bipartite if $N=2$, and multiparty if $N>2$.} entanglement remains elusive. Numerous fundamentally distinct measures have been explored thus far\cite{walter2016multipartite}, including the Schmidt measure \cite{eisert2001schmidt, Schmidtmeasure}, the trace-squared or the entropy of the reduced density matrix \cite{Bennett_Ent}, the tangle \cite{Tangle}, the entanglement of formation \cite{EOF}, majorization-based entanglement monotones \cite{Vidal_monotone,Nielsen_majorization}, geometric measures \cite{DegEnt, Barnum2001, GME}, and specialized notions such as absolute maximal or genuine multiparty entanglement \cite{helwig2012absolute,huber2013structure}. 

Given the lack of a universal measure, an alternative is to quantify multi-party entanglement indirectly in terms of its utility as a resource \cite{DSC4}, e.g., by the gains in communication efficiency\footnote{Indeed, it is conjectured that communication efficiency may provide a concise information-theoretic axiomatic basis for characterizing quantum mechanics\cite{Shutty_Wootters_Hayden}.}  that are made possible by quantum  entanglement for accomplishing various classical multiparty computation tasks. 
Much effort has traditionally been aimed at finding tasks that  gain a lot from quantum entanglement in terms of communication complexity measures \cite{Yao_QC, QSM}.
A limitation of this approach is that the tasks thus identified may  turn out to be   artificial. For example, the celebrated Deutsch-Jozsa algorithm \cite{deutsch1992rapid} translates into a  computation task in a $2$ user quantum SMP setting \cite{brassard1999cost, PSQM} where entanglement shows an exponential advantage. However, such a computation task is seldom encountered in practice.

A complementary approach that we explore in this work, is to focus instead on some elementary computation tasks that are quite natural, such as linear computations, and  explore the gains in efficiency due to quantum entanglements for such tasks. Specifically, we study an elemental setting, called the $\Sigma$-QMAC, to be described shortly, where the computation task is simply a finite-field summation task, over an ideal (noise-free) quantum multiple access network  with arbitrarily distributed inputs and entanglements. Since the gains may be modest, a finer accounting of efficiency, such as the exact information theoretic capacity, is needed.

The pursuit of exact capacity faces numerous challenges -- 
1) sharp capacity characterizations are quite rare even for classical communication networks when many parties are involved,  2) computation networks tend to be even less tractable than communication networks for information theoretic analysis, and 3) the quantum setting  further compounds the difficulty of any such endeavor. Indeed the foremost challenge in pursuing this direction is to identify formulations that are both insightful from a multiparty quantum entanglement perspective and also information theoretically tractable. Our choice of the $\Sigma$-QMAC setting draws inspiration from the following observations --- 1) the many-to-one (multiple access (MAC)) setting is among the most tractable in network information theory, 2)  the capacity of finite field linear computations in noiseless settings, while still open in general, has seen much progress in the network coding literature,  in particular elegant solutions have been found  for the case of scalar linear computations (i.e., sum computations)  \cite{Rai_Dey, Ramamoorthy_Langberg, Nazer_Gastpar_Compute, Appuswamy1, Appuswamy3}, and 3) the  stabilizer formalism \cite{Gottesman97,Calderbank_Shor_CSS_code,Steane_CSS_code} from quantum error correction lends itself nicely to linear black-box abstractions for the quantum multiple access channel (QMAC)\cite{Allaix_N_sum_box}, and  has been instrumental to recent advances in quantum private information retrieval (QPIR) \cite{song_multiple_server_PIR, song_colluding_PIR, QMDSTPIR, song_all_but_one_collusion, Allaix_N_sum_box,aytekin2023quantum} that implicitly involve optimal specialized  linear computations. 
The $\Sigma$-QMAC  setting is  described next.

\begin{figure}[h]
\begin{center}
\begin{tikzpicture}
\coordinate (O) at (0,0){};

\node [draw, cylinder, aspect=0.2,shape border rotate=90,fill=black!10, text=black, inner sep =0cm, minimum height=1.1cm, minimum width=1.5cm, above right = 0cm and -3cm of O, align=center] (AB)  {\footnotesize ~\\[-0.25cm]\scalebox{0.7}{(Server $\mathcal{S}_1$)}\\[0cm]\footnotesize ${\sf A},{\sf B}$};
\node [draw, cylinder, aspect=0.2,shape border rotate=90,fill=black!10, text=black, inner sep =0cm, minimum height=1.1cm, minimum width=1.5cm, above right = 0cm and -1cm of O, align=center] (AC)  {\footnotesize ~\\[-0.25cm]\scalebox{0.7}{(Server $\mathcal{S}_2$)}\\[0cm]\footnotesize ${\sf A},{\sf C}$};
\node [draw, cylinder, aspect=0.2,shape border rotate=90,fill=black!10, text=black, inner sep =0cm, minimum height=1.1cm, minimum width=1.5cm, above right = 0cm and 1cm of O, align=center] (BC)  {\footnotesize ~\\[-0.25cm]\scalebox{0.7}{(Server $\mathcal{S}_3$)}\\[0cm]\footnotesize ${\sf B},{\sf C}$};
\node [draw, cylinder, aspect=0.2,shape border rotate=90,fill=black!10, text=black, inner sep =0cm, minimum height=1.1cm, minimum width=1.5cm, above right = 0cm and 3cm of O, align=center] (D)  {\footnotesize ~\\[-0.25cm]\scalebox{0.7}{(Server $\mathcal{S}_4$)}\\[0cm]\footnotesize $~{\sf D}~$};

\node[alice, draw=black, text=black, minimum size=0.8cm, inner sep=0, above right = -3cm and 0cm of O] (U) at (0,0) {\footnotesize Alice};

\node[rectangle, fill=pink!10, rounded corners=15, minimum width = 2.7cm, minimum height=1.7cm, above right = 1.7cm and -3.3cm] (Q1) {};
\node[align=center] at ($(Q1.north)+(0,-0.4)$) {\tiny \color{magenta}Quantum system $\mathcal{Q}_1$ \\[-0.2cm] \tiny \color{magenta} Initial state: $\rho_1$};

\node[rectangle, fill=blue!5, rounded corners=15, minimum width = 3.6cm, minimum height=1.7cm, above right = 1.7cm and 0.3cm] (Q2) {};
\node[align=center] at ($(Q2.north)+(0,-0.4)$) {\tiny \color{violet} Quantum system $\mathcal{Q}_2$ \\[-0.2cm] \tiny \color{violet} Initial state: $\rho_2$};

\node [draw, circle, dotted, aspect=0.2, fill=pink!30, text=blue, inner sep =0cm, minimum width=0.8cm, align=center, above right = 2cm and -2.8cm of O] (Q11)  {\color{magenta} $\mathcal{Q}_{1,1}$};
\node [draw, circle, dotted, aspect=0.2, fill=pink!30, text=blue, inner sep =0cm, minimum width=0.8cm, align=center, above right = 2cm and -1.6cm of O] (Q12)  {\color{magenta} $\mathcal{Q}_{1,2}$};

\node [draw, circle, dotted, aspect=0.2, fill=blue!15, text=blue, inner sep =0cm, minimum width=0.8cm, align=center, above right = 2cm and 0.6cm of O] (Q22)  {\color{violet}$\mathcal{Q}_{2,2}$};
\node [draw, circle, dotted, aspect=0.2, fill=blue!15, text=blue, inner sep =0cm, minimum width=0.8cm, align=center, above right = 2cm and 1.8cm of O] (Q23)  {\color{violet}$\mathcal{Q}_{2,3}$};
\node [draw, circle, dotted, aspect=0.2, fill=blue!15, text=blue, inner sep =0cm, minimum width=0.8cm, align=center, above right = 2cm and 3cm of O] (Q24)  {\color{violet}$\mathcal{Q}_{2,4}$};

\draw[-, magenta, snake it] (Q11.south) to  (AB.north) ;
\draw[-, magenta, snake it] (Q12.south) to  ($(AC.north)+(-0.4,0)$) ;
\draw[-, blue, snake it] (Q22.south) to  ($(AC.north)+(+0.4,0)$) ;
\draw[-, blue, snake it] (Q23.south) to  (BC.north) ;
\draw[-, blue, snake it] (Q24.south) to  (D.north) ;

\node (mx1) [meter, above right = -1.8cm and -1.8cm] {};
\node (mx2) [meter, above right = -1.8cm and 1.5cm] {};

\draw[-, magenta, snake it] (AB.south) to  ($(mx1.north)+(-0.2,0)$) ;
\draw[-, magenta, snake it] ($(AC.south)+(-0.4,0)$) to  ($(mx1.north)+(+0.2,0)$) ;
\draw[-, blue, snake it] ($(AC.south)+(+0.4,0)$) to  ($(mx2.north)+(-0.2,0)$) ;
\draw[-, blue, snake it] (BC.south) to  ($(mx2.north)+(0,0)$) ;
\draw[-, blue, snake it] (D.south) to  ($(mx2.north)+(0.2,0)$) ;

\draw[-latex] (mx1)--(U) node[pos=0.2,right=0.1] {\footnotesize $Y_1$};
\draw[-latex] (mx2)--(U) node[pos=0.2,left=0.1] {\footnotesize $Y_2$};

\node (Ans) [right= 0.75cm of U]{\footnotesize $({\sf A}_\ell+{\sf  B}_\ell+{\sf C}_\ell+{\sf D}_\ell)_{\ell\in[L]}$};
\draw[-latex] (U)--(Ans);

\end{tikzpicture}
\end{center}
\vspace*{-3mm}\caption{A $\Sigma$-QMAC  setting, with $K=4$ data streams $(\sf{A}, \sf{B}, \sf{C}, \sf{D})$, $S=4$ servers $(\mathcal{S}_1,\mathcal{S}_2,\mathcal{S}_3,\mathcal{S}_4)$, and $T=2$ quantum systems $(\mathcal{Q}_1,\mathcal{Q}_2)$. The data replication map $\mathcal{W}=(\{1,2\},\{1,3\},\{2,3\},\{4\})$ specifies that data stream $\sf{A}$ is replicated at Servers $\mathcal{S}_1,\mathcal{S}_2$;  $\sf{B}$ at $\mathcal{S}_1,\mathcal{S}_3$; $\sf{C}$ at $\mathcal{S}_2,\mathcal{S}_3$; and data stream $\sf{D}$ is available only at $\mathcal{S}_4$. The entanglement distribution map $\mathcal{E}=(\{1,2\},\{2,3,4\})$ is such that entangled subsystems of $\mathcal{Q}_1$ are distributed to Servers $\mathcal{S}_1,\mathcal{S}_2$, and entangled subsystems of $\mathcal{Q}_2$ are distributed to Servers $\mathcal{S}_2$, $\mathcal{S}_3$ and $\mathcal{S}_4$.}\label{fig:abbccad}
\end{figure}

In a $\Sigma$-QMAC setting (see Section \ref{sec:probform} for a formal definition), a user, say Alice, wants to compute the \emph{sum} of $K$ classical data streams (comprised of $d$-ary symbols (\emph{dits}) from a finite field $\mathbb{F}_d$) that are replicated across various subsets of $S$ servers (cf. graph-based replication \cite{Raviv_Tamo_Yaakobi,Jia_Jafar_GXSTPIR,Fei_Chen_Wang_Jafar}) according to an arbitrary data replication map $(\mathcal{W})$. Independent of the data streams, $T$ quantum systems are prepared, and  subsystems of each quantum system are  distributed to various subsets of servers (called cliques) according to an arbitrary entanglement distribution map ($\mathcal{E}$).  The states and the eventual measurements of different quantum systems must remain separable throughout, but the \emph{subsystems} of each quantum system are in general entangled even as they are distributed to different servers within that clique, thus allowing such a clique of servers to exploit their quantum entanglement. The servers locally encode their classical data into their quantum subsystems, maintaining separation among different quantum systems associated with different cliques, and send them to Alice, who does separate measurements on each quantum system. From the measurement outcome, Alice must be able to recover the desired sum. The computation rate is the number of dits of the desired sum computed by Alice per qudit of download.\footnote{In contrast to a \emph{dit}, which is a classical $d$-ary symbol, a \emph{qudit}, short for a quantum-dit, represents a $d$-dimensional quantum system. For $d=2$ these are the common `bit' and `qubit,' respectively.} An example is illustrated in Fig. \ref{fig:abbccad}. If Alice is able to compute $L$ \emph{dits} of the desired sum (${\sf A+B+C+D}$ in Fig. \ref{fig:abbccad}) with total communication cost  $N$ \emph{qudits}  then the rate achieved is $L/N$ (dits/qudit).  The capacity $C$ is the supremum of achievable rates. 

In order to quantify the utility of multiparty quantum entanglements, the key figure of merit in the $\Sigma$-QMAC is the multiplicative gain in capacity that is enabled by entanglement, relative to the corresponding unentangled setting.  In the literature such a gain is known as \emph{distributed superdense coding gain} \cite{Superdense, DSC1,DSC2,DSC3,DSC4, QuantumAdvantage, song_multiple_server_PIR, song_colluding_PIR, QMDSTPIR, song_all_but_one_collusion,Allaix_N_sum_box}. 
Quantifying the utility of multiparty entanglements by characterizing the distributed superdense coding (DSC) gain in the $\Sigma$-QMAC is the immediate focus of this work. The broader motivation is that success in the $\Sigma$-QMAC setting may pave the way for future studies of \emph{general} linear computation tasks that shed further light on the fundamental limits of the utility of multiparty quantum entanglement.

\section{Significance of the $\Sigma$-QMAC  and  relationship to prior works}
\subsection{Relationship to prior works}
{\bf Connection to Simultaneous Message Passing (SMP):} The underlying quantum multiple access  (QMAC) communication model in the $\Sigma$-QMAC is  similar to  what is known in the literature as \emph{simultaneous message passing} (SMP) model with quantum messages \cite{buhrman1998quantum, QSM, PSQM} . A noteworthy distinction is that the SMP model is typically studied from a communication complexity perspective which does not allow batch processing, whereas since our perspective is information theoretic, batch processing is not only allowed, it is  essential to our problem formulation. For brevity, and to underscore the information theoretic perspective, we say QMAC when we mean an SMP model with quantum messages and batch processing.

{\noindent{\bf Connection to Quantum Metrology:} The  $\Sigma$-QMAC is conceptually related to various physically motivated and commonly studied models in the active area of distributed quantum sensing and quantum metrology. A general theme in this area is how the entanglement across quantum sensors allows higher precision (approaching the Heisenberg limit) in the computation of a function of distributed classical parameters, than what is  possible without entanglement (the standard Quantum limit) \cite{Lloyd}. For example, note the similarity of Fig. \ref{fig:abbccad} and the quantum metrology protocol illustrated in \cite{Uman}. In the quantum metrology protocol, entangled quantum systems are distributed to sensors (which take the role of servers in Fig. \ref{fig:abbccad}), classical parameters are encoded into them through quantum transducers, and the quantum systems are sent to a central receiver where the desired function is estimated by a joint measurement.  As noted in \cite{Childs1999}, the gains in capacity due to quantum entanglement in an elemental QMAC can be harnessed to yield gains in measurement precision. Entanglement is particularly useful if the goal is to estimate a global function of distributed parameters rather than a separate estimation of each parameter \cite{GlobalMetrology, DistributedQS}. The function of parameters to be computed could be the average value, essentially a sum of parameters as in \cite{QRF}, similar to the $\Sigma$-QMAC. Remarkably,  establishing the utility of multipartite entanglement by explicitly accounting for the partitioning of entangled systems as illustrated in Fig. \ref{fig:abbccad}, is also of interest in quantum metrology \cite{Young}. While our current focus is limited to finite fields, a successful capacity characterization of the $\Sigma$-QMAC could perhaps be a stepping stone to the much more challenging problem of determining the fundamental limits of stochastic computations over continuous random variables as is often needed in quantum metrology.}

\noindent{\bf Connection to Quantum Private Information Retrieval (QPIR):} The $\Sigma$-QMAC model over finite fields is  essential in the  context of QPIR \cite{song_multiple_server_PIR, song_colluding_PIR, QMDSTPIR, song_all_but_one_collusion}. In QPIR, classical messages (possibly in coded forms) are stored at servers  that  have entangled quantum systems distributed among them. Based on a query submitted by a user, each server encodes its classical answer into its quantum system, sends it to the user through an elemental QMAC, and the user retrieves the desired message by a joint measurement. As in the $\Sigma$-QMAC, the desired computation is typically a linear function of the servers' classically coded inputs. 

\noindent{\bf Connection to $N$-sum Box:} The $\Sigma$-QMAC is also immediately related to the $N$-sum box  \cite{Allaix_N_sum_box}, which is a black box abstraction of stabilizer based linear computations. The DSC gains in many previously studied settings, including recent applications in quantum private information retrieval (QPIR), can be realized through the $N$-sum box abstraction \cite{DSC1,DSC2,DSC3,DSC4, QuantumAdvantage, song_multiple_server_PIR, song_colluding_PIR, QMDSTPIR, song_all_but_one_collusion}. Nonetheless, it is important to note that the $\Sigma$-QMAC capacity formulation does not limit the coding schemes to the $N$-sum box, or to stabilizer based constructions in general. Indeed, the $\Sigma$-QMAC allows arbitrary entangled states to be shared among the servers, and stabilizer states are a very small fraction of those states. Similarly, the $\Sigma$-QMAC allows servers to perform arbitrary unitary transformations, whereas the $N$-sum box is limited to  $X$ and $Z$ gates, which represent a very small fraction of all possible unitaries. Whether the stabilizer based $N$-sum box suffices to achieve the capacity of the $\Sigma$-QMAC in all cases is another fundamental question to be answered in this work.

\noindent{\bf Connection to Network Function Computation:} While different from the traditional multiple access channel, the capacity formulation for the $\Sigma$-QMAC is quite standard in the broad area of  \emph{classical network function computation} \cite{Kowshik_Kumar, NetworkFC, Huang_Tan_Yang_Guang}, from where it is inherited. In network function computation,  capacity is defined in terms of the communication cost incurred per instance of a desired function computation in a distributed setting. The network function computation paradigm finds numerous applications in private information retrieval, coded computing, distributed storage repair and learning \cite{PIR_tutorial}. In particular the capacity of sum-networks, where the goal is to compute the sum of distributed inputs as in the $\Sigma$-QMAC, has been a topic of much interest \cite{Rai_Dey, Ramamoorthy_Langberg, Nazer_Gastpar_Compute, Appuswamy1, Appuswamy3}. As a quantum extension of an elemental classical sum-network setting, the $\Sigma$-QMAC brings together classical and quantum information theoretic perspectives. The interplay of quantum theoretic ideas like quantum entanglement, superdense coding and the stabilizer formalism with classical information theoretic ideas such as network coding and interference alignment offers a promising arena from which new insights may emerge.

\subsection{Significance of Key Assumptions}
In the $\Sigma$-QMAC the inputs are prior-free, the desired output is exact, the channel is idealized, and the number of servers can be arbitrarily large. The significance of these assumptions is explained as follows. The prior-free model is desirable for computation problems, because unlike conventional communication problems where independent messages can be separately compressed to their entropy limit to yield uniform data, for computation problems the compression of inputs cannot be taken for granted as it changes the nature of the computation. The prior-free model is also quite robust as the results are not limited to one data distribution or another. Exact computation goes hand-in-hand with the assumption of prior-free inputs, because probabilistic bounds on errors are less meaningful when no particular distribution is assumed on the data. The $\Sigma$-QMAC setting is chosen to be elemental, i.e., the channels through which the quantum systems are conveyed from the servers to Alice are elementary — they are simply assumed to be noise-free.  The assumption of an idealized channel is important to ensure that the capacity of the $\Sigma$-QMAC reflects only the fundamental limitations of the multipartite quantum-entanglements for the chosen task, and not other artifacts that arise out of channel imperfections and not directly from entanglement per se. Furthermore, since we wish to explore multi-partite entanglements among a large number of parties it is important that we explore  $\Sigma$-QMAC settings with arbitrarily large number of entangled servers. These considerations highlight the essential distinctions that separate our work from other interesting research directions pursued, for example  in \cite{sohail2022unified}, \cite{sohail2022computing} and \cite{hayashi2021computation}, that explore $\epsilon$-error sum computation over a QMAC with correlated data streams and noisy quantum channels, albeit with only $2$ servers (transmitters).

In the $\Sigma$-QMAC we allow  arbitrary entanglements across the subsystems of each quantum system, but quite importantly, we do \emph{not} allow entanglements across systems. We require that strict separation between quantum systems be preserved throughout. The significance of this assumption is that it allows us to determine whether multiparty entanglement is in fact necessary to achieve the  DSC gain for a given $\Sigma$-QMAC setting. For example, suppose we wish to determine the DSC gain achievable in a $\Sigma$-QMAC with only bipartite entanglements. To this end we allow every pair of servers to share unlimited amount of bipartite entanglements, but no multiparty entanglements are initially provided to the servers. Without the separate processing constraint it may still be possible for the $\Sigma$-QMAC capacity to benefit from multiparty entanglement, e.g., due to implicit or explicit fusion \cite{BellBellGHZ} among bipartite quantum systems through joint processing and measurements. The DSC gain thus achieved could not be categorically attributed to only bipartite entanglement. With the separate processing constraint on the other hand, the resulting DSC gain can \emph{only} be attributed to bipartite entanglements. In fact, it is crucial for our motivation of understanding fundamental limits of multiparty entanglements that we are able to convincingly determine whether the DSC gains necessarily require  multiparty entanglement, and furthermore to be able to  distinguish between DSC gains possible with $\beta$-party entanglement from those possible with $\beta'$-party entanglements, for $\beta\neq \beta'$.

In the $\Sigma$-QMAC no prior entanglement is allowed between the servers and the receiver Alice. Recall that entanglement between transmitter and receiver is required in the original setting where superdense coding is introduced \cite{Superdense}. This essential distinction also separates our work from prior efforts to classify entanglements according to their utility for DSC gains in \cite{DSC4}, where the tasks chosen were simple communication tasks and prior entanglements between transmitters and receivers were allowed. In the $\Sigma$-QMAC setting since there is no  prior entanglement between the servers and Alice, it follows from the Holevo bound that no DSC gain is possible for the direct communication task where each server wants to send an independent message to Alice. Thus, computation is essential to our setting which provides a richer space to explore multipartite entanglements.

In the $\Sigma$-QMAC the data streams may be replicated across multiple servers. The significance of this assumption is explained as follows. Without data-replication, the summation $(\Sigma)$  is a \emph{total} function of the computing parties' inputs, but with replication it is only a \emph{partial} function. 
Arbitrarily large (e.g., exponential in the size of inputs)  gains have been established for the exact computation of certain \emph{partial} functions \cite{QFingerprinting, QSM} in the SMP model, i.e., without batch processing. 
The QMAC model, with batch processing, also allows arbitrarily large DSC gains for certain partial functions (see Appendix \ref{app:largeDSC}). On the other hand, the largest observed DSC gain  for exact computation of \emph{total} functions thus far is only $2$ in the SMP setting \cite{PSQM} with or without batch processing \cite{PSQM,DSC4, song_multiple_server_PIR, Allaix_N_sum_box}. To the best of our knowledge, DSC gains larger than $2$ for exact computations of total functions, while  unlikely, have not been formally  ruled out. Specifically for our purpose, it is interesting that both directions are open for exact linear computations, i.e., there are no known instances of linear computations over the QMAC that achieve DSC gain larger than $2$, nor is it known that gains larger than $2$ are impossible in such settings.

\section{Overview of Contribution}\label{sec:contribution}
The ability to  precisely quantify \emph{useful} multiparty entanglements via  DSC gains over a QMAC boils down to following three requirements that must be simultaneously satisfied. 1) the desired computation must represent a natural task, 2) the exact capacity must be tractable, and 3) the capacity must in general require genuine multiparty entanglements, e.g., bipartite entanglements must not be sufficient to achieve the capacity in general. We show that the $\Sigma$-QMAC problem formulation satisfies all $3$ criteria. Since sum-computation with distributed data is obviously a natural task, what remains is to show that the capacity is tractable and sensitive to genuine multiparty entanglement.

The tractability of capacity is established by an explicit capacity characterization in Theorem \ref{thm:main}. The key ideas that make the capacity tractable include the capacity of sum-networks from network coding \cite{Rai_Dey, Ramamoorthy_Langberg, Nazer_Gastpar_Compute, Appuswamy1, Appuswamy3}, the $N$-sum box abstraction  \cite{Allaix_N_sum_box}, dual GRS codes \cite{macwilliams1977theory}, stabilizer based CSS code constructions \cite{Calderbank_Shor_CSS_code, Steane_CSS_code}, and quantum-information theoretic converse arguments \cite{holevo1973bounds, massar2015hyperdense,pawlowski2009information}. The capacity depends on both the data replication and entanglement distribution maps. For example, Corollary \ref{cor:symmetric} shows that in the symmetric setting with $K=\scalebox{0.8}{$\binom{S}{\alpha}$}$ data-streams, each replicated among a distinct $\alpha$-subset\footnote{An $\alpha$-subset of $[S]$ is a subset of $\{1,2,\cdots,S\}$ that has cardinality $\alpha$.} of $[S]$, and $T=\scalebox{0.8}{$\binom{S}{\beta}$}$  quantum systems, each distributed among a distinct $\beta$-subset of $[S]$, the capacity of the $\Sigma$-QMAC is $\frac{1}{\beta T}\sum_{\gamma = (\alpha+\beta-S)^+}^{\min(\alpha,\beta)} \min(\beta,2\gamma) \cdot \binom{\alpha}{\gamma}\cdot \binom{S-\alpha}{\beta-\gamma}$. 

For the motivating example in Fig. \ref{fig:abbccad}, let us explicitly state the capacity results for various entanglement distribution maps. It will be useful to adopt a more intuitive notation by denoting the servers $\mathcal{S}_1,\mathcal{S}_2,\mathcal{S}_3,\mathcal{S}_4$ as $\mathcal{S}_{\sf ab}, \mathcal{S}_{\sf ac}, \mathcal{S}_{\sf bc},\mathcal{S}_{\sf d}$, respectively. The entanglement as shown in Fig. \ref{fig:abbccad} can then be described as $(\{\mathcal{S}_{\sf ab},\mathcal{S}_{\sf ac}\},\{\mathcal{S}_{\sf ac},\mathcal{S}_{\sf bc}, \mathcal{S}_{\sf d}\})$, reflecting the fact that there are $2$ quantum systems, $\mathcal{Q}_1,\mathcal{Q}_2$, such that entangled subsystems of $\mathcal{Q}_1$ are distributed to the servers $\mathcal{S}_{\sf ab},\mathcal{S}_{\sf ac}$, and entangled subsystems of $\mathcal{Q}_2$ are distributed to the servers $\mathcal{S}_{\sf ac},\mathcal{S}_{\sf bc},\mathcal{S}_{\sf d}$. The capacity for this, and various other entanglement distribution maps, is listed in Table \ref{tab:abacbcd}. 

\begin{table}[H]
\center
\begin{tabular}{c|c}
    Entanglement distribution map & Capacity \\\hline
    $(\{\mathcal{S}_{\sf ab},\mathcal{S}_{\sf ac},\mathcal{S}_{\sf bc},\mathcal{S}_{\sf d}\})$ & $4/5$\\\hline
    $(\{  \mathcal{S}_{\sf ab},\mathcal{S}_{\sf ac},\mathcal{S}_{\sf bc} \}, \{ \mathcal{S}_{\sf ab},\mathcal{S}_{\sf ac},\mathcal{S}_{\sf d} \}, \{ \mathcal{S}_{\sf ab},\mathcal{S}_{\sf bc},\mathcal{S}_{\sf d} \}, \{ \mathcal{S}_{\sf ac},\mathcal{S}_{\sf bc},\mathcal{S}_{\sf d} \})$& $3/4$\\\hline
    $(\{\mathcal{S}_{\sf ab},\mathcal{S}_{\sf ac}\}, \{\mathcal{S}_{\sf ab},\mathcal{S}_{\sf d}\}, \{\mathcal{S}_{\sf ac},\mathcal{S}_{\sf d}\}, \{\mathcal{S}_{\sf bc},\mathcal{S}_{\sf d}\})$& $3/4$\\\hline
    $(\{\mathcal{S}_{\sf ab},\mathcal{S}_{\sf ac}\},\{\mathcal{S}_{\sf ac},\mathcal{S}_{\sf bc}, \mathcal{S}_{\sf d}\})$ & $2/3$ \\\hline
    $(\{\mathcal{S}_{\sf ab},\mathcal{S}_{\sf ac}\},\{\mathcal{S}_{\sf ac},\mathcal{S}_{\sf bc}\},\{\mathcal{S}_{\sf ac},\mathcal{S}_{\sf d}\},\{\mathcal{S}_{\sf bc},\mathcal{S}_{\sf d}\})$ & $2/3$ \\\hline
    $(\{\mathcal{S}_{\sf ab}\},\{\mathcal{S}_{\sf ac},\mathcal{S}_{\sf bc},\mathcal{S}_{\sf d}\})$ & $2/3$\\\hline
	$(\{\mathcal{S}_{\sf ab}\},\{\mathcal{S}_{\sf ac},\mathcal{S}_{\sf bc}\},\{\mathcal{S}_{\sf ac},\mathcal{S}_{\sf d}\},\{\mathcal{S}_{\sf bc},\mathcal{S}_{\sf d}\})$ & $2/3$\\\hline
    $(\{\mathcal{S}_{\sf ab},\mathcal{S}_{\sf ac},\mathcal{S}_{\sf bc}\},\{\mathcal{S}_{\sf d}\})$ & $1/2$ \\\hline
    $(\{\mathcal{S}_{\sf ab},\mathcal{S}_{\sf ac}\},\{\mathcal{S}_{\sf ab},\mathcal{S}_{\sf bc}\},\{\mathcal{S}_{\sf ac},\mathcal{S}_{\sf bc}\},\{\mathcal{S}_{\sf d}\})$ & $1/2$ \\\hline
    $(\{\mathcal{S}_{\sf ab},\mathcal{S}_{\sf ac}\},\{\mathcal{S}_{\sf bc},\mathcal{S}_{\sf d}\})$ & $1/2$ \\\hline
       $(\{\mathcal{S}_{\sf ab}\},\{\mathcal{S}_{\sf ac}\},\{\mathcal{S}_{\sf bc}\},\{\mathcal{S}_{\sf d}\})$ & $2/5$\\\hline
\end{tabular}
\vspace{0.2in}
\caption{Capacity of the $\Sigma$-QMAC in Fig. \ref{fig:abbccad} for various entanglement distribution maps} \label{tab:abacbcd}
\end{table}

Note that the last row of the table corresponds to the \emph{unentangled} case, i.e., with no entanglements allowed between servers. The unentangled capacity of the $\Sigma$-QMAC for a given data replication map $\mathcal{W}$ is the same as the classical capacity for the corresponding setting. In this case, the unentangled capacity is $2/5$ computations/qudit, meaning that without quantum entanglement the fundamental limit dictates that each instance of the desired sum $\sf{A+B+C+D}$ requires $5/2$ qudits ($5/2$ dits in the classical case) to be sent to Alice. The first row represents the opposite extreme, the fully entangled case that allows all $4$ servers to be entangled, and we note that the capacity in this case is $4/5$ computations/qudit.  For each specified entanglement distribution map the DSC gain is the ratio of the capacity for that case to the unentangled capacity. The DSC gain for the fully entangled setting is called the maximal distributed superdense coding gain. In this case, the maximal DSC gain is $2$.

In terms of the main motivation of quantifying useful genuine multiparty entanglements via DSC gains, the capacity analysis allows us to draw various insights. Some of these are highlighted below along with pointers to relevant results in the paper.
\begin{enumerate}
\item The maximal DSC gain of $\Sigma$-QMAC\footnote{The possibility of DSC gains larger than $2$ for \emph{vector} linear computations over the QMAC remains open.} is $2$. The maximal DSC gain of $2$ is  achievable  in the fully entangled $\Sigma$-QMAC if and only if the unentangled capacity  is not more than $1/2$ (computations/qudit).  When the unentangled capacity is not less than $1/2$, the fully entangled $\Sigma$-QMAC has DSC gain exactly equal to the reciprocal of the unentangled capacity. These observations are implied by Corollary \ref{cor:cqcc} of Section \ref{sec:results}.
\item Bipartite ($2$-party) entanglement is in general insufficient to achieve the maximal DSC gain in the $\Sigma$-QMAC, even if unlimited bipartite entanglement is made available to every pair of transmitters. In other words, multiparty entanglement is \emph{necessary} in general. This  can be seen from Table \ref{tab:abacbcd} --- the  capacity is $2/3$ if only bipartite  entanglement is allowed, whereas the capacity is $4/5$ with entanglement allowed across all four servers. The necessity of multiparty entanglements is  also evident from Corollary \ref{cor:symmetric} in Section \ref{sec:results}.
\item If each data-stream is only available to a unique server, then bipartite entanglement suffices  to achieve the maximal DSC gain. This observation is based on Corollary \ref{cor:unique_data_stream} in Section \ref{sec:results}.
\item The symmetric setting in Corollary \ref{cor:symmetric}  reveals that both extremes of too much data replication  and too little data replication  require relatively little entanglement  to achieve their maximal DSC gain,  rather the intermediate regimes of data replication are the ones that require the most entanglement. See discussion in Section \ref{sec:ex_min_beta}.
\item The minimal entanglement-size $\beta$ such that $\beta$-party entanglement is necessary to achieve a desired (feasible) DSC gain value in a $\Sigma$-QMAC can be determined from the capacity characterizations. For example, from  Table \ref{tab:abacbcd} we note that  $4$-party entanglement is necessary to achieve maximal DSC gain. The best DSC gain with $3$-party entanglements is only $(3/4)/(2/5)=15/8$ which can also be achieved with $2$-party entanglements
\item $3$-party entanglement is never \emph{necessary} to achieve the capacity of the $\Sigma$-QMAC. Any $3$-party entanglement can be replaced by $2$-party entanglements with the capacity unchanged. This observation is based on Corollary \ref{cor:tripartite} of Section \ref{sec:results}. An example is presented in Section \ref{sec:ex_3_party}. 
\item For every $S\neq 3$, there is a $\Sigma$-QMAC setting with $S$ servers where  $S$-party entanglement is \emph{necessary} to achieve the maximal DSC gain. In such settings, even with unlimited $S-1$ party entanglements among all $(S-1)$-subsets of servers, the DSC gain is strictly smaller than that with $S$-party entanglement. This observation is based on Corollary \ref{cor:S_partite_necessary} of Section \ref{sec:results}. Related discussion is  provided in Section \ref{sec:S_partite_necessary}.
\item Entanglements restricted to stabilizer states, along with Pauli operations ($X$ and $Z$ gates) at the servers (i.e., coding via the $N$-sum box abstraction) suffice to achieve the capacity of the $\Sigma$-QMAC. Thus, while the capacity formulation does allow non-stabilizer states and general unitary operations, neither of those can improve the DSC gains in the $\Sigma$-QMAC. This observation emerges from the proof of achievability of Theorem \ref{thm:main}.
\end{enumerate}

\noindent {\it Notation:} $\mathbb{N}$ denotes the set of positive integers. $\mathbb{Z}^+ = \{0\}\cup \mathbb{N}$. For $n\in \mathbb{N}$, $[n]$ denotes the set $\{1,2,\cdots, n\}$. For $n_1, n_2\in \mathbb{N}$, $[n_1:n_2]$ denotes the set $\{n_1,n_1+1,\cdots,n_2\}$ if $n_1\leq n_2$ and $\emptyset$ otherwise. For a set $\mathcal{X}$, define $\mathcal{X}^n \triangleq \mathcal{X} \times \mathcal{X} \times \cdots \times \mathcal{X}$ as the $n$-fold Cartesian product. Define compact notations $A^{[n]} \triangleq (A^{(1)}, A^{(2)},\cdots, A^{(n)})$ and $A_{[n]} \triangleq (A_1, A_2, \cdots, A_n)$.    $\mathbb{F}_d$ denotes the finite field with $d=p^r$ a power of a prime. $\mathbb{R}$ denotes the set of real numbers. $\mathbb{R}_+$ denotes the set of non-negative real numbers. $\mathbb{C}$ denotes the set of complex numbers. For a field $\mathbb{F}$, $\mathbb{F}^{a\times b}$ denotes the set of $a\times b$ matrices with elements in $\mathbb{F}$. $\tr(M)$ denotes the trace of a matrix $M$. For a matrix $M$ with elements in $\mathbb{C}$, $M^\dagger$ denotes its conjugate transpose. ${\bf I}_{a}$ denotes the $a\times a$ identity matrix. ${\bf 0}_{a\times b}$ denotes the zero matrix with size $a\times b$. ${\sf Pr}(E)$ denotes the probability of an event $E$. ${\sf Pr}(E_1|E_2)$ denotes the conditional probability of $E_1$ given $E_2$. $(x)^{+} \triangleq \max(x, 0)$. For $m,n \in\mathbb{Z}^+, m\leq n$, $\binom{n}{m} \triangleq \frac{n!}{m!(n-m)!}$ denotes the binomial coefficient. For a set $\mathcal{N}$, the set of its cardinality-$m$ sub-sets is denoted as $\binom{\mathcal{N}}{m} \triangleq \{\mathcal{A} \subset \mathcal{N} \mid |\mathcal{A}| = m \}$ if $|\mathcal{N}| \geq m$. The notation $2^{\mathcal{N}}$ denotes the power set of $\mathcal{N}$. The notation $f:\mathcal{A} \mapsto \mathcal{B}$ denotes a map $f$ from  $\mathcal{A}$ to  $\mathcal{B}$. If $f$ is a bijection from $\mathcal{A}$ to $\mathcal{B}$, we write $f: \mathcal{A} \biject \mathcal{B}$ and denote the inverse of $f$ as $f^{-1}$. The dimension of a quantum system $\mathcal{Q}$ is denoted as $|\mathcal{Q}|$.

\section{Problem Formulation}\label{sec:probform}

\subsection{$\Sigma$-QMAC}\label{sec:qmacdef}
The $\Sigma$-QMAC problem is specified by a $6$-tuple $\big(\mathbb{F}_d,S,K,T, \mathcal{W},\mathcal{E}\big)$. $\mathbb{F}_d$ is a finite field of order $d$ with $d=p^r$ being a power of a prime. $S$ is the number of servers.  $K$ is the number of independent classical data-streams, denoted as ${\sf W}_1,{\sf W}_2,\cdots, {\sf W}_K$. The $k^{th}$ data stream, ${\sf{W}}_k$, is comprised of symbols ${\sf W}_k^{(\ell)} \in\mathbb{F}_d, \ell\in\mathbb{N}$. $T$ is the number of independent quantum resources, denoted as $\mathcal{Q}_1,\mathcal{Q}_2,\cdots,\mathcal{Q}_T$. The \emph{data replication map} is a mapping $\mathcal{W}:[K]\to 2^{[S]}$ that identifies $\mathcal{W}(k) \subset [S]$ as the subset of servers where ${\sf W}_k$ is available. The \emph{entanglement distribution map} is a mapping $\mathcal{E}:[T]\to 2^{[S]}$ that identifies $\mathcal{E}(t) \subset [S]$ as the subset of servers among which the quantum system $\mathcal{Q}_t$ is distributed. Such a subset of servers is referred to as a \emph{clique} in this paper. The quantum system $\mathcal{Q}_t$ is partitioned into entangled subsystems $\mathcal{Q}_{t,s}, s\in\mathcal{E}(t)$ such that Server $s$ receives the quantum subsystem $\mathcal{Q}_{t,s}$ from the quantum system $\mathcal{Q}_t=(\mathcal{Q}_{t,s},s\in\mathcal{E}(t))$. 

\subsection{Feasible Quantum Coding Schemes} \label{sec:def_schemes}
A quantum coding scheme is specified by a $6$-tuple 
\begin{align}
  \big(L, ~((\delta_{t,s})_{s\in\mathcal{E}(t)})_{t\in[T]}, ~\rho_{[T]},~ ((\Phi_{t,s})_{s\in\mathcal{E}(t)})_{t\in[T]},~ (\{M_{t,y}\}_{y\in \mathcal{Y}_t})_{t\in [T]},~ \Psi \big).
\end{align}
$L\in \mathbb{N}$ is the batch size, which is the number of sums to be computed by the coding scheme, i.e., the coding scheme allows Alice to compute ${\sf W}_\Sigma(\ell)$ for all $\ell\in[L]$. 
For $k\in [K]$, denote the first $L$ instances of the $k^{th}$ data stream as 
${\sf W}_k^{[L]} = \big( {\sf W}_k^{(1)},{\sf W}_k^{(2)},\cdots,{\sf W}_k^{(L)} \big) \in \mathbb{F}_d^{L}$, and the desired computation at Alice as 
${\sf W}_{\Sigma}^{[L]} = \big({\sf W}_{\Sigma}^{(1)}, $ ${\sf W}_{\Sigma}^{(2)},\cdots,{\sf W}_{\Sigma}^{(L)} \big) \in \mathbb{F}_d^{L}$, where ${\sf W}_{\Sigma}(\ell) \triangleq \sum_{k=1}^K{\sf W}_k(\ell), \forall \ell \in [L]$.
For $t\in[T]$ and $s\in\mathcal{E}(t)$, $\delta_{t,s}\in \mathbb{Z}^+$ specifies the dimension of the quantum subsystem $\mathcal{Q}_{t,s}$, i.e., $|\mathcal{Q}_{t,s}| = \delta_{t,s}$. 
For $t\in[T]$, say the $t^{th}$ clique is $\mathcal{E}(t) = \{s_1,s_2,\cdots,s_{|\mathcal{E}(t)|}\}$. The quantum system $\mathcal{Q}_t = \mathcal{Q}_{t,s_1}\mathcal{Q}_{t,s_2}\cdots\mathcal{Q}_{t,s_{|\mathcal{E}(t)|}}$ is prepared in the initial state $\rho_t \in \mathbb{C}^{|\mathcal{Q}_t| \times |\mathcal{Q}_t|}$.
$\mathcal{Q}_1, \mathcal{Q}_2,\cdots,\mathcal{Q}_T$ are unentangled with each other. Without loss of generality, we assume that the initial state of the composite system is a pure state, and thus it can be written as $\rho = \rho_1\otimes \rho_2 \otimes \cdots \otimes \rho_T$. 
For $t\in[T]$, Server $s\in\mathcal{E}(t)$ applies a unitary operator $U_{t,s} = \Phi_{t,s}({\sf W}_k^{[L]}, k: s\in \mathcal{W}(k))$ to $\mathcal{Q}_{t,s}$. Thus, the resulting state of $\mathcal{Q}_t$ is $\rho'_t = U_t \rho_t ~U_t^\dagger$, where $U_t \triangleq U_{t,s_1}\otimes U_{t,s_2}\otimes \cdots \otimes U_{t,s_{|\mathcal{E}(t)|}}$. 
Note that since $\mathcal{Q}_1,\mathcal{Q}_2,\cdots,\mathcal{Q}_T$ are unentangled initially, and separate operations are done for $\mathcal{Q}_1,\mathcal{Q}_2,\cdots,\mathcal{Q}_T$, they remain unentangled after the servers apply the operations. All subsystems are then sent to Alice, who performs separate quantum measurements (POVM) on each of the $T$ quantum systems. Specifically, for $t\in[T]$, the set of operators for the measurement of $\mathcal{Q}_t$ is specified as $\{M_{t,y}\}_{y\in \mathcal{Y}_t}$ by the coding scheme. The output of the measurement is denoted as $Y_t$, which is a random variable with realizations in $\mathcal{Y}_t$. 
 Finally, the function $\Psi: \mathcal{Y}_1\times \mathcal{Y}_2 \cdots\times \mathcal{Y}_T \to \mathbb{F}_d^{L \times 1}$ maps the outputs of the measurements $ Y_{[T]} = (Y_1,Y_2,\cdots,Y_T)$ to the desired computation (sum), i.e., ${\sf{W}}_\Sigma^{[L]}=\Psi(Y_1,\cdots,Y_T)$.
Any feasible coding scheme must work for all $d^{KL}$ realizations of $\big( {\sf W}_1^{[L]},{\sf W}_2^{[L]},\cdots, {\sf W}_K^{[L]} \big)$. Let $\mathfrak{C}$ denote the set of such coding schemes.

\subsection{Feasible Region and Capacity}
For the $\Sigma$-QMAC$\big(\mathbb{F}_d,S,K,T,\mathcal{W},\mathcal{E}\big)$, the  download-cost per computation (qudits/dit) tuple,
\begin{align}
	{\bm \Delta} = (\Delta_{t,s})_{t\in[T],s\in\mathcal{E}(t)} \in\mathbb{R}_+^{\Gamma}, ~~ \Gamma \triangleq \sum_{t\in[T]} |\mathcal{E}(t)|,
\end{align}
is said to be feasible, if there exists a coding scheme 
\begin{align}
	\big(L,~ ((\delta_{t,s})_{s\in\mathcal{E}(t)})_{t\in[T]},~ \rho_{[T]}, ~((\Phi_{t,s})_{s\in\mathcal{E}(t)})_{t\in[T]},~ (\{M_{t,y}\}_{y\in \mathcal{Y}_t})_{t\in [T]},~ \Psi \big) \in \mathfrak{C} \notag
\end{align}
such that 
\begin{align}
  \Delta_{t,s} \geq \frac{\log_d |\mathcal{Q}_{t,s}|}{L} = \frac{\log_d \delta_{t,s}}{L}, && \forall t\in[T], s\in \mathcal{E}(t).
\end{align}
Define $\mathcal{D}$ as the closure of the set of all feasible download-cost tuples ${\bm \Delta}$ so that any ${\bm \Delta}$ inside $\mathcal{D}$ is feasible, and any ${\bm \Delta}$ outside $\mathcal{D}$ is not feasible. 
In terms of computation rates (dits of computation/qudit of download), a rate $R$ is feasible if there exists a coding scheme in $\mathfrak{C}$
such that
\begin{align}
	R \leq \frac{L}{\sum_{t\in[T],s\in\mathcal{E}(t)}\log_d |\mathcal{Q}_{t,s}|} = \frac{L}{\sum_{t\in[T],s\in\mathcal{E}(t)}\log_d \delta_{t,s}}.
\end{align}
Define 
\begin{align}
  C \triangleq \sup_{\mathfrak{C}} R
\end{align}
as the computation capacity. Note that a capacity $C$ characterization is implied by a characterization of $\mathcal{D}$ because $C = (\min_{{\Delta}\in \mathcal{D}}\sum_{t\in[T],s\in\mathcal{E}(t)} \Delta_{t,s})^{-1}$.
Since $S,K,T$ can be inferred from $\mathcal{W},\mathcal{E}$, a $\Sigma$-QMAC problem is fully specified by the parameters $(\mathbb{F}_d,\mathcal{W},\mathcal{E})$. As our first result (Theorem \ref{thm:main}) will show, the capacity is independent of $\mathbb{F}_d$ (which reflects the merit of using qudit to measure the cost). Therefore the capacity $C(\mathcal{W},\mathcal{E})$ is only a function of $(\mathcal{W},\mathcal{E})$. While comparing capacities of problems that have the same data replication map but different entanglement distribution maps, the data replication map $\mathcal{W}$ may be occasionally omitted for brevity when it is clear from the context.

\subsection{Fully-entangled, Fully-unentangled and Fully-$\beta$-party-entangled Capacities}
Symmetric entanglement distribution maps can be especially insightful because the symmetry facilitates compact capacity descriptions that are easier to compare. To prepare for discussions of symmetric entanglements,  let us define the \emph{fully-$\beta$-party-entangled} setting as the entanglement distribution map comprised of all  cliques of size $\beta$. Note that this allows genuine multiparty entanglement to be established among every subset of servers containing no more than $\beta$ servers. Formally, \emph{fully-$\beta$-party-entangled} setting\footnote{We may shorten `fully-$\beta$-party-entangled'  to simply `$\beta$-party-entangled'  in subsequent discussions when the context is obvious.} refers to the \emph{bijective} entanglement distribution map,
\begin{align}
\mathcal{E}^{(\beta)} \colon [\binom{S}{\beta}] \biject \binom{[S]}{\beta}.
\end{align}
For example, with $S=4$ servers, fully-$2$-party-entangled setting means that $\binom{4}{2}=6$ separate quantum systems are available, each comprised of $2$ entangled subsystems that are distributed among a distinct pair ($2$-clique) of servers. $\mathcal{E} ^{(2)}$ is a bijection from $\{1,2,\cdots,6\}$ to the set of all $2$-subsets of $[4]$, i.e., $\{\{1,2\},\{2,3\},\{3,4\},\{4,1\},\{1,3\},\{2,4\}\}$. For example, $\mathcal{E} ^{(2)}(1)=\{1,2\}$, $\mathcal{E} ^{(2)}(2)=\{2,3\}$, $\mathcal{E} ^{(2)}(3)=\{3,4\}$, $\mathcal{E} ^{(2)}(4)=\{4,1\}$, $\mathcal{E} ^{(2)}(5)=\{1,3\}$, $\mathcal{E} ^{(2)}(6)=\{2,4\})$. Given a data replication map $\mathcal{W}$, the fully-$\beta$-party-entangled capacity is correspondingly defined as,
\begin{align}
	C^{(\beta)}(\mathcal{W}) \triangleq C(\mathcal{W}, \mathcal{E}^{(\beta)}).
\end{align} Extreme cases of fully-$\beta$-party-entangled capacity include the $\beta=S$ setting, known simply as the fully-entangled capacity $C^{\fullent}(\mathcal{W}) = C^{(S)}(\mathcal{W})$, and the $\beta=1$ setting, known simply as the \emph{fully-unentangled} capacity $C^{\unent}(\mathcal{W}) = C^{(1)}(\mathcal{W})$.

\subsection{Distributed Superdense Coding (DSC) Gain}
Given data replication map $\mathcal{W}$ and entanglement distribution map $\mathcal{E}$, the distributed super dense coding (DSC) gain is defined as the ratio 
\begin{align}
	C(\mathcal{W}, \mathcal{E})/C^{\unent}(\mathcal{W}),
\end{align}
which indicates the multiplicative gain, compared to the fully-unentangled capacity $C^{\unent}(\mathcal{W})$, that is enabled by quantum entanglement subject to the entanglement distribution map $\mathcal{E}$. The \emph{maximal} DSC gain for data replication map $\mathcal{W}$ is defined as $C^{\fullent}(\mathcal{W})/C^{\unent}(\mathcal{W})$, which is the ratio of the fully-entangled capacity to the fully-unentangled capacity.

\section{Results} \label{sec:results}
The exact capacity of the $\Sigma$-QMAC is fully characterized in the following theorem.
\begin{theorem}[$\Sigma$-QMAC] \label{thm:main}
The capacity of the $\Sigma$-QMAC $\big(\mathbb{F}_d, S,K,T, \mathcal{W}, \mathcal{E} \big)$,  is 
  \begin{align}
C(\mathcal{W},\mathcal{E}) = \left(\min_{{\bm \Delta}\in \mathcal{D}(\mathcal{W},\mathcal{E})}\sum_{t\in[T],s\in\mathcal{E}(t)} \Delta_{t,s} \right)^{-1}
  \end{align}
with the feasible download-cost region characterized as $(\Gamma \triangleq \sum_{t\in[T]}|\mathcal{E}(t)|)$,
  \begin{align} \label{eq:region_clique}
    \mathcal{D}(\mathcal{W},\mathcal{E})  = \left\{ 
    {\bm \Delta} \in \mathbb{R}_+^{\Gamma} \left|
    \begin{array}{l}
      \sum_{t\in[T]} \min\Big\{\sum_{s\in \mathcal{E}(t)}\Delta_{t,s}, \sum_{s\in \mathcal{E}(t)\cap \mathcal{W}(k)} 2\Delta_{t,s} \Big\} \geq 1, \forall k\in [K]
    \end{array}
    \right.\right\}.
   \end{align}
\end{theorem}
\noindent The proof of Theorem \ref{thm:main} appears in Sections \ref{proof:main_conv} and \ref{proof:main_achievability}. Note that Theorem \ref{thm:main} characterizes the capacity of the $\Sigma$-QMAC in terms of the solution of a linear program that finds the minimum download cost over the feasible region $\mathcal{D}$ that is explicitly characterized. The capacity does not depend on the field $\mathbb{F}_d$. In fact the capacity depends \emph{only} on the data replication and entanglement distribution maps ($\mathcal{W}, \mathcal{E}$) since the remaining parameters $S,K,T$ can be inferred from $\mathcal{W}, \mathcal{E}$.

The remainder of this section  identifies a few interesting corollaries that follow from Theorem \ref{thm:main}. We start with the  specializations for the opposite extremes, fully-entangled and fully-unentangled capacities,  stated as corollaries next.
\begin{corollary}[Fully-entangled] \label{cor:unconstrained}
The fully-entangled capacity for a data replication map $\mathcal{W}$ is,
\begin{align}
	C^{\fullent}(\mathcal{W}) = \left(\min_{(\Delta_1,\cdots, \Delta_S) \in \mathcal{D}^{\fullent}(\mathcal{W})}\sum_{s\in [S]} \Delta_{s} \right)^{-1}, 
\end{align}
where
	\begin{align} \label{eq:region}
    \mathcal{D}^{\fullent}(\mathcal{W}) = \left\{ (\Delta_1,\cdots, \Delta_S) \in \mathbb{R}_+^S \left|
      \sum_{s\in[S]} \Delta_s \geq 1, \sum_{s\in \mathcal{W}(k)} \Delta_s \geq 1/2, \forall k\in [K]
    \right.\right\}.
    \end{align}
\end{corollary}

\begin{corollary}[Fully-unentangled]\label{cor:unentangled}
The fully-unentangled capacity for a data replication map $\mathcal{W}$ is,
\begin{align}
	C^{\unent}(\mathcal{W}) = \left(\min_{(\Delta_1,\cdots, \Delta_S)\in \mathcal{D}^{\unent}(\mathcal{W})}\sum_{s\in[S]} \Delta_{s} \right)^{-1}
\end{align}
where
  \begin{align} \label{eq:region_singleton}
    \mathcal{D}^{\unent}(\mathcal{W}) = \left\{ (\Delta_1,\cdots, \Delta_S) \in \mathbb{R}_+^S ~\left|~  \mathbf
      \sum_{s \in \mathcal{W}(k)} \Delta_s \geq 1, \forall k\in [K]
    \right.\right\}.
  \end{align}
\end{corollary}

From  Corollary \ref{cor:unconstrained} and Corollary \ref{cor:unentangled} we have the following  characterization of the maximal DSC gain for any data replication map $\mathcal{W}$.

\begin{corollary}[Maximal DSC gain]\label{cor:cqcc}
The maximal distributed superdense coding gain for the data replication map $\mathcal{W}$ is,
\begin{align}
C^{\fullent}(\mathcal{W})/C^{\unent}(\mathcal{W}) &= \min\big(2,1/C^{\unent}(\mathcal{W})\big).
\end{align}
\end{corollary}
\noindent The proof of Corollary \ref{cor:cqcc} is relegated to Appendix \ref{proof:cqcc}.

Next we explore a class of $\Sigma$-QMAC settings with symmetric data replication and entanglement distribution maps. A symmetric setting is specified by three parameters, $S, \alpha$ and $\beta$.
The data replication map is fully symmetric, so that for each $\alpha$-subset of $[S]$ there is a unique data stream replicated among this subset of servers. The goal is to characterize explicitly the fully-$\beta$-party-entangled capacity for such data replication maps. The explicit characterization is provided next.
\begin{corollary}[Symmetric] \label{cor:symmetric}
	If $\alpha \in [S]$, $K = \binom{S}{\alpha}$, $T = \binom{S}{\beta}$, and  $\mathcal{W}: [K] \biject \binom{[S]}{\alpha}$ and $\mathcal{E}: [T] \biject \binom{[S]}{\beta}$ are bijective mappings,  then the $\Sigma$-QMAC capacity (denoted as $C_{\alpha}^{(\beta)}$) is 
	\begin{align}
		C_{\alpha}^{(\beta)} &= \frac{1}{\beta T}\sum_{\gamma = (\alpha+\beta-S)^+}^{\min(\alpha,\beta)} \min(\beta,2\gamma) \cdot \binom{\alpha}{\gamma}\cdot \binom{S-\alpha}{\beta-\gamma} \label{eq:symmetric} \\
		& = \frac{2\alpha}{S} - \frac{1}{\beta T} \sum_{\gamma = \max(\alpha + \beta -S, \lceil \beta/2 \rceil)}^{\min (\alpha, \beta)} (2\gamma - \beta) \cdot \binom{\alpha}{\gamma} \cdot \binom{S-\alpha}{\beta - \gamma} \label{eq:symmetric_equiv_1}  \\
	    & = 1- \frac{1}{\beta T}\sum_{\gamma = (\alpha+\beta-S)^+}^{\min (\alpha, \lfloor \beta/2 \rfloor)} (\beta -2\gamma) \cdot \binom{\alpha}{\gamma} \cdot \binom{S-\alpha}{\beta -\gamma} \label{eq:symmetric_equiv_2}
			\end{align}
\end{corollary}
\noindent The proof of Corollary \ref{cor:symmetric} is relegated to Appendix \ref{proof:symmetric}.

The next corollary sheds light on $C^{(2)}$, i.e., the capacity with only bipartite entanglements. Let us first provide the necessary context before presenting the next corollary. Given a $\Sigma$-QMAC problem $\mathcal{P}$ with data replication $\mathcal{W}$, $S$ servers and $K$ data streams, we want to explicitly characterize the $2$-party-entangled capacity $C^{(2)}(\mathcal{W})$. To do so, let us construct another hypothetical $\Sigma$-QMAC problem referred to as $\widetilde{\mathcal{P}}$ with data replication map $\widetilde{\mathcal{W}}$, $\binom{S}{2}$ servers and the same $K$ data streams  as in $\mathcal{P}$.  Now let us specify $\widetilde{\mathcal{W}}$. Each server in $\widetilde{\mathcal{P}}$ is indexed by a $2$-element set $\{i,j\} \subset [S]$.  Server $\mathcal{S}_{\{i,j\}}$ in $\widetilde{\mathcal{P}}$ has the access to the data streams that are available to Servers $\mathcal{S}_i$ or $\mathcal{S}_j$ in $\mathcal{P}$. In other words, a data stream is available to Server $\mathcal{S}_{\{i,j\}}$ in $\widetilde{\mathcal{P}}$ if and only if that data stream is available to either Server $\mathcal{S}_i$ or Server $\mathcal{S}_j$ (or to both) in $\mathcal{P}$. Mathematically, $\widetilde{\mathcal{W}}(k) = \big\{ \{i,j\}\in \binom{[S]}{2} ~\big|~  \big(\{i,j\} \cap \mathcal{W}(k) \big) \not=\emptyset \big\}$ for $k\in [K]$. Now we are ready to present the next corollary.

\begin{corollary}[Fully-$2$-party-entangled capacity] \label{cor:bipartite}
$C^{(2)}(\mathcal{W}) = C^{\unent}(\widetilde{\mathcal{W}})$. In addition, it can be shown that $C^{(2)}(\mathcal{W})$ can always be achieved by a scheme that involves only the $2$-sum protocol.\footnote{Schemes that only apply the $2$-sum protocol are special cases of the quantum coding schemes formulated in Section \ref{sec:def_schemes} when $T=\binom{S}{2}$ and $\mathcal{E}: [T]\biject \binom{[S]}{2}$. In other words, each of the $T$ quantum systems $\mathcal{Q}_1,\mathcal{Q}_2, \cdots,\mathcal{Q}_T$ is available to a unique pair of the $S$ servers.}
\end{corollary}
\noindent The proof of Corollary \ref{cor:bipartite} is relegated to Appendix \ref{proof:2sum}. In plain words, Corollary \ref{cor:bipartite} states that the $2$-party-entangled capacity for the data replication map $\mathcal{W}$ is equal to the fully-unentangled capacity for the data replication map $\widetilde{\mathcal{W}}$, comprised of a new set of servers that are obtained by merging pairs of original servers.

\begin{corollary}[Disjoint data] \label{cor:unique_data_stream}
Given data replication map $\mathcal{W}$ with $S$ servers and $K$ data streams, if $S\geq 2$ and each data-stream is available to only one server, i.e., $|\mathcal{W}(k)| = 1$ for all $k\in[K]$, then $C^{\fullent}(\mathcal{W})=C^{(2)}(\mathcal{W}) = 2/S$. 
\end{corollary}
In other words, if no data stream is replicated across more than one server, then genuine multiparty entanglement (between more than $2$ parties) is not needed, i.e., the fully-entangled capacity is equal to the $2$-party-entangled capacity. Together with Corollary \ref{cor:bipartite}, this implies that $2$-sum protocol based schemes suffice to achieve the fully-entangled capacity in this case.

\begin{proof}

Since each server can locally add the data streams and regard the sum as one data stream, it suffices to consider $K=S$ data streams such that each server has a unique data stream. The reduced setting belongs to the symmetric settings specified in Corollary \ref{cor:symmetric} with $S$ servers and $\alpha=1$. It can then be verified by  \eqref{eq:symmetric_equiv_1} that $C_1^{(2)} = C_1^{(S)} = 2/S$.
\end{proof}

\begin{corollary}[$3$-party entanglement is unnecessary] \label{cor:tripartite}
Given any $\Sigma$-QMAC problem with data replication map $\mathcal{W}$ and entanglement distribution map $\mathcal{E}$ that identifies a $3$-party clique, say $\mathcal{E}(t) = \{s_1,s_2,s_3\}$ for some $t\in[T]$, consider another entanglement distribution map $\mathcal{E}'$, which is created by first making a copy of $\mathcal{E}$, and then replacing the $3$-party clique $\{s_1,s_2,s_3\}$ with three $2$-party cliques $\{s_1,s_2\}, \{s_1,s_3\}$ and $\{s_2,s_3\}$. Then we always have $C(\mathcal{W},\mathcal{E}) = C(\mathcal{W},\mathcal{E}')$.
\end{corollary}

\noindent The proof of Corollary \ref{cor:tripartite} is relegated to Appendix \ref{proof:tripartite}. It shows that any $3$-party entanglement can be substituted by $2$-party entanglements established by the same three servers, for the purpose of $\Sigma$-QMAC capacity (not necessarily for other function computations). Note that it immediately follows that $C^{(3)}(\mathcal{W}) = C^{(2)}(\mathcal{W})$ for any data replication map $\mathcal{W}$.

\begin{corollary}[Necessity of multiparty entanglements]  \label{cor:S_partite_necessary}
For any $S\geq 2$ and $S\not=3$, there exists a data replication map $\mathcal{W}$ with $S$ servers for which $C^{(S)}(\mathcal{W}) > C^{(S-1)}(\mathcal{W})$.	
\end{corollary}
\noindent The proof of Corollary \ref{cor:S_partite_necessary} is relegated to Appendix \ref{proof:asymmetry}. The corollary states that given the number of servers $S$, if $S\geq 2$ and $S\not=3$, there is a data replication map $\mathcal{W}$ with $S$ servers for which the fully-entangled capacity $C^{\fullent}(\mathcal{W})$ is strictly greater than the $(S-1)$-party-entangled capacity $C^{(S-1)}(\mathcal{W})$. In other words, in order to attain the maximal DSC gain for the $\Sigma$-QMAC with this data replication map, one must allow the entanglement to be established among all $S$ servers. A concrete example for this corollary can be found in Section \ref{sec:S_partite_necessary}.

\section{Examples}

In this section we present examples to illustrate the results.

\subsection{Example: Achieving Fully-Entangled Capacity $C^{\fullent}$ for the Data Replication Map of Fig. \ref{fig:abbccad}}
Consider the data replication map $\mathcal{W}$ as in Fig. \ref{fig:abbccad}. Let us  sketch the solution for the fully-entangled setting, i.e., $\mathcal{E} = \{\{1,2,3,4\}\}$. There are $4$ data-streams, denoted by ${\sf A}, {\sf B}, {\sf C}$ and ${\sf D}$. For intuitive notation, let us use subscripts ${\sf ab}$, ${\sf ac}$, ${\sf bc}$, ${\sf d}$  to represent $1$, $2$, $3$, $4$, respectively, reflecting the data-streams available at the corresponding servers. For example,  we indicate Server $\mathcal{S}_1$ as $\mathcal{S}_{\sf ab}$, making it explicit that this server has data-streams ${\sf A},{\sf B}$. Similarly, Server $\mathcal{S}_2$ has data-streams ${\sf A}, {\sf C}$, and is thus denoted as $\mathcal{S}_{\sf ac}$. The other two servers are $\mathcal{S}_{\sf bc}$ and $\mathcal{S}_{\sf d}$ accordingly. The corresponding normalized download costs are denoted as $\Delta_{\sf ab}, \Delta_{\sf ac}, \Delta_{\sf bc}$ and $\Delta_{\sf d}$.  With these notations,   the feasible region in Corollary \ref{cor:unconstrained} is,
  \begin{align} \label{eq:region_ex}
    \mathcal{D}^{\fullent} = \left\{ (\Delta_{\sf ab},\Delta_{\sf ac},\Delta_{\sf bc},\Delta_{\sf d}) \in \mathbb{R}_+^4 \left|
    \begin{array}{l}
      \Delta_{\sf ab}+\Delta_{\sf ac}+\Delta_{\sf bc}+\Delta_{\sf d} \geq 1, \\
      \Delta_{\sf ab}+\Delta_{\sf ac}\geq 1/2,\\
      \Delta_{\sf ab}+\Delta_{\sf bc}\geq 1/2,\\        
      \Delta_{\sf ac}+\Delta_{\sf bc}\geq 1/2,\\
      \Delta_{\sf d}\geq 1/2.
    \end{array}
    \right.\right\}.
  \end{align}
From the converse standpoint, let us note informally that of the $5$ bounds that appear in \eqref{eq:region_ex}, the first bound says that the normalized \emph{total} download cost is at least $1$ qudit/dit. This is because there is no entanglement between the servers and Alice, so it follows from the Holevo bound \cite{holevo1973bounds} that one qudit cannot carry more than one dit of information. The remaining four bounds are typical cut-set arguments, by separating the parties (servers and Alice) into two groups such that the servers that know one of the data streams are collectively regarded as the transmitter, while the other servers join Alice as the receiver. Since now entanglement can be established between the transmitter and the receiver, it follows from (e.g., \cite{holevo1998capacity, DSC4}) that the one qudit can carry at most $2$ dits of information, yielding the factor $1/2$ on the RHS of the other four bounds.

Minimizing $\Delta_{\sf ab}+\Delta_{\sf ac}+\Delta_{\sf bc}+\Delta_{\sf d}$ over $\mathcal{D}$ leads to a linear program with optimal value $5/4$, thus establishing the fully-entangled capacity for this example as $C^{\fullent}=4/5$. To show the achievability of $4/5$, we specify a coding scheme that allows Alice to recover $L=4$ instances of  the desired sums, denoted as ${\sf A}_{[4]}+{\sf B}_{[4]}+{\sf C}_{[4]}+{\sf D}_{[4]}$, based on an $(N=5)$-sum box in $\mathbb{F}_d$ so that in the $5$-sum box Server $\mathcal{S}_{\sf ab}$ controls $1$ pair of inputs $x_1,x_6$; Server $\mathcal{S}_{\sf ac}$ controls $1$ pair of inputs $x_2,x_7$; Server $\mathcal{S}_{\sf bc}$ controls $1$ pair of inputs $x_3,x_8$; and Server $\mathcal{S}_{\sf d}$ controls $2$ pairs of inputs $x_4,x_5,x_9,x_{10}$.
Let us denote the transfer matrix corresponding to the $5$-sum box as ${\bf M}$, which can be chosen freely in $\mathbb{F}_d^{5\times 10}$, as long as ${\bf M}{\bf J}_{10} {\bf M}^T = {\bf 0}_{5\times 5}$. Denote the input of the box as ${\bf x}=[x_1,x_2,\cdots, x_{10}]^T$ and the output of the box as ${\bf y} \in \mathbb{F}_d^{5\times 1}$, and thus ${\bf y} = {\bf M}{\bf x}$ can be measured at Alice. Since Server $\mathcal{S}_{\sf ab}$ knows only ${\sf A}$ and ${\sf B}$, $(x_1,x_6)$ must be determined by ${\sf A}_{[4]}$ and ${\sf B}_{[4]}$. Similarly, $(x_2, x_7)$ must be determined by ${\sf A}_{[4]}$ and ${\sf C}_{[4]}$; $(x_3, x_8)$ must be determined by ${\sf B}_{[4]}$ and ${\sf C}_{[4]}$; $(x_4, x_5, x_9, x_{10})$ must be determined by ${\sf D}_{[4]}$. Aside from the use of the $5$-sum box, our scheme uses linear precoding and decoding at the servers and the receiver. Specifically, the mapping from $({\sf A}_{[4]}, {\sf B}_{[4]}, {\sf C}_{[4]}, {\sf D}_{[4]})$ to ${\bf x}$ is linear, and the mapping from ${\bf y}$ to the desired sum ${\sf A}_{[4]}+{\sf B}_{[4]}+{\sf C}_{[4]}+{\sf D}_{[4]}$ is also linear. As an example, the precoding at Server $\mathcal{S}_{\sf ab}$ is represented as $[x_1,x_6]^T = V_{\sf ab}^{\sf a}{\sf A}_{[4]} + V_{\sf ab}^{\sf b}{\sf B}_{[4]}$, where $V_{\sf ab}^{\sf a}, V_{\sf ab}^{\sf b}$ are precoding matrices with appropriate size and elements chosen in $\mathbb{F}_d$. The precoding matrices at the other servers are denoted in a similar fashion. The decoding at Alice is represented as $V_{\dec} {\bf y}$, which must be equal to ${\sf A}_{[4]} + {\sf B}_{[4]}+ {\sf C}_{[4]} + {\sf D}_{[4]}$. 
Fig. \ref{fig:coding_example} illustrates the precoding and decoding operations.
\begin{figure}[h]
\center
\begin{tikzpicture}
  \node[] (SA) at (0,4.5) {\footnotesize ${\sf A}_{[4]}$};
  \node[] (SB) at (2,4.5) {\footnotesize ${\sf B}_{[4]}$};
  \node[] (SC) at (4,4.5) {\footnotesize ${\sf C}_{[4]}$};
  \node[] (SD) at (6,4.5) {\footnotesize ${\sf D}_{[4]}$};
  
  \node[rectangle, draw, fill=white, minimum height = 0.6cm, minimum width =0.7cm] (VABA) at (-0.5,3.2) {\footnotesize $V_{\sf ab}^{\sf a}$};
  \node[rectangle, draw, fill=white, minimum height = 0.6cm, minimum width =0.7cm] (VABB) at (0.3,3.2) {\footnotesize $V_{\sf ab}^{\sf b}$};
  \node[rectangle, draw, fill=white, minimum height = 0.6cm, minimum width =0.7cm] (VACA) at (1.5,3.2) {\footnotesize $V_{\sf ac}^{\sf a}$};
  \node[rectangle, draw, fill=white, minimum height = 0.6cm, minimum width =0.7cm] (VACC) at (2.3,3.2) {\footnotesize $V_{\sf ac}^{\sf c}$};
  \node[rectangle, draw, fill=white, minimum height = 0.6cm, minimum width =0.7cm] (VBCB) at (3.5,3.2) {\footnotesize $V_{\sf bc}^{\sf b}$};
  \node[rectangle, draw, fill=white, minimum height = 0.6cm, minimum width =0.7cm] (VBCC) at (4.3,3.2) {\footnotesize $V_{\sf bc}^{\sf c}$};
  \node[rectangle, draw, fill=white, minimum height = 0.6cm, minimum width =0.7cm] (VDD) at (6,3.2) {\footnotesize $V_{{\sf {\sf d}}}^{\sf d}$};

  \draw [thick] ($(SA.south)+(-0,0)$)--($(VABA.north)+(-0,0)$) node[pos=0.4,left] {};
  \draw [thick] ($(SA.south)+(-0,0)$)--($(VACA.north)+(-0,0)$) node[pos=0.4,left=0.1] {};
  \draw [thick] ($(SB.south)+(-0,0)$)--($(VABB.north)+(-0,0)$) node[pos=0.4,right=0.2] {};
  \draw [thick] ($(SB.south)+(-0,0)$)--($(VBCB.north)+(-0,0)$) node[pos=0.4,left=0.1] {};
  \draw [thick] ($(SC.south)+(-0,0)$)--($(VACC.north)+(-0,0)$) node[pos=0.4,right=0.2]{};
  \draw [thick] ($(SC.south)+(-0,0)$)--($(VBCC.north)+(-0,0)$) node[pos=0.4,right=0] {};
  \draw [thick] ($(SD.south)+(-0,0)$)--($(VDD.north)+(-0,0)$) node[pos=0.4,left] {};

  \node[rectangle, draw, fill=white, minimum width=6cm, minimum height = 1.2cm, align=center] (M) at (3,0.5) {$~$\\[-0.25cm]\scalebox{0.81}{${\bf y}={\bf M}{\bf x}$}};

  \node (M1) at ($(M.north)+(-2.7,0.3)$) {$\oplus$};
  \node (M2) at ($(M.north)+(-2.1,0.3)$) {$\oplus$};
  \node (M3) at ($(M.north)+(-1.5,0.3)$) {$\oplus$};
  \node (M4) at ($(M.north)+(-0.9,0.3)$) {$\oplus$};
  \node (M5) at ($(M.north)+(-0.3,0.3)$) {$\oplus$};
  \node (M6) at ($(M.north)+(+0.3,0.3)$) {$\oplus$};
  \node (M7) at ($(M.north)+(+0.9,0.3)$) {$\oplus$};
  \node (M8) at ($(M.north)+(+1.5,0.3)$) {$\oplus$};
  \node (M9) at ($(M.north)+(+2.1,0.3)$) {$\oplus$};
  \node (M10) at ($(M.north)+(+2.7,0.3)$) {$\oplus$};
  
  \draw [] ($(M1.south)+(0,0.1)$)--($(M1.south)+(0,-0.08)$);
  \draw [] ($(M2.south)+(0,0.1)$)--($(M2.south)+(0,-0.08)$);
  \draw [] ($(M3.south)+(0,0.1)$)--($(M3.south)+(0,-0.08)$);
  \draw [] ($(M4.south)+(0,0.1)$)--($(M4.south)+(0,-0.08)$);
  \draw [] ($(M5.south)+(0,0.1)$)--($(M5.south)+(0,-0.08)$);
  \draw [] ($(M6.south)+(0,0.1)$)--($(M6.south)+(0,-0.08)$);
  \draw [] ($(M7.south)+(0,0.1)$)--($(M7.south)+(0,-0.08)$);
  \draw [] ($(M8.south)+(0,0.1)$)--($(M8.south)+(0,-0.08)$);
  \draw [] ($(M9.south)+(0,0.1)$)--($(M9.south)+(0,-0.08)$);
  \draw [] ($(M10.south)+(0,0.1)$)--($(M10.south)+(0,-0.08)$);

  \node[] at ($(M1.south)+(0,-0.2)$) {\tiny $x_1$};
  \node[] at ($(M2.south)+(0,-0.2)$) {\tiny $x_2$};
  \node[] at ($(M3.south)+(0,-0.2)$) {\tiny $x_3$};
  \node[] at ($(M4.south)+(0,-0.2)$) {\tiny $x_4$};
  \node[] at ($(M5.south)+(0,-0.2)$) {\tiny $x_5$};
  \node[] at ($(M6.south)+(0,-0.2)$) {\tiny $x_6$};
  \node[] at ($(M7.south)+(0,-0.2)$) {\tiny $x_7$};
  \node[] at ($(M8.south)+(0,-0.2)$) {\tiny $x_8$};
  \node[] at ($(M9.south)+(0,-0.2)$) {\tiny $x_9$};
  \node[] at ($(M10.south)+(0,-0.2)$) {\tiny $x_{10}$};
  
   \draw [] ($(VABA.south)+(-0.1,0)$)--($(M1.north)+(0,-0.1)$);
   \draw [] ($(VABA.south)+(+0.1,0)$)--($(M6.north)+(0,-0.1)$);
   \draw [] ($(VABB.south)+(-0.1,0)$)--($(M1.north)+(0,-0.1)$);
   \draw [] ($(VABB.south)+(+0.1,0)$)--($(M6.north)+(0,-0.1)$);
   \draw [] ($(VACA.south)+(-0.1,0)$)--($(M2.north)+(0,-0.1)$);
   \draw [] ($(VACA.south)+(+0.1,0)$)--($(M7.north)+(0,-0.1)$);
   \draw [] ($(VACC.south)+(-0.1,0)$)--($(M2.north)+(0,-0.1)$);
   \draw [] ($(VACC.south)+(+0.1,0)$)--($(M7.north)+(0,-0.1)$);
   \draw [] ($(VBCB.south)+(-0.1,0)$)--($(M3.north)+(0,-0.1)$);
   \draw [] ($(VBCB.south)+(+0.1,0)$)--($(M8.north)+(0,-0.1)$);
   \draw [] ($(VBCC.south)+(-0.1,0)$)--($(M3.north)+(0,-0.1)$);
   \draw [] ($(VBCC.south)+(+0.1,0)$)--($(M8.north)+(0,-0.1)$);
   \draw [] ($(VDD.south)+(-0.15,0)$)--($(M4.north)+(0,-0.1)$);
   \draw [] ($(VDD.south)+(-0.05,0)$)--($(M5.north)+(0,-0.1)$);
   \draw [] ($(VDD.south)+(+0.05,0)$)--($(M9.north)+(0,-0.1)$);
   \draw [] ($(VDD.south)+(+0.15,0)$)--($(M10.north)+(0,-0.1)$);
 
 \def\xshift{1.8}
 
 \begin{scope}[shift={(\xshift,0)}]
  \node[rectangle, draw, fill=white] (V) at (3,-0.81) {\footnotesize $V_{\dec}:4\times 5$};

  \draw [] ($(M.south)+(-0.8+\xshift,0)$)--($(V.north)+(-0.8,0)$);
  \draw [] ($(M.south)+(-0.4+\xshift,0)$)--($(V.north)+(-0.4,0)$);
  \draw [] ($(M.south)+(0+\xshift,0)$)--($(V.north)+(0,0)$);
  \draw [] ($(M.south)+(0.4+\xshift,0)$)--($(V.north)+(0.4,0)$);
  \draw [] ($(M.south)+(0.8+\xshift,0)$)--($(V.north)+(0.8,0)$);
 
 \node[] at ($(M.south)+(-1.6,0.2)$) {\tiny (Comm. Cost: $5$ qudits)};
\node[] at ($(M.south)+(-0.8+\xshift,0.2)$) {\tiny $y_1$};
\node[] at ($(M.south)+(-0.4+\xshift,0.2)$) {\tiny $y_2$};
\node[] at ($(M.south)+(0+\xshift,0.2)$) {\tiny $y_3$};
\node[] at ($(M.south)+(0.4+\xshift,0.2)$) {\tiny $y_4$};
\node[] at ($(M.south)+(0.8+\xshift,0.2)$) {\tiny $y_5$};

  \node[] (F) at (3,-1.75) {\footnotesize ${\sf A}_{[4]}+{\sf B}_{[4]}+{\sf C}_{[4]}+{\sf D}_{[4]}$};

  \draw [] ($(V.south)+(-0.3,0)$)--($(F.north)+(-0.3,0)$);
  \draw [] ($(V.south)+(-0.1,0)$)--($(F.north)+(-0.1,0)$);
  \draw [] ($(V.south)+(0.1,0)$)--($(F.north)+(0.1,0)$);
  \draw [] ($(V.south)+(0.3,0)$)--($(F.north)+(0.3,0)$);
\end{scope}
  
\end{tikzpicture}
\caption{Precoding and decoding structure to achieve fully-entangled capacity $C^{\fullent}$ for the data replication map of Fig. \ref{fig:abbccad}.}
\label{fig:coding_example}
\end{figure}

For the general proof of achievability in Section \ref{proof:main_achievability} we  only need to show the existence of appropriate encoders and decoders, so the specific choices will not be explicitly stated. For this example, however, let us explicitly write a  $5\times 10$ transfer matrix ${\bf M}$, the precoding matrices $V_{\sf ab}^{\sf a}, V_{\sf ab}^{\sf b}, \cdots, V_{\sf d}^{\sf d}$ and the decoding matrix $V_{\dec}$, such that they yield the desired relationships.

\noindent {\bf The choice of ${\bf M}$:}
The transfer function of the $5$-sum box is ${\bf M}\in\mathbb{F}_d^{5\times 10}$ specified as,
{\small
\begin{align} \label{eq:box_example}
  {\bf M} = \begin{bmatrix}
    1 & 0 & 0 & 0 & 0 & 0 & 1 & 1 & 0 & 1\\
    0 & 1 & 0 & 0 & 0 & 1 & 0 & 0 & 0 & 1\\
    0 & 0 & 1 & 0 & 0 & 1 & 0 & 0 & 1 & 0\\
    0 & 0 & 0 & 1 & 0 & 0 & 0 & 1 & 1 & 1\\
    0 & 0 & 0 & 0 & 1 & 1 & 1 & 0 & 1 & 0
  \end{bmatrix}.
\end{align}
}It is easily verified that ${\bf M}{\bf J}_{10}{\bf M}^T = {\bf 0}_{5\times 5}$ and thus it is a valid $5$-sum box.\footnote{In fact, any ${\bf M}\in\mathbb{F}_d^{N\times 2N}$ of the form $[{\bf I}_N, {\bf S}_N]$ where ${\bf S}_N = {\bf S}_N^T$ satisfies ${\bf M}{\bf J}_{2N}{\bf M}={\bf 0}_{N\times N}$ and is therefore a valid $N$-sum box \cite{Allaix_N_sum_box}.} 

\noindent {\bf The choice of $V_{\dec}$:}
To the output ${\bf y}\in\mathbb{F}_d^{5\times 1}$, Alice applies a $4\times 5$ decoding matrix $V_{\dec}$ specified as,
{\small
\begin{align}
V_{\dec}&=
\begin{bmatrix}
0&1&0&1&1\\
0&0&0&1&1\\
1&0&0&0&1\\
0&0&1&1&1
\end{bmatrix}.
\end{align}
}

\noindent {\bf The choice of the precoding matrices:} Due to the data dependency, the input vector ${\bf x}\in\mathbb{F}_d^{10\times 1}$ can be written as,
\begin{align}
{\bf x}&=\begin{bmatrix}x_1\\ x_2\\ \vdots\\ x_{10}\end{bmatrix}=\begin{bmatrix}V_{{\sf ab},1}^{\sf a}\\ V_{{\sf ac},1}^{\sf a}\\ {\bf 0}_{1\times 4}\\ {\bf 0}_{2\times 4}\\ V_{{\sf ab},2}^{\sf a}\\V_{{\sf ac},2}^{\sf a}\\{\bf 0}_{1\times 4}\\ {\bf 0}_{2\times 4}\end{bmatrix}
\underbrace{\begin{bmatrix}{\sf A}_1\\{\sf A}_2\\{\sf A}_3\\{\sf A}_4\end{bmatrix}}_{{\sf A}_{[4]}}+\begin{bmatrix}V_{{\sf ab},1}^{\sf b}\\{\bf 0}_{1\times 4}\\ V_{{\sf bc},1}^{\sf b}\\ {\bf 0}_{2\times 4}\\ V_{{\sf ab},2}^{\sf b}\\{\bf 0}_{1\times 4}\\V_{{\sf bc},2}^{\sf b}\\ {\bf 0}_{2\times 4}\end{bmatrix}
\underbrace{\begin{bmatrix}{\sf B}_1\\{\sf B}_2\\{\sf B}_3\\{\sf B}_4\end{bmatrix}}_{{\sf B}_{[4]}}
+\begin{bmatrix}{\bf 0}_{1\times 4}\\ V_{{\sf ac},1}^{\sf c}\\ V_{{\sf bc},1}^{\sf c}\\ {\bf 0}_{2\times 4}\\ {\bf 0}_{1\times 4}\\ V_{{\sf ac},2}^{\sf c}\\V_{{\sf bc},2}^{\sf c}\\ {\bf 0}_{2\times 4}\end{bmatrix}
\underbrace{\begin{bmatrix}{\sf C}_1\\{\sf C}_2\\{\sf C}_3\\{\sf C}_4\end{bmatrix}}_{{\sf C}_{[4]}}
+\begin{bmatrix} {\bf 0}_{3\times 4}\\ V_{{\sf {\sf d}},1}^{\sf d}\\V_{{\sf {\sf d}},2}^{\sf d}\\ {\bf 0}_{3\times 4}\\ V_{{\sf {\sf d}},3}^{\sf d}\\ V_{{\sf {\sf d}},4}^{\sf d}\end{bmatrix}\underbrace{\begin{bmatrix}{\sf D}_1\\{\sf D}_2\\{\sf D}_3\\{\sf D}_4\end{bmatrix}}_{{\sf D}_{[4]}}
\end{align}
which indicates the precoding operations at each server with an encoding matrix $V_{-}^-$. For example, Server $\mathcal{S}_{\sf ab}$, precodes the $L\times 1 = 4\times 1$ vector of data stream ${\sf A}$ (denoted as ${\sf A}_{[4]}$) with the $2N_{\sf ab}\times L = 2\times 4$ precoding matrix $V_{\sf ab}^{\sf a}$, whose $i^{th}$ row is denoted by $V_{{\sf ab},i}^{\sf a}$. Similarly, Server $\mathcal{S}_{\sf ab}$ precodes data stream ${\sf B}$ with the $2\times 4$ precoding matrix $V_{\sf ab}^{\sf b}$. The precoded symbols are then mapped to the inputs controlled by  Server ${\sf ab}$, i.e., $x_1,x_6$, so that we have,
\begin{align}
\begin{bmatrix}x_1\\ x_6\end{bmatrix}&=V_{\sf ab}^{\sf a}{\sf A}_{[4]}+V_{\sf ab}^{\sf b}{\sf B}_{[4]}.
\end{align}
Each server similarly precodes the data streams available to it with its corresponding precoding matrices.  

The precoding matrices are now specified as,
{\small
\begin{align}\label{eq:precod}
  &\begin{bmatrix}
    V_{\sf ab}^{\sf a}\\V_{\sf ac}^{\sf a}
  \end{bmatrix}
  =
  \big(V_{\dec}{\bf M}_{(1,6,2,7)}\big)^{-1}, 
  ~\begin{bmatrix}
    V_{\sf ab}^{\sf b}\\V_{\sf bc}^{\sf b}
  \end{bmatrix}
  =
  \big(V_{\dec}{\bf M}_{(1,6,3,8)}\big)^{-1}, \notag\\
  &\begin{bmatrix}
    V_{\sf ac}^{\sf c}\\V_{\sf bc}^{\sf c}
  \end{bmatrix}
  =
  \big(V_{\dec}{\bf M}_{(2,7,3,8)}\big)^{-1}, 
  ~ V_{{\sf {\sf d}}}^{\sf d}
  =
  \big(V_{\dec}{\bf M}_{(4,5,9,10)}\big)^{-1}.
\end{align}
}
where ${\bf M}_{(i_1,i_2,\cdots, i_n)}$ is an $N\times n$ submatrix of ${\bf M}$ comprised of the $(i_1,i_2,\cdots, i_n)^{th}$ columns of ${\bf M}$.  It is easy to verify that $\det(V_{\dec}{\bf M}_{(1,6,2,7)})=\det(V_{\dec}{\bf M}_{(2,7,3,8)})=1 $ and $\det(V_{\dec}{\bf M}_{(1,6,3,8)})=\det(V_{\dec}{\bf M}_{(4,5,9,10)})=-1$, thus all $4$ inverses in \eqref{eq:precod} exist. 

\noindent {\bf Correctness:}
With all choices explicitly specified, it is similarly easy to verify that we have,
\begin{align}\label{eq:success}
V_{\dec}{\bf y}&=V_{\dec}{\bf M}{\bf x}={\sf A}_{[4]}+{\sf B}_{[4]}+{\sf C}_{[4]}+{\sf D}_{[4]}.
\end{align}
Thus, Alice is able to compute $4$ instances of the desired sum, with the total download cost of $5$ qudits. The coding scheme achieves the rate $4/5$ qudits/computation, matching the capacity of this $\Sigma$-QMAC setting.

\subsection{Example: Symmetric $\Sigma$-QMAC with $S=8$}\label{sec:symsqmac}
Setting $S=8$, for $\alpha,\beta \in \{1,2,\cdots,8\}$, we show the values of $C_{\alpha}^{(\beta)}$ in the following table according to Corollary \ref{cor:symmetric}.
\begin{table}[h]
\center
\caption{$C_{\alpha}^{(\beta)}$ for $S=8$.} \label{tab:S8}
\begin{tabular}{|c|c|c|c|c|c|c|c|c|}
\hline
\diagbox{$\alpha$}{$C_{\alpha}^{(\beta)}$}{$\beta$}
& 1 & 2 & 3 & 4 & 5 & 6 & 7 & 8 \\ \hline
1 & $1/8$ & $1/4$ & $1/4$ & $1/4$ & $1/4$ & $1/4$ & $1/4$ & $1/4$ \\ \hline
2 & $1/4$ & $13/28$ & $13/28$ & $1/2$ & $1/2$ & $1/2$ & $1/2$ & $1/2$ \\ \hline
3 & $3/8$ & $9/14$ & $9/14$ & $5/7$ & $5/7$ & $3/4$ & $3/4$ & $3/4$ \\ \hline
4 & $1/2$ & $11/14$ & $11/14$ & $61/70$ & $61/70$ & $13/14$ & $13/14$ & $1$ \\ \hline
5 & $5/8$ & $25/28$ & $25/28$ & $27/28$ & $27/28$ & $1$ & $1$ & $1$ \\ \hline
6 & $3/4$ & $27/28$ & $27/28$ & $1$ & $1$ & $1$ & $1$ & $1$ \\ \hline
7 & $7/8$ & $1$ & $1$ & $1$ & $1$ & $1$ & $1$ & $1$ \\ \hline
8 & $1$ & $1$ & $1$ & $1$ & $1$ & $1$ & $1$ & $1$ \\ \hline
\end{tabular}
\end{table}
The first column $(\beta=1)$ corresponds to the fully-unentangled capacities. The second column $(\beta=2)$ corresponds to the $2$-party-entangled capacities. The capacities in this column are achievable by the $2$-sum protocol. The last column $(\beta = 8)$ corresponds to the fully-entangled capacities. Comparing the columns corresponding to $\beta=1$ and $\beta=8$, note that the DSC gain is $2$ provided that the capacity does not exceed $1$ dit/qudit. Also, note that the bipartite entanglement (i.e., $\beta=2$) is in general not enough to achieve the maximal DSC gain.

\subsection{Example: Minimizing the Maximal Entanglement} \label{sec:ex_min_beta}
Maintaining entanglement across many parties tends to be increasingly challenging, as more parties become involved. So it is desirable to have smaller cliques without losing the capacity, motivating the problem of identifying entanglement distribution maps $\mathcal{E}$, with the maximal clique size as small as possible such that the capacity achieved with $\mathcal{E}$ is the same as $C^{\fullent}$.  Mathematically, the problem is to find $\beta^* \triangleq \min \{ \beta  \mid C^{(\beta)} = C^{\fullent} \}$ for a fixed data replication map.  For example, consider the data distribution $\mathcal{W}$ shown in Fig. \ref{fig:abbccad}. We can see from Table \ref{tab:abacbcd} that the smallest value of $\beta$ for this example is $4$, same as $S$, i.e., all servers need to be entangled, because even if every subset of $3$ of the $4$ servers has an entangled system, the capacity is still only $3/4$, which is still less than $C^{\fullent}=4/5$.
However, as evident from Table \ref{tab:S8}, where the data replication map is symmetric, it is in general not necessary to have all servers entangled in order to achieve the fully-entangled capacity $C^{\fullent}$. For example, for $S=8$ servers and $\alpha=3$, we have $C_3^{(6)} = C_3^* = 3/4$, which means that it suffices to have $\beta = 6$ in order to achieve the fully-entangled capacity. 
Therefore, based on Table \ref{tab:S8}, for the cases with symmetric data replication, i.e., $K=\binom{S}{\alpha}$ and $\mathcal{W}: [K] \biject \binom{[S]}{\alpha}$, we have $\beta^* = 2,4,6,8,6,4,2,1$ for $\alpha = 1,2,3,4,5,6,7,8$. From Corollary \ref{cor:symmetric} it can further be verified that for general $S$,
\begin{align} \label{eq:beta_star}
	\beta^* = \begin{cases}
		2\alpha, &\alpha \leq \lfloor S/2 \rfloor, \\
		2(S-\alpha), &\lceil S/2 \rceil \leq \alpha \leq S-1, \\
		1, & \alpha = S.
	\end{cases}
\end{align}
The proof of \eqref{eq:beta_star} can be found in Appendix \ref{proof:symmetric}.
The intuition that emerges from this is that both extremes of too much data replication (large $\alpha$) and too little data replication (small $\alpha$) require relatively little entanglement (small $\beta^*$) to achieve their maximal DSC gain,  rather the intermediate regimes of data replication are the ones that require the most entanglement to maximize their DSC gain.

\subsection{Example: $2$-sum Protocol Based Coding for Fig. \ref{fig:abbccad}}\label{sec:2sumcapacity}
The main purpose of this example is to illustrate Corollary \ref{cor:bipartite}. Note that there are two $\Sigma$-QMAC problems involved in Corollary \ref{cor:bipartite}. In this example, let us again consider the data replication map defined in Fig. \ref{fig:abbccad}. We consider the following two $\Sigma$-QMAC problems.
\begin{enumerate}
	\item $\mathcal{P}$: The $\Sigma$-QMAC problem with data replication map $\mathcal{W}$ as shown in Fig. \ref{fig:abbccad}. Specifically, there are four data streams, denoted as ${\sf A}, {\sf B}, {\sf C}$ and ${\sf D}$. The four servers are denoted as $\mathcal{S}_{\sf ab},\mathcal{S}_{\sf ac},\mathcal{S}_{\sf bc}$ and $\mathcal{S}_{\sf d}$. We wish to find the $2$-party-entangled capacity $C^{(2)}(\mathcal{W})$.
	\item $\widetilde{\mathcal{P}}$: The (hypothetical) $\Sigma$-QMAC problem with data replication map $\widetilde{\mathcal{W}}$. Specifically, there are $\binom{4}{2} = 6$ servers, each has the data streams that are available to a unique pair of servers in $\mathcal{P}$. Therefore, the data streams available to the servers in $\widetilde{\mathcal{P}}$ are ${\sf ABC}, {\sf ABC}, {\sf ABD}, {\sf ABC}, {\sf ACD}, {\sf BCD}$, respectively. Without loss of generality, the three servers that know ${\sf ABC}$ can be considered as one server. We refer to the server that has data streams ${\sf ABC}$ as $\mathcal{S}_{\sf abc}$, and similarly we define the rest $3$ servers as $\mathcal{S}_{\sf abd}, \mathcal{S}_{\sf acd}$ and $\mathcal{S}_{\sf bcd}$. We wish to find the fully-unentangled capacity $C^{\unent}(\widetilde{\mathcal{W}})$.
\end{enumerate}
According to Corollary \ref{cor:bipartite}, these two $\Sigma$-QMAC problems have the same capacity, i.e., $C^{(2)}(\mathcal{W}) = C^{\unent}(\widetilde{\mathcal{W}})$,  both equal to $3/4$ by Corollary \ref{cor:unentangled}. A scheme for the problem $\widetilde{\mathcal{P}}$ is illustrated in Table \ref{tab:2sum_coding}.
\begin{table}[h]
\center
\begin{tabular}{c|l}
    Server & Transmission \\\hline
    $\mathcal{S}_{\sf abc}$ & $Y_{\sf abc} = ({\sf A}_1-{\sf A}_2+{\sf A}_3)+({\sf B}_1-{\sf B}_2) + {\sf C}_1$ \\\hline
    $\mathcal{S}_{\sf abd}$ & $Y_{\sf abd} = ({\sf A}_2 -{\sf A}_3)+{\sf B}_2+{\sf D}_1$ \\\hline
    $\mathcal{S}_{\sf acd}$ & $Y_{\sf acd} = {\sf A}_3+{\sf C}_2+({\sf D}_2-{\sf D}_1)$ \\\hline
    $\mathcal{S}_{\sf bcd}$ & $Y_{\sf bcd} = {\sf B}_3 +({\sf C}_3-{\sf C}_2)+({\sf D}_3-{\sf D}_2+{\sf D}_1)$
\end{tabular}
\vspace{0.1in}
\caption{Coding scheme for $\widetilde{\mathcal{P}}$} \label{tab:2sum_coding}
\end{table}
Note that in $\widetilde{\mathcal{P}}$, no entanglement across different servers is allowed, and our scheme in Table \ref{tab:2sum_coding} simply treats qudits as dits. Since $Y_{\sf abc}+Y_{\sf abd} = {\sf A}_1+{\sf B}_1+{\sf C}_1+{\sf D}_1$, $Y_{\sf abd}+Y_{\sf acd} = {\sf A}_2+{\sf B}_2+{\sf C}_2+ {\sf D}_2$ and $Y_{\sf acd}+Y_{\sf bcd} = {\sf A}_3+{\sf B}_3+{\sf C}_3+{\sf D}_3$, the scheme allows Alice to compute $L=3$ instances of the desired sum with a total cost of $4$ qudits ($Y_{\sf abc}$, $Y_{\sf abd}$, $Y_{\sf acd}$ and $Y_{\sf bcd}$ each costs one qudit). The capacity $3/4$ is thus achieved.

Fig. \ref{fig:parallel} shows that a  scheme thats achieve $C^{(2)}(\mathcal{W}) = 3/4$ in the problem $\mathcal{P}$ can be deduced from Table \ref{tab:2sum_coding}, which allows Alice to compute $6$ instances of the sum in the problem $\mathcal{P}$, with $4$ uses of the $2$-sum protocols, thus achieving the rate $6/8=3/4$.
\begin{figure}[!h]
\begin{center}
\begin{tikzpicture}
\coordinate (O) at (0,0){};
\node [draw, cylinder, aspect=0.2,shape border rotate=90,fill=black!10, text=black, inner sep =0cm, minimum height=0.5cm, minimum width=0.5cm, below left = 0.5cm and 0.4cm of O, align=center] (ABC)  {\footnotesize ${\sf ABC}$};
\node [draw, cylinder, aspect=0.2,shape border rotate=90,fill=black!10, text=black, inner sep =0cm, minimum height=0.5cm, minimum width=0.5cm, right=0.5cm of ABC, align=center] (ABD)  {\footnotesize ${\sf ABD}$};
\node [draw, cylinder, aspect=0.2,shape border rotate=90,fill=black!10, text=black, inner sep =0cm, minimum height=0.5cm, minimum width=0.5cm, right=0.5cm of ABD, align=center] (ACD)  {\footnotesize ${\sf ACD}$};
\node [draw, cylinder, aspect=0.2,shape border rotate=90,fill=black!10, text=black, inner sep =0cm, minimum height=0.5cm, minimum width=0.5cm, right=0.5cm of ACD, align=center] (BCD)  {\footnotesize ${\sf BCD}$};

\node[alice, draw=black, text=black, minimum size=0.8cm, inner sep=0, below = 2.3cm of O] (U2) at (1,-0.6) {};

\draw [-latex] ($(ABC.south)+(0,0)$)--($(U2.north)+(-0.3,0)$) node[pos=0.5,left=0] {\tiny $Y_{\sf abc}$};
\draw [-latex] ($(ABD.south)+(0,0)$)--($(U2.north)+(-0.1,0)$) node[pos=0.3,left=-0.1] {\tiny $Y_{\sf abd}$};
\draw [-latex] ($(ACD.south)+(0,0)$)--($(U2.north)+(+0.1,0)$) node[pos=0.3,right=-0.05] {\tiny $Y_{\sf acd}$};
\draw [-latex] ($(BCD.south)+(0,0)$)--($(U2.north)+(+0.3,0)$) node[pos=0.5,right=0] {\tiny $Y_{\sf bcd}$};

\node at ($(U2.south)+(0,-0.5)$) {\footnotesize { Rate: $\frac{3}{(1+1+1+1)}$ (dits/dit)}};

\begin{scope}[shift={(6.5,0)}]
\coordinate (O) at (0,0){};
\node [draw, cylinder, aspect=0.2,shape border rotate=90,fill=black!10, text=black, inner sep =0cm, minimum height=0.5cm, minimum width=0.5cm, below left = 0.5cm and 0.4cm of O, align=center] (AB)  {\footnotesize ${\sf AB}$};
\node [draw, cylinder, aspect=0.2,shape border rotate=90,fill=black!10, text=black, inner sep =0cm, minimum height=0.5cm, minimum width=0.5cm, right=0.5cm of AB, align=center] (AC)  {\footnotesize ${\sf AC}$};
\node [draw, cylinder, aspect=0.2,shape border rotate=90,fill=black!10, text=black, inner sep =0cm, minimum height=0.5cm, minimum width=0.5cm, right=0.5cm of AC, align=center] (BC)  {\footnotesize ${\sf BC}$};
\node [draw, cylinder, aspect=0.2,shape border rotate=90,fill=black!10, text=black, inner sep =0cm, minimum height=0.5cm, minimum width=0.5cm, right=0.5cm of BC, align=center] (D)  {\footnotesize ${\sf D}$};

\node[rectangle, draw, thick, fill=white, minimum width=0.5cm, minimum height = 0.5cm, align=center, below = 0.7cm of AB] (Box1) {\tiny 2-sum};
\node[rectangle, draw, thick, fill=white, minimum width=0.5cm, minimum height = 0.5cm, align=center, below = 0.7cm of AC] (Box2) {\tiny 2-sum};
\node[rectangle, draw, thick, fill=white, minimum width=0.5cm, minimum height = 0.5cm, align=center, below = 0.7cm of BC] (Box3) {\tiny 2-sum};
\node[rectangle, draw, thick, fill=white, minimum width=0.5cm, minimum height = 0.5cm, align=center, below = 0.7cm of D] (Box4) {\tiny 2-sum};

\node[alice, mirrored, draw=black, text=black, minimum size=0.8cm, inner sep=0, below = 2.3cm of O] (U) at (1,-0.6) {};

\draw [double,-latex] ($(AB.south)+(0,0)$)--($(Box1.north)+(0,0)$);
\draw [double,-latex] ($(AC.south)+(0,0)$)--($(Box1.north)+(0,0)$);
\draw [double,-latex] ($(AB.south)+(0,0)$)--($(Box2.north)+(0,0)$);
\draw [double,-latex] ($(D.south)+(0,0)$)--($(Box2.north)+(0,0)$);
\draw [double, -latex] ($(AC.south)+(0,0)$)--($(Box3.north)+(0,0)$);
\draw [double, -latex] ($(D.south)+(0,0)$)--($(Box3.north)+(0,0)$);
\draw [double, -latex] ($(BC.south)+(0,0)$)--($(Box4.north)+(0,0)$);
\draw [double, -latex] ($(D.south)+(0,0)$)--($(Box4.north)+(0,0)$);
\draw [double,-latex] ($(Box1.south)+(0,0)$)--($(U.north)+(-0.5,0)$) node[pos=0.7,left=0.1] {\tiny $Y_{\sf abc}$};
\draw [double,-latex] ($(Box2.south)+(0,0)$)--($(U.north)+(-0.2,0)$) node[pos=0.4,left=-0.1] {\tiny $Y_{\sf abd}$};
\draw [double,-latex] ($(Box3.south)+(0,0)$)--($(U.north)+(+0.2,0)$) node[pos=0.4,right=-0.05] {\tiny $Y_{\sf acd}$};
\draw [double,-latex] ($(Box4.south)+(0,0)$)--($(U.north)+(+0.5,0)$) node[pos=0.7,right=0.1] {\tiny $Y_{\sf bcd}$};

\node at ($(U.south)+(0.2,-0.5)$) {\footnotesize {Rate: $\frac{6}{2(1+1+1+1)}$ (dits/qudit)}};
\end{scope}

\end{tikzpicture}
\end{center}
\caption{A comparison of the schemes that achieve $C^{\unent}(\widetilde{\mathcal{W}}) = 3/4$ in the problem $\widetilde{\mathcal{P}}$ (LHS) and $C^{(2)}(\mathcal{W}) = 3/4$ in the problem $\mathcal{P}$ (RHS). In RHS, each of $Y_{\sf abc}$, $Y_{\sf abd}$, $Y_{\sf acd}$ and $Y_{\sf bcd}$ contains two instances, e.g., $Y_{\sf abc} = \big( Y_{\sf abc}^{(1)}, Y_{\sf abc}^{(2)} \big)$, where $Y_{\sf abc}^{(1)}$ is the function of $({\sf A}_i,{\sf B}_i,{\sf C}_i,{\sf D}_i)_{i=1}^3$ as shown in Table \ref{tab:2sum_coding} and $Y_{\sf abc}^{(2)}$ is the corresponding function of $({\sf A}_i,{\sf B}_i,{\sf C}_i,{\sf D}_i)_{i=4}^6$.}\label{fig:parallel}
\end{figure}
Note that each use of the $2$-sum protocol transmits $2$ instances of the symbol that is sent by a server in the problem $\widetilde{\mathcal{P}}$. Specifically, for example, the two servers $\mathcal{S}_{\sf ab}$ and $\mathcal{S}_{\sf ac}$ in the problem $\mathcal{P}$ use the $2$-sum protocol once, with the two inputs at $\mathcal{S}_{\sf ab}$ specified as $({\sf A}_1-{\sf A}_2+{\sf A}_3)$ and  $({\sf A}_4-{\sf A}_5+{\sf A}_6)$, the two inputs at $\mathcal{S}_{\sf ac}$ specified as $({\sf B}_1-{\sf B}_2)+{\sf C}_1$ and $({\sf B}_4-{\sf B}_5)+{\sf C}_4$, so that Alice gets $Y_{\sf abc}^{(1)} = ({\sf A}_1-{\sf A}_2+{\sf A}_3) + ({\sf B}_1-{\sf B}_2)+{\sf C}_1$ and $Y_{\sf abc}^{(2)} = ({\sf A}_4-{\sf A}_5+{\sf A}_6) + ({\sf B}_4-{\sf B}_5)+{\sf C}_4$.

\subsection{$3$-party entanglement is not necessary in $\Sigma$-QMAC} \label{sec:ex_3_party}
As an example to illustrate Corollary \ref{cor:tripartite}, let us once again consider the data replication map in Fig. \ref{fig:abbccad}. One can quickly check from Table \ref{tab:abacbcd} that for any capacity that is associated with an entanglement distribution map that contains a $3$-party clique, the same capacity is achievable for another entanglement distribution map where the $3$-party clique is replaced with three $2$-party cliques, each containing a unique pair of servers in the $3$-party cliques.  Another example is the symmetric settings with $S=8$ servers as shown in Table \ref{tab:S8}. Note that in each row, the $3$-party-entangled capacity is always equal to the $2$-party-entangled capacity, meaning that $3$-party entanglement is not necessary for achieving a higher capacity.

\subsection{The necessity of $S$-party entanglement for $\Sigma$-QMAC $(S\neq 3)$} \label{sec:S_partite_necessary}
 The symmetric setting in Section \ref{sec:symsqmac} may support the intuition that $\beta$-partite entanglement is unnecessary for odd $\beta$, i.e., $C^{(\beta)}(\mathcal{W}) = C^{(\beta-1)}(\mathcal{W})$, because the columns corresponding to odd values of $\beta$ in Table \ref{tab:S8} are identical to their preceding columns. The intuition may even be strengthened by Corollary \ref{cor:tripartite} which shows that indeed $C^{(3)}(\mathcal{W}) = C^{(2)}(\mathcal{W})$ for any $\mathcal{W}$. Perhaps surprisingly then, Corollary \ref{cor:S_partite_necessary} reveals  that $\beta=3$ is only an exception and it is not generally true that $C^{(\beta)}(\mathcal{W}) = C^{(\beta-1)}(\mathcal{W})$ for all data replication patterns $\mathcal{W}$ if $\beta$ is an odd number.

Leaving the proof of Corollary \ref{cor:S_partite_necessary} to  Appendix \ref{proof:asymmetry}, let us consider here the case $\beta=S=5$  to see that indeed there exists a data replication pattern $\mathcal{W}$ such that $C^{(5)}(\mathcal{W}) > C^{(4)}(\mathcal{W})$. Let ${\sf A}, {\sf B}, {\sf C}, {\sf D}$ and ${\sf E}$  denote $K=5$ data streams. The data replication map $\mathcal{W}$ is such that Server $\mathcal{S}_1$ has ${\sf A,B,C}$, Server $\mathcal{S}_2$ has ${\sf A,B,D}$, Server $\mathcal{S}_3$ has ${\sf A,C,D}$, Server $\mathcal{S}_4$ has ${\sf B,C,D}$ and Server $\mathcal{S}_5$ has ${\sf E}$. If we only look at the first $4$ servers, this is the symmetric data replication map with $4$ data streams, each being replicated in a unique subset of $3$ servers. The asymmetry comes from the additional data stream, ${\sf E}$, that is only available at server $\mathcal{S}_5$. It can be verified by Theorem \ref{thm:main} that for this data replication map, $C^{\fullent}(\mathcal{W}) \triangleq C^{(5)}(\mathcal{W})  = 6/7$ and $C^{(4)}(\mathcal{W}) = 5/6$, which together show that the $5$-partite entanglement is necessary for this $\mathcal{W}$.

\section{Proof of Theorem \ref{thm:main}: Converse} \label{proof:main_conv}
Our converse bound for proving Theorem \ref{thm:main} is based on the cut-set argument with the capacity result of classical-quantum communication channel (e.g., \cite{holevo1998capacity,DSC4}). 
The definitions of quantum coding schemes, feasible region, capacity and the DSC gain follow from those of the $\Sigma$-QMAC.

\subsection{Prerequisite: Dense Coding Capacity} \label{sec:bg_conv}
Consider a point to point quantum communication setting with a sender, Bob, and a receiver, Alice. Quantum systems $\mathcal{Q}_A$ and $\mathcal{Q}_B$ are provided to Alice and Bob, respectively. We use $|\mathcal{Q}|$ to denote the dimension of a quantum system $\mathcal{Q}$. The composite system $\mathcal{Q}_B\mathcal{Q}_A$ is in the initial state $\rho^{BA}$, described by the density operator. Independent of $\rho^{BA}$ there is a random variable $X$, so that with probability $p_X(x)$, Bob applies a unitary operation $U_x$ on $\mathcal{Q}_B$. The resulting state of the system $\mathcal{Q}_B\mathcal{Q}_A$ is thus $\rho^{BA'}_x = (U_B\otimes I_A)\rho^{BA}(U_B^\dagger \otimes I_A)$ for $X=x$. Then Bob sends $\mathcal{Q}_B$ to Alice, which allows Alice to measure $\mathcal{Q}_B\mathcal{Q}_A$  and obtain the outcome $Y$. The \emph{dense coding capacity} \cite{DSC4}, defined as the maximum amount of information that Alice can learn about $X$ from $Y$, is equal to $\max I(X;Y)$ where the maximum is taken over all $p_X$ and all possible measurement at Alice. This value can be strictly larger than $\log_d|\mathcal{Q}_B|$ (dits), in which case the coding scheme is called a \emph{dense} coding. As shown by \cite{DSC4}, $\max I(X;Y) = \log_d|\mathcal{Q}_B| + S(\rho^{A})-S(\rho^{BA})$ (dits), where $S(\cdot)$ denotes the von Neumann entropy and $\rho^A$ is the reduced density operator for $\mathcal{Q}_A$. Due to the inequalities $|S(\rho^A)-S(\rho^B)| \leq S(\rho^{BA})$, $S(\mathcal{Q}_A)\leq \log_d|\mathcal{Q}_A|$ and non-negativity of von Neumann entropy, we obtain that 
\begin{align} \label{eq:bg_conv}
  I(X;Y) \leq \min(\log_d |\mathcal{Q}_B\mathcal{Q}_A|, 2\log_d |\mathcal{Q}_B|) \mbox{ dits}.
\end{align}

\subsection{Proof of Converse}
Consider any feasible coding scheme specified by
\begin{align}
	\big(L, ~((\delta_{t,s})_{s\in\mathcal{E}(t)})_{t\in[T]}, ~\rho_{[T]},~ ((\Phi_{t,s})_{s\in\mathcal{E}(t)})_{t\in[T]},~ (\{M_{t,y}\}_{y\in \mathcal{Y}_t})_{t\in [T]},~ \Psi \big).
\end{align}
\begin{lemma}[Conditional Independence]\label{lem:conditional_independence}
	Given any feasible scheme, ${\sf Y}_1,{\sf Y}_2, \cdots, {\sf Y}_T$ are mutually independent conditioned on the event $({\sf W}_1^{[L]}=w_1, {\sf W}_2^{[L]}=w_2, \cdots, {\sf W}_K^{[L]}=w_K)$ for any $w_1,w_2,\cdots, w_K \in \mathbb{F}_d^L$.
\end{lemma}
\begin{proof}
Recall that Alice receives the quantum system $\mathcal{Q}_t$ in the state $\rho_t' = (\otimes_{s\in \mathcal{E}(t)} U_{t,s})\rho_t (\otimes_{s\in \mathcal{E}(t)} U_{t,s})^\dagger$ for $t\in[T]$. Thus, for any $w_1,w_2,\cdots, w_K \in \mathbb{F}_d^L$,
\begin{align} \label{eq:prob_measure_t}
  {\sf Pr}({\sf Y}_t = y_t \mid  {\sf W}_1^{[L]}=w_1, \cdots, {\sf W}_K^{[L]}=w_K ) = \mbox{tr}(\rho_t' M_{t,y_t}), ~~\forall t\in[T]
\end{align}
and
\begin{align}
  &{\sf Pr}\big({\sf Y}_1 = y_1, \cdots, {\sf Y}_T = y_T \mid {\sf W}_1^{[L]}=w_1, \cdots, {\sf W}_K^{[L]}=w_K \big) \notag \\
  &= \mbox{tr}\big((\rho_1'\otimes\cdots\otimes\rho_T') ( M_{1,y_1}\otimes\cdots\otimes M_{T,y_T} )\big) \\ 
  & = \mbox{tr}\big( (\rho_1'M_{1,y_1}) \otimes \cdots \otimes (\rho_T'M_{T,y_T})\big) \\
  & = \prod_{t=1}^T \mbox{tr}(\rho_t'M_{t,y_t}) \\
  & = \prod_{t=1}^T{\sf Pr}(Y_t = y_t \mid {\sf W}_1^{[L]}=w_1, \cdots, {\sf W}_K^{[L]}=w_K)
\end{align}
where the last step uses \eqref{eq:prob_measure_t}. It follows that ${\sf Y}_1, {\sf Y}_2,\cdots,{\sf Y}_T$ are conditionally independent.
\end{proof}
Let $\big({\sf W}_1^{(\ell)},{\sf W}_2^{(\ell)},\cdots, {\sf W}_K^{(\ell)} \big)$ be the $\ell^{th}$ instance of the data streams. Since any feasible scheme must guarantee successful decoding for all realizations of $\big({\sf W}_1^{(\ell)},{\sf W}_2^{(\ell)},\cdots, {\sf W}_K^{(\ell)} \big) \in \mathbb{F}_d^K$ for all $\ell$, it must still guarantee successful decoding if we assume $\big({\sf W}_1^{(\ell)},{\sf W}_2^{(\ell)},\cdots, {\sf W}_K^{(\ell)} \big)$ to be uniform over $\mathbb{F}_d^K$ for any $\ell \in [L]$, and independent over $\ell\in [L]$.
For $t\in[T]$, let us account for the separate measurement corresponding to clique $\mathcal{E}(t)$. 
Following a regular cut set argument, for $k\in [K]$, denote $\mathcal{A}_{t,k} \triangleq \big([S]\backslash \mathcal{W}(k) \big) \cap \mathcal{E}(t)$ and let Servers $s\in \mathcal{A}_{t,k}$ join Alice as the receiver by bringing their quantum resource and data.  
Let $\mathcal{B}_{t,k} \triangleq \mathcal{E}(t) \backslash \mathcal{A}_{t,k} = \mathcal{E}(t)\cap \mathcal{W}(k)$ and consider Servers $s\in \mathcal{B}_{t,k}$ collectively as the transmitter. Denote the subsystem of $\mathcal{Q}_t$ that is sent from Servers $s\in \mathcal{B}_{t,k}$ as $\mathcal{Q}_{B}$ and the quantum subsystem of $\mathcal{Q}_t$ that is brought from Servers $s\in \mathcal{A}_{t,k}$ as $\mathcal{Q}_{A}$. 
Since the $K$ data streams are mutually independent, conditioned on ${\sf W}_{[K]\backslash \{k\}}^{[L]}$, it follows from \eqref{eq:bg_conv} that,
\begin{align}
& I\big({\sf W}_{k}^{[L]};{\sf Y}_t ~\big|~ {\sf W}_{[K]\backslash \{k\}}^{[L]} \big)\notag \\
& \leq \min\big(\log_d |\mathcal{Q}_{B}\mathcal{Q}_{A}|, 2 \log_d |\mathcal{Q}_{B}|\big) \\
& =\min\Big(\log_d\prod_{s\in \mathcal{E}(t)}|\mathcal{Q}_{t,s}|, 2 \log_d\prod_{s\in \mathcal{B}_{t,k}}|\mathcal{Q}_{t,s}|\Big) \\
& = \min\Big( \sum_{s\in \mathcal{E}(t)}\log_d |\mathcal{Q}_{t,s}|, 2\sum_{s\in \mathcal{B}_{t,k}} \log_d |\mathcal{Q}_{t,s}|\Big) \label{eq:conv_1}
\end{align}
where ${\sf Y}_t$ denotes the result after measuring the quantum system $\mathcal{Q}_{B}\mathcal{Q}_{A}$, i.e., the composite system comprised of $\mathcal{Q}_{B}$ and $\mathcal{Q}_{A}$.  

We then have
\begin{align}
  & \sum_{t\in[T]} I\big({\sf W}_{k}^{[L]};{\sf Y}_t ~\big|~ {\sf W}_{[K]\backslash \{k\}}^{[L]}\big) \notag \\
  & = \sum_{t\in[T]} H\big({\sf Y}_t ~\big|~ {\sf W}_{[K]\backslash \{k\}}^{[L]}\big) - \sum_{t\in[T]}H\big({\sf Y}_t ~\big|~ {\sf W}_{[K]}^{[L]}\big) \\
  & \geq H\big({\sf Y}_{[T]} ~\big|~ {\sf W}_{[K]\backslash \{k\}}^{[L]} \big) - \sum_{t\in[T]}H\big({\sf Y}_t ~\big|~ {\sf W}_{[K]}^{[L]}\big) \\
  & = H\big({\sf Y}_{[T]} ~\big|~ {\sf W}_{[K]\backslash \{k\}}^{[L]} \big) - \sum_{t\in[T]} H\big( {\sf Y}_t ~\big|~ {\sf W}_{[K]}^{[L]}, {\sf Y}_{[t-1]} \big) \label{eq:cond_independence} \\
  & = H\big({\sf Y}_{[T]} ~\big|~ {\sf W}_{[K]\backslash \{k\}}^{[L]} \big) - H\big({\sf Y}_{[T]}~\big|~ {\sf W}_{[K]}^{[L]}\big) \\
  & = I\big({\sf W}_{k}^{[L]};{\sf Y}_{[T]} ~\big|~ {\sf W}_{[K]\backslash \{k\}}^{[L]}\big) \\
  & = H\big( {\sf W}_k^{[L]} ~\big|~ {\sf W}_{[K]\backslash \{k\}}^{[L]} \big) \label{eq:conv_decoding} \\
  & = H\big( {\sf W}_k^{[L]} \big) \\
  & = L ~\mbox{(dits)} \label{eq:conv_2}
\end{align}
where Step \eqref{eq:cond_independence} makes use of the conditional independence of ${\sf Y}_1,{\sf Y}_2,\cdots,{\sf Y}_T$ as implied by Lemma \ref{lem:conditional_independence}. Step \eqref{eq:conv_decoding} holds because conditioned on ${\sf W}_{[K]\backslash \{k\}}^{[L]}$,  Alice must be able to recover ${\sf W}_k^{[L]}$ from ${\sf Y}_{[T]}$.

Combining \eqref{eq:conv_1} and \eqref{eq:conv_2}, we obtain that for any $k\in[K]$,
\begin{align}
	\sum_{t\in[T]}\min\Big( \sum_{s\in \mathcal{E}(t)}\log_d |\mathcal{Q}_{t,s}|, 2\sum_{s\in \mathcal{B}_{t,k}}|\mathcal{Q}_{t,s}|\Big) \geq L ~\mbox{(dits)}.
\end{align}
Dividing by $L$ on both sides gives us
\begin{align}
  \sum_{t\in[T]}\min \left\{\sum_{s\in \mathcal{E}(t)}\Delta_{t,s}, \sum_{s\in \mathcal{B}_{t,k}} 2\Delta_{t,s}\right\}  \geq 1, ~\forall k\in [K],
\end{align}
which matches the condition of the feasible region in Theorem \ref{thm:main}.

\section{Proof of Theorem \ref{thm:main}: Achievability}\label{proof:main_achievability}
\subsection{Prerequisite: The $N$-sum Box}\label{sec:nsumbox}
Building on the stabilizer formalism and quantum error correction literature on stabilizer codes, an implicit generalization of the $2$-sum protocol is presented in \cite{song_colluding_PIR}, and subsequently crystallized as an $N$-sum box abstraction in \cite{Allaix_N_sum_box}. The $N$-sum box has $2N$ classical inputs, labeled $x_1,x_2,\cdots, x_{2N}\in \mathbb{F}_d$, and $N$ classical outputs $y_1,y_2,\cdots,y_N\in \mathbb{F}_d$, related by a MIMO MAC channel formulation as,
\begin{align}
\begin{bmatrix}y_1\\\vdots\\y_N\end{bmatrix}&=\begin{bmatrix}M_{1,1}&\cdots&M_{1,{2N}}\\ \vdots&\vdots&\vdots\\M_{N,1}&\cdots&M_{N,{2N}}\end{bmatrix}\begin{bmatrix}x_1\\\vdots\\x_{2N}\end{bmatrix}
\end{align}
which can be represented compactly as ${\bf y}={\bf M}{\bf x}$. 
The $N$-sum box abstraction represents the setting where $N$ entangled qudits  are distributed among $K$ transmitters, such that each transmitter can perform conditional quantum $X,Z$ gate operations on its qudit(s) to encode classical information. The transmitter that has the $n^{th}$ qudit controls the inputs $x_{n}$ and $x_{N+n}$ of the $N$-sum box. For example, if Qudits $1$ and $3$ are given to Transmitter $1$, then in the $N$-sum box abstraction the inputs $x_1, x_{1+N}, x_3, x_{3+N}$ are the inputs available to Transmitter $1$. The $N$ outputs are the result of the quantum measurement performed by Alice. Since the $N$ qudits are sent to Alice for the quantum measurement, the $N$-sum box has a quantum communication cost of $N$ qudits. Now let us consider the channel matrix ${\bf M}$. Different choices of entanglement states and quantum-measurement bases produce different channel matrices. Depending on the desired computation task a suitable ${\bf M}$ may be chosen from the set of feasible choices. The channel matrices that can be obtained from the stabilizer-based construction are precisely those (see \cite{Allaix_N_sum_box}) that are strongly self-orthogonal, i.e., that satisfy the following two conditions,
\begin{align}
\rk({\bf M})&=N,&&
{\bf M}{\bf J}_{2N}{\bf M}^T={\bf 0}_{N\times N}
\end{align}
where 
${\bf J}_{2N}=\ppsmatrix{{\bf 0}&-{\bf I}_{N}\\ {\bf I}_N&{\bf 0}}$
and ${\bf I}_N$ is the $N\times N$ identity matrix.  Designing quantum-codes for the  $\Sigma$-QMAC using the $N$-sum box abstraction entails a choice of not only which $N$-sum boxes to use, how many of the inputs of each $N$-sum box to assign to each transmitter, and how to precode at each transmitter in the MIMO MAC for the desired computation, but in contrast to conventional (wireless) MIMO MAC settings where the channels are randomly chosen by nature, here we also have the freedom to design suitable channel matrices ${\bf M}$ for the desired computation task, within the class of feasible choices. The $N$-sum box abstraction then guarantees that corresponding to these choices there exist initial quantum entanglements, quantum-coding operations at the transmitters, and quantum-measurement operations at the receiver, that achieve the desired MIMO MAC functionality, at the communication cost of $N$ qudits for each $N$-sum box utilized by the coding scheme.

\subsection{Proof of Achievability} 
Our proof of achievability combines a series of results in network coding literature, together with the $N$-sum box formulation. Therefore, let us first summarize these results into the following lemmas to facilitate the proof. Consider the following setup.
Let $K\in \mathbb{N}$. For $k\in [K]$, let ${\bf H}_k \in \mathbb{F}_q^{n\times m_k}$. Let ${\sf W}$ and ${\sf W}_k, k\in[K]$ be sources generating symbols in $\mathbb{F}_q$. Let us define the following two network coding type problems.

\noindent {\bf Sum-network:} There is a MIMO multiple access channel with $K$ transmitters and one receiver. The input at Transmitter $k$ is $X_k \in \mathbb{F}_q^{m_k \times 1}$. The output at the receiver is $Y = \sum_{k\in[K]} {\bf H}_k X_k$. Transmitter $k$ knows ${\sf W}_k$. The receiver wants to know ${\sf W}_1+{\sf W}_2+\cdots+{\sf W}_K$. A feasible  coding scheme can be described by $(L,N,\phi_{[K]}, \psi)$ so that the encoders $\phi_k$ map ${\sf W}^{[L]}_k$ to $X_k^{[N]}, \forall k\in[K]$, and the decoder $\psi$ maps $Y^{[N]}$ to $\sum_{k\in[K]} {\sf W}_k^{[L]}$. A rate $R$ is achievable if there exists a scheme so that $R \leq L/N$.

\noindent {\bf Multicast:}
There is a MIMO broadcast channel with one transmitter and $K$ receivers. The input at the Transmitter is $\widetilde{X} \in \mathbb{F}_q^{n \times 1}$. The output at Receiver $k$ is $\widetilde{Y}_k = {\bf H}_k^T \widetilde{X}$. The transmitter knows a message ${\sf W}$. All $K$ receivers want to decode ${\sf W}$. A feasible  coding scheme can be described by $(L,N,\widetilde{\phi}, \widetilde{\psi}_{[K]})$ so that the encoder $\widetilde{\phi}$ maps ${\sf W}^{[L]}$ to $\widetilde{X}^{[N]}$, and the decoders $\widetilde{\psi}_k$ map $\widetilde{Y}_k^{[N]}$ to ${\sf W}^{[L]}, \forall k\in[K]$. A rate $R$ is achievable if there exists a scheme so that $R \leq L/N$.

\begin{lemma}[Duality \cite{Rai_Dey}]
	If $(L,N,\phi_{[K]}, \psi)$ is a feasible linear coding scheme\footnote{If the coding functions $\phi_{[K]}$ and decoding function $\psi$ are linear functions, the scheme is said to be linear. The linearity for the Multicast setting is similarly defined.} for the sum-network, then there exists a  feasible linear coding scheme $(L,N,\widetilde{\phi}, \widetilde{\psi}_{[K]})$ for the corresponding multicast problem, and vice versa. 
\end{lemma}

\begin{lemma}[Multicast Capacity \cite{Li_Yeung_linear_network_coding}] \label{lem:multicast}
	The capacity (supreme of  achievable rates) of the multicast problem is $\min_{k\in [K]} \rk({\bf H}_k)$, and it can be achieved by linear coding schemes.
\end{lemma}
We only need the achievability side of Lemma \ref{lem:multicast}.
Although the idea essentially follows from \cite{Ahlswede_network_information_flow, Li_Yeung_linear_network_coding}, since our formulation here is slightly different, we provide an alternative proof of this lemma in Appendix \ref{proof:multicast}. Based on these two lemmas, the next corollary becomes obvious.
\begin{corollary} \label{cor:sum_network}
	The rate $R=\min_{k\in[K]} \rk({\bf H}_k)$ is achievable by a linear coding scheme in the sum-network.
\end{corollary}

Going back to the proof, our achievable scheme is essentially based on the achievable scheme of a sum-network, which is constructed by the $N$-sum box formulation. Recall that there are in total $T$ cliques (sets of servers that are allowed to share an entangled quantum system). For the $t^{th}$ clique ($t\in[T]$), we let the servers $\mathcal{E}(t)$ implement an $N_t$-sum box in $\mathbb{F}_q = \mathbb{F}_{d^z}$, so that Server $s\in \mathcal{E}(t)$ controls $2N_{t,s}$ inputs in $\mathbb{F}_q$, by using a quantum subsystem $\mathcal{Q}_{t,s}$ with dimension specified to $\delta_{t,s} = q^{N_{t,s}}$. Equivalently, $\mathcal{Q}_{t,s}$ can be considered as $N_{t,s}$ $q$-ary quantum subsystems, or $N_{t,s}z$ qudits. $z \in \mathbb{N}$ is free to be chosen later.

For $t\in [T]$, the transfer matrix of the $t^{th}$ box is denoted as ${\bf M}_t \in \mathbb{F}_q^{N_t\times 2N_t}$.
Recall that for $k\in [K]$, each data stream ${\sf W}_k$ generates symbols in $\mathbb{F}_d$. Equivalently, we can consider these symbols as in $\mathbb{F}_q$ where $q = d^z$ for any $z\in \mathbb{N}$, by regarding each $z$ symbols as one super-symbol in $\mathbb{F}_q$. Since we do not put any constraint on the batch size $L$, let ${\bf W}_k \in \mathbb{F}_q^{L'\times 1}$ denote the first $L'$ symbols of ${\sf W}_k$ considered in $\mathbb{F}_q$ (which correspond to $L = L'z$ symbols in the original field $\mathbb{F}_d$, or ${\sf W}_k^{(L'z)}$). 

Next let us define the input and the output of the $t^{th}$ box. The input of the $t^{th}$ box is ${\bf x}_t\in \mathbb{F}_q^{2N_t \times 1}$, and therefore the output of the $t^{th}$ box is ${\bf y}_t = {\bf M}_t{\bf x}_t \in \mathbb{F}_q^{N_t \times 1}$. ${\bf x}_t$ consists of the $N_{t,s}$ pairs of inputs controlled by Servers $s\in \mathcal{E}(t)$. It is then obvious that the $t^{th}$ box is a MIMO-MAC channel with all symbols defined in $\mathbb{F}_q$, and the input of each server $s\in \mathbb{E}(t)$ corresponds to some $2N_{t,s}$ columns of ${\bf M}_t$. We will also say that these $2N_{t,s}$ columns are controlled by Server $s$.
A column of ${\bf M}_t$ is said to be \emph{accessible} by a data stream ${\sf W}_k, k\in [K]$ if and only if this column is controlled by some server $s \in \mathcal{E}(t)$ that also knows this data stream, i.e., $s \in \mathcal{E}(t)\cap \mathcal{W}(k)$.  

For $k\in [K]$, define ${\bf M}_{t,k}$ as the set of columns of ${\bf M}_t$ that are \emph{accessible} by the $k^{th}$ data stream. Note that ${\bf M}_{t,k}$ can be empty for some $k$, if the $t^{th}$ clique does not contain any server that knows ${\sf W}_k$. We claim that the output of the $t^{th}$ box can be made as ${\bf y}_t = \sum_{k\in[K]}{\bf M}_{t,k} \phi_{t,k}({\bf W}_k)$, where $\phi_{t,k}$ maps ${\bf W}_k$ to a vector in $\mathbb{F}_q$ with length $\sum_{s\in \mathcal{E}(t) \cap \mathcal{W}(k)}2N_{t,s}$. 

To prove the claim, for any $t\in [T], k\in [K]$, let $\phi_{t,s}^k: \mathbb{F}_q^{L'\times 1} \mapsto \mathbb{F}_q^{2N_{t,s} \times 1}$ describe a map if $s\in \mathcal{E}(t) \cap \mathcal{W}(k)$, and let the input for any one use of the $t^{th}$ MIMO-MAC channel (i.e., the $t^{th}$ box) at Server $s\in \mathcal{E}(t)$ be specified as $\sum_{k: s \in  \mathcal{W}(k)} \phi_{t,s}^k ({\bf W}_k)$.
The output of the $t^{th}$ box is then determined as ${\bf y}_t\in \mathbb{F}_q^{N_t \times 1}$ such that,
\begin{align} \label{eq:output_box_clique}
  {\bf y}_t = {\bf M}_t{\bf x}_t = \sum_{k\in [K]} {\bf M}_{t,k} 
  \underbrace{\begin{bmatrix}
    \phi_{t,s_1}^k ({\bf W}_k) \\ \phi_{t,s_2}^k ({\bf W}_k) \\ \vdots \\ \phi_{t,s_n}^k ({\bf W}_k)
  \end{bmatrix}}_{\phi_{t,k}({\bf W}_k)}
\end{align}
where $\{s_1,s_2,\cdots,s_n\} = \mathcal{E}(t)\cap \mathcal{W}(k)$, and ${\bf M}_{t,k}$ is a submatrix of ${\bf M}_t$ comprised of $\sum_{s \in  \mathcal{E}(t) \cap  \mathcal{W}(k)} 2N_{t,s}$ columns of ${\bf M}_t$ that are accessible by ${\bf W}_k$.
Therefore, we obtain the general expression of the output ${\bf y}_t$ as in \eqref{eq:output_box_clique}.

Note that we have in total $T$ boxes and therefore $T$ MIMO-MAC channels. We can equivalently consider the $T$ channels as one big channel, and the output of the big channel (for one channel use) can be written as
\begin{align} \label{eq:MIMO_MAC}
  {\bf y} = 
  \begin{bmatrix}
    {\bf y}_1 \\ {\bf y}_2 \\ \vdots \\ {\bf y}_T
  \end{bmatrix} = 
  \sum_{k\in [K]}
  \underbrace{\begin{bmatrix}
    {\bf M}_{1,k} & {\bf 0} & \cdots & {\bf 0} \\
    {\bf 0} & {\bf M}_{2,k} & \cdots & {\bf 0} \\
    \vdots & \vdots & \ddots & \vdots \\
    {\bf 0} & {\bf 0} & \cdots & {\bf M}_{T,k}
  \end{bmatrix}}_{\overline{\bf M}_k}
  \underbrace{\begin{bmatrix}
    \phi_{1,k}({\bf W}_k) \\ \phi_{2,k}({\bf W}_k) \\ \vdots \\ \phi_{T,k}({\bf W}_k)
  \end{bmatrix}}_{\phi_k({\bf W}_k)}.
\end{align}
Note that there is no restriction on $\phi_{k}$ except the dimensions of its input and output. \eqref{eq:MIMO_MAC} thus describes the input-output relation of a $K$-transmitter MIMO MAC with the $k^{th}$ transmitter knowing only ${\sf W}_k$. According to Corollary \ref{cor:sum_network}, the receiver of this channel (Alice) can compute the sum $\sum_{k\in [K]}{\sf W}_k$ (in $\mathbb{F}_q$) at the rate $\min_{k\in[K]} \rk_q(\overline{\bf M}_k)$ (per channel). This is saying  that with each use of the big channel, Alice is able to get $\min_{k\in[K]} \rk_q(\overline{\bf M}_k)z$ sums in $\mathbb{F}_d$.
Recall that each use of the big channel corresponds to the use of quantum subsystem $\mathcal{Q}_{t,s}$ with its dimension equal to $\delta_{t,s} = q^{N_{t,s}}$ for all $t\in [T], s\in \mathcal{E}(t)$. Since $q = d^z$, it follows that $\log_d \delta_{t,s} = N_{t,s}z$ and that the following set of tuples are feasible,
\begin{align} \label{eq:feasible_set_1}
  \mbox{closure}\left\{ {\bm \Delta} \in \mathbb{R}_+^\Gamma \left|
  \begin{array}{l}
  	  N_{t,s}\in \mathbb{N}, \forall t\in[T], s\in\mathcal{E}(t),\\
  	  \Delta_{t,s} \geq N_{t,s}/\min_{k\in[K]}\rk(\overline{\bf M}_k), t\in [T], s\in \mathcal{E}(t)
  \end{array}
  \right. \right\}.
\end{align}
For fixed $N_{t,s}, t\in [T], s\in \mathcal{E}(t)$, we would like $\min_{k\in[K]}\rk(\overline{\bf M}_k)$ to attain its largest possible value to obtain the largest feasible set. According to \eqref{eq:MIMO_MAC}, it suffices to let ${\bf M}_{t,k}$  have full rank for all $t\in [T], k\in [K]$. This can be done by letting ${\bf M}_1,{\bf M}_2,\cdots, {\bf M}_T$  be `Half-MDS', which is defined as follows.

\begin{definition}[Half-MDS]
	Say ${\bf M} = [{\bf M}^l, {\bf M}^r]\in \mathbb{F}_q^{N\times 2N}$ is the transfer matrix of an $N$-sum box operating in $\mathbb{F}_q$. ${\bf M}^l\in \mathbb{F}_q^{N\times N}$ denotes the left half and ${\bf M}^r\in \mathbb{F}_q^{N\times N}$ denotes the right half. Let $i_1, i_2,\cdots, i_{n} \in \mathbb{N}$ be $n\leq N$ distinct indices not greater than $N$. We say ${\bf M}$ is half-MDS if for all such indices, 
	$$\rk\big(\big[{\bf M}^l_{(i_1,\cdots, i_{n})}, {\bf M}^r_{(i_1,\cdots, i_{n})}\big]\big) = \min\{2n, N\},$$ where ${\bf M}_{(i_1,i_2,\cdots, i_n)}$ denotes the $N\times n$ submatrix of ${\bf M}$ comprised of the $(i_1,i_2,\cdots,i_n)^{th}$ columns of ${\bf M}$. As an example, consider feasible transfer matrices for $2$-sum boxes,
\begin{align}
  {\bf M}_1 = \begin{bmatrix}
    1&1&0&0\\0&0&1&1
  \end{bmatrix},~~
  {\bf M}_2 = \begin{bmatrix}
    1&0&1&0\\0&1&0&0
  \end{bmatrix}.
\end{align}
Note that ${\bf M}_1$ is half-MDS while ${\bf M}_2$ is not. The submatrix comprised of the $2^{nd}$ and $4^{th}$ columns of ${\bf M}_2$ has rank $1<2$.
\end{definition}

\begin{lemma}[Half-MDS $N$-sum box] \label{lem:halfMDS}
  If $q \geq N$, then there exists an $N$-sum box with transfer matrix ${\bf M}\in \mathbb{F}_q^{N\times 2N}$ that is half-MDS.
\end{lemma}
\noindent The proof of Lemma \ref{lem:halfMDS} is presented in Appendix \ref{sec:proof_halfMDS}.

Recall that we are free to choose $z$. Therefore, if we choose $z > \log_d (\max_{t\in [T]} N_t) \implies q = d^z > \max_{t\in [T]} N_t$, then $\overline{\bf M}_t$ can be made half-MDS for all $t\in [T]$. Now that ${\bf M}_t$ is half-MDS, the rank of ${\bf M}_{t,k}$ is equal to $\min(N_t, \sum_{s \in \mathcal{E}(t)\cap \mathcal{W}(k)}2N_{t,s})$. Due to the diagonal block structure of $\overline{\bf M}_k$, the rank of $\overline{\bf M}_k$ is equal to $\sum_{t\in[T]}\min\left(N_t, \sum_{s\in \mathcal{E}(t)\cap \mathcal{W}(k)}2N_{t,s}\right)$. Plugging this value into \eqref{eq:feasible_set_1} we further obtain that the following set of tuples are feasible,
\begin{align}
\mbox{closure}&\left\{ {\bm \Delta} \in \mathbb{R}_+^\Gamma \left|
  \begin{array}{l}
  	  N_{t,s}\in \mathbb{N}, \forall t\in[T], s\in\mathcal{E}(t),\\
  	  R \in \mathbb{N},\\
  	  \sum_{t\in[T]}\min\left(N_t, \sum_{s\in \mathcal{E}(t)\cap \mathcal{W}(k)}2N_{t,s}\right) \geq R,~ \forall k\in[K], \\
  	  \Delta_{t,s} \geq N_{t,s}/R, t\in [T], s\in \mathcal{E}(t).
  \end{array}
  \right. \right\} \\
  = &\left\{ {\bm \Delta} \in \mathbb{R}_+^\Gamma \left|
  \begin{array}{l}
   	  \sum_{t\in[T]}\min\left(\sum_{s\in\mathcal{E}(t)}\Delta_{t,s}, \sum_{s\in \mathcal{E}(t)\cap \mathcal{W}(k)}2\Delta_{t,s}\right) \geq 1,~ \forall k\in[K]
  \end{array}
  \right. \right\},
\end{align}
which is the same as the region $\mathcal{D}$ specified in Theorem \ref{thm:main}.

\begin{remark}
We remark that the initial quantum state that is used when constructing a half-MDS $N$-sum box is an absolutely maximally entangled (AME) state \cite{helwig2012absolute}. Specifically, say a quantum system $\mathcal{Q}$ of $N$ qudits is in a pure state $\ket{\psi}$. Then the state is absolutely maximally entangled if for any partition of the system into $\mathcal{Q}_A$ with $k \leq \lfloor N/2 \rfloor$-qudits and $\mathcal{Q}_B$ with the remaining $N-k$ qudits,  we have $S(\rho^A) = S(\rho^B) = k$ (dits), where $\rho^A$ and $\rho^B$ denote the reduced density matrices of $\mathcal{Q}_A$ and $\mathcal{Q}_B$, respectively. $S(\cdot)$ denotes the von Neumann entropy. According to  \cite{Allaix_N_sum_box}, the generator matrix of the stabilizer code when constructing the box with the transfer matrix ${\bf M} = [{\bf M}_l, {\bf M}_r]$ is ${\bf G} = \ppsmatrix{-{\bf M}_r^T \\ {\bf M}_l^T}$. Denote the initial state for this maximal stabilizer code as $\ket{\bf G}$. \cite{bahramgiri2006graph} shows that $\ket{\bf G}$ is equivalent to a graph state under local Clifford operations. \cite{huber2013structure} then shows that the graph state is maximally entangled if ${\bf G}^T$ is the generator matrix of an $[[n,k,d]]$ quantum MDS (stabilizer) code with $n=2k$. This is satisfied because ${\bf M}$ is half-MDS, and so is ${\bf G}^T$. Since local Clifford operations do not affect the property of AME, we conclude that the initial state corresponding to the box ${\bf M}$ is an AME state.
\end{remark}

\section{Conclusion}
A sharp capacity characterization for the $\Sigma$-QMAC bodes well for future generalizations, that include in particular, the Linear Computation QMAC (LC-QMAC). As the quantum extension of the LC-MAC, which is the counterpart of the LCBC (linear computation broadcast) problem studied in \cite{Sun_Jafar_CBC, Yao_Jafar_KLCBC}, the LC-QMAC assumes $S$ servers, $K$ data streams of $\mathbb{F}_d$ symbols, and a user (Alice) who wants to compute an \emph{arbitrary $\mathbb{F}_d$ linear function} of the data streams. For example, with data streams ${\sf W}_1,{\sf W}_2,\cdots, {\sf W}_K$ represented as vectors over $\mathbb{F}_d$, Alice wants to compute ${\sf F}={V}_1{\sf W}_1+{V}_2{\sf W}_2+\cdots+{V}_K{\sf W}_K$ for arbitrary linear transformations (matrices) $V_1, V_2, \cdots, V_K$ that are specified by the problem. Note that if $V_1, V_2, \cdots, V_K$ are invertible square matrices then the problem reduces to the $\Sigma$-QMAC, whose capacity is found in this work. This is because without loss of generality each $V_i{\sf W}_i$ can be defined to be a data stream $\widetilde{\sf W}_i$ over an extension field, leaving Alice only with the task of computing the sum, i.e., the $\Sigma$-QMAC setting. 
The general LC-QMAC setting, however, allows arbitrary matrices $V_1, V_2, \cdots, V_K$. In addition, the LC-QMAC specification includes arbitrary side-information at Alice of the form ${\sf F}'={V}_1'{\sf W}_1+{V}_2'{\sf W}_2+\cdots+{V}_K'{\sf W}_K$, which can be quite useful for improving the communication efficiency of linear computation. Furthermore, the LC-QMAC allows the data available to each server to be arbitrary linear functions of the data streams, i.e., Server $s$ has data of the form ${U}_{1s}'{\sf W}_1+{U}_{2s}'{\sf W}_2+\cdots+{U}_{Ks}'{\sf W}_K$, which adds another layer of both conceptual and combinatorial complexity. Indeed the capacity of even the classical LCMAC setting is not yet fully known. For example, consider an LC-QMAC setting over $\mathbb{F}_3$, with $K=2$ data streams ${\sf A},{\sf B}$, and $3$ servers that have $({\sf A}), ({\sf B}), ({\sf A+2B})$ respectively. Say Alice has no side-information and only wants  to compute ${\sf A}+{\sf B}$.  The  capacity of this LC-MAC is not known to the best of our knowledge, but  the corresponding quantum setting is trivial, i.e., the rate $R = 1$ qudit/dit is achieved simply if any two servers apply the $2$-sum protocol. Similarly, there are QPIR settings where the capacity is known, while the corresponding classical cases remains open \cite{song_multiple_server_PIR, song_colluding_PIR, QMDSTPIR, song_all_but_one_collusion}. Thus, quantum settings can be tractable even when their classical counterparts are not. The sufficiency of the $N$-sum box abstraction for the LC-QMAC is an especially intriguing question. Aside from the LC-QMAC, generalizations in other directions, e.g., towards noisy quantum channels and correlated inputs as in \cite{sohail2022unified, sohail2022computing, hayashi2021computation},  other forms of decoding locality restrictions as in \cite{hayashilocal}, and to non-linear computations as in \cite{christensen2023private,Lu_Yao_Jafar_Prod} are also of interest. 

\section*{Acknowledgment}
The authors gratefully acknowledge helpful discussions with Matteo Allaix from Aalto University and Yuxiang Lu from University of California Irvine.

\appendix

\section{Arbitrarily large DSC Gain over QMAC}\label{app:largeDSC}
Here we show that the DSC gain in the QMAC for certain partial function computations can be arbitrarily large by providing an example. The construction of this example largely relies on a bound of the chromatic number of the power of a family of graphs (referred to as the quarter-orthogonality graphs) presented in \cite{briet2014entanglement} that is based on earlier results in zero-error information theory such as \cite{Witsenhausen, alon1998shannon}, together with insights from quantum communication complexity literature such as \cite{brassard1999cost, PSQM}.

Let $\kappa = 4p^r$ for an odd prime $p$ and a positive integer $r$. For two vectors $v_1,v_2$ with the same length, let $h(v_1,v_2)$ denote the Hamming distance between them, i.e., the number of positions where their elements are distinct. Let ${\sf A}, {\sf B}$ be data streams with  realizations in $\{+1,-1\}^\kappa$ such that $h({\sf A}, {\sf B}) \in \{0, \kappa/2\}$. There are $2$ servers in the QMAC. Server $A$ knows only ${\sf A}$ and Server $B$ knows only ${\sf B}$. The function to compute at Alice is defined as ${\sf F} = h({\sf A}, {\sf B})$. Note that due to the dependence between ${\sf A}$ and ${\sf B}$, the function to compute has binary output.

Let us first consider the capacity of the fully-unentangled case, $C^{\unent}$. Denote by $\mathcal{Q}_A$ a quantum system of dimension $\delta_A$ that is sent from Server $A$ and by $\mathcal{Q}_B$ a quantum system of dimension $\delta_B$ that is sent from Server $B$. Since our QMAC formulation does not allow the POVMs to depend on the data, and since there is no entanglement established between the two servers, without loss of generality, consider that Server $A$ sends to Alice a state that is picked from a set of orthogonal states $\{\ket{a}_A\}_{a\in [\delta_A]}$ and that Server $B$ sends to Alice a state from a set of orthogonal states $\{\ket{b}_B\}_{b\in [\delta_B]}$, which let Alice to perfectly recover $(a,b)$. Suppose a genie provided Alice with ${\sf A}^{[L]}$. She would still need to recover ${\sf F}^{[L]}$ from $(b,{\sf A}^{[L]})$. Let $G_{\kappa}(V,E)$ be the \emph{orthogonality graph} with vertices $V$ uniquely mapping to $\{+1,-1\}^{\kappa}$ and for $v_1\not=v_2 \in V$, $(v_1,v_2) \in E$ if and only if $h(v_1,v_2) = \kappa/2$. Let $H_{\kappa-1}(V',E')$ be the \emph{quarter-orthogonality graph} with vertices $V'$ uniquely mapping to the vectors in $\{+1,-1\}^{\kappa-1}$ that have an even number of ``$-1$" entries, and for $v_1'\not=v_2' \in V'$, $(v_1',v_2') \in E'$ if and only if $h(v_1',v_2') = \kappa/2$. The quarter-orthogonality graph $H_{\kappa-1}$ is a subgraph of the orthogonality graph $G_{\kappa}$ \cite{briet2014entanglement}.
  Similar to the reasoning in \cite{Witsenhausen}, $\delta_B \geq \chi(G_{\kappa}^{\boxtimes L})$ where $\boxtimes$ denotes the strong product of graphs\footnote{For graphs $G$ and $H$ with respective vertex sets $V(G)$ and $V(H)$, define $G\boxtimes H$ as the strong product of $G$ and $H$ such that the vertex set of $G\boxtimes H$ is the Cartesian product $V(G)\times V(H)$; and distinct vertices $(u,u')$ and $(v,v')$ are adjacent in $G\boxtimes H$ if and only if: $u=v$ and $u'$ is adjacent to $v'$, or $u'=v'$ and $u$ is adjacent to $v$, or $u$ is adjacent to $v$ and $u'$ is adjacent to $v'$. $G^{\boxtimes L}$ is then defined as $\underbrace{G\boxtimes G\boxtimes \cdots \boxtimes G}_L$.} and $\chi(\cdot)$ denotes the chromatic number. 
  \cite[Cor. 5.8]{briet2014entanglement} shows that $\chi(H_{\kappa-1}^{\boxtimes L}) \geq 2^{(0.154\kappa-1.154)L}$. Since $H_{\kappa-1}$ is a subgraph of $G_{\kappa}$, $H_{\kappa-1}^{\boxtimes L}$ is a subgraph of $G_{\kappa}^{\boxtimes L}$. It follows that $\chi(G_{\kappa}^{\boxtimes L}) \geq 2^{(0.154\kappa-1.154)L}$ as the chromatic number of a graph cannot be less than the chromatic number of its subgraph. Thus, we obtain that $\log_2\delta_B/L \geq 0.154\kappa-1.154$. Due to symmetry between the two servers, we obtain that $C^{\unent} \leq \sup_{L\to \infty} \frac{L}{\log_2\delta_A + \log_2\delta_B} \leq \frac{1/2}{0.154\kappa-1.154}$ (computations per qubit).

Next, we show that the fully-entangled capacity $C^{\fullent} \geq \frac{1/2}{\log_2 \kappa}$. Thus, the DSC gain is at least $\frac{0.154\kappa -1.154}{\log_2 \kappa}$, which can be made arbitrary large by choosing $\kappa$ large enough. The scheme that achieves $\frac{1/2}{\log_2 \kappa}$ only needs batch size $L=1$. The scheme is a generalization of the scheme in the problem referred to as the distributed Deutsch-Jozsa problem \cite{brassard1999cost, PSQM} where $p=2$ (but here we need $p$ to be an odd prime). Let Server $A$ and Server $B$ share an entangled state $\ket{\mbox{GHZ}} = \frac{1}{\sqrt{\kappa}}\sum_{x\in [\kappa]}\ket{xx}$. 
Let $U_{\sf A} = \mbox{diag}({\sf A})$ be a $\kappa \times \kappa$ diagonal matrix with the elements of vector ${\sf A}$ on the main diagonal. Similarly let $U_{\sf B} = \mbox{diag}({\sf B})$. Note that $U_{\sf A}$ and $U_{\sf B}$ are unitary matrices. Let Server $A$ apply the unitary operator $U_{\sf A}$ and Server $B$ apply the unitary operator $U_{\sf B}$ to their respective quantum subsystems. Then the resulting state is $(U_{\sf A} \otimes U_{\sf B})\ket{\mbox{GHZ}} = \frac{1}{\sqrt{\kappa}}\mbox{vec}(U_{\sf B}^TU_{\sf A})$, where $\mbox{vec}(\cdot)$ denotes the column-major vectorization function\footnote{$\mbox{vec}(A) \triangleq [a_{1,1},\cdots, a_{m,1},a_{1,2},\cdots, a_{m,2},\cdots,a_{1,n},\cdots, a_{m,n}]^T$, where $a_{i,j}$  represents the element in the $i^{th}$ row and $j^{th}$ column of $A$.}. Note that we used the identity $\mbox{vec}(ABC) = (C^T\otimes A)\mbox{vec}(B)$ \cite{hardy2019matrix} along with the fact that $\ket{\mbox{GHZ}} = \frac{1}{\sqrt{\kappa}}\mbox{vec}({\bf I}_{\kappa})$. Alice measures the quantum system by a PVM with two projectors $P_1 = \ket{\mbox{GHZ}}\bra{\mbox{GHZ}}$ and $P_2 = {\bf I}_{\kappa} - P_1$. The measurement result being $1$ (associated with $P_1$) has probability $\tr(U_{\sf B}^T U_{\sf A})/\kappa = {\sf B}^T{\sf A}/\kappa$, which is equal to $1$ if $h({\sf A},{\sf B}) = 0$, and equal to $0$ if $h({\sf A},{\sf B}) = \kappa/2$. This means that Alice is able to distinguish the two possibilities of $h({\sf A},{\sf B})$ with certainty.

\section{Proof of Corollary \ref{cor:cqcc}} \label{proof:cqcc}
Let
{\small $\mathcal{D}^{\unent}_{1/2} \triangleq \{(\Delta_1/2,\cdots,\Delta_S/2) \mid (\Delta_1,\cdots, \Delta_S) \in \mathcal{D}^{\unent}\}$} and
{\small $\mathcal{D}_1 \triangleq \left \{(\Delta_1,\cdots, \Delta_S) \mid \sum_{s\in[S]} \Delta_s \geq 1 \right \}$}.  
Then Corollaries \ref{cor:unconstrained} and \ref{cor:unentangled} together imply that
\begin{align}
\mathcal{D}^{\fullent} = \mathcal{D}^{\unent}_{1/2} \cap \mathcal{D}_1.
\end{align}
Let
{\small  $(\Delta_1^*, \cdots, \Delta_S^*)$ be a solution of $\footnotesize \arg\min_{(\Delta_1,\cdots,\Delta_S) \in \mathcal{D}^{\unent}} \sum_{s\in[S]} \Delta_s$}.
It follows that 
{\small $(\Delta_1^*/2, \cdots, \Delta_S^*/2)$} is a solution of $\arg\min_{(\Delta_1,\cdots,\Delta_S) \in \mathcal{D}^{\unent}_{1/2}} \sum_{s\in[S]} \Delta_s$. Consider two cases.
\begin{enumerate}
	\item If {\small $\sum_{s\in[S]} \Delta^*_s/2 \geq 1$}, i.e., {\small $C^{\unent} \leq 1/2$}, then {\small $(\Delta_1^*/2, \cdots, \Delta_S^*/2) \in \mathcal{D}_1$ and thus $(\Delta_1^*/2, \cdots, \Delta_S^*/2) \in \mathcal{D}^{\fullent}$}. It follows that {\small $(\Delta_1^*/2, \cdots, \Delta_S^*/2)$} is a solution of $\arg\min_{(\Delta_1,\cdots,\Delta_S) \in \mathcal{D}^{\fullent}} \sum_{s\in[S]} \Delta_s$. This implies that {\small $C^{\fullent} = ({\sum_{s\in[S]}\Delta_s^*/2})^{-1} = 2 ({\sum_{s\in[S]}\Delta_s^*})^{-1} =  2C^{\unent}$}.
	\item Otherwise, if {\small $\sum_{s\in[S]} \Delta^*_s/2 < 1$}, i.e., {\small $C^{\unent} > 1/2$}, there exists {\small $(\Delta_1', \cdots,\Delta_S')$} such that {\small $\Delta_s'\geq \Delta_s^*/2,\forall s\in [S]$} and {\small $\sum_{s\in[S]}\Delta_s' = 1$}. Note that {\small $(\Delta_1',\cdots,\Delta_S') \in \mathcal{D}^{\unent}_{1/2}$} by the definition of $\mathcal{D}^{\unent}$ (Corollary \ref{cor:unentangled}) and the definition of $\mathcal{D}^{\unent}_{1/2}$. Since {\small $\sum_{s\in[S]}\Delta_s' = 1$}, we have {\small $(\Delta_1',\cdots,\Delta_S')\in \mathcal{D}_1$} and therefore {\small $(\Delta_1',\cdots,\Delta_S')\in \mathcal{D}^{\fullent}$}. This implies that {\small $C^{\fullent} \geq 1$} (and thus {\small $C^{\fullent} = 1$} as {\small $C^{\fullent}$} is also upper bounded by $1$). 
\end{enumerate}
Combining the two cases, we have {\small $C^{\fullent} = \min(1,2C^{\unent})$} and thus {\small $C^{\fullent}/C^{\unent} = \min(2,1/C^{\unent})$}.

\section{Proof of Corollary \ref{cor:symmetric}} \label{proof:symmetric}
\subsection{Proof of \eqref{eq:symmetric}}
Recall that for the symmetric settings, $K = \binom{S}{\alpha}, T= \binom{S}{\beta}$ and $\mathcal{W}: [K] \biject \binom{[S]}{\alpha}$, $\mathcal{E}: [T]\biject \binom{[S]}{\beta}$. We want to find
\begin{align}
	F^* \triangleq \min_{{\bm \Delta} \in \mathcal{D}} \sum_{t\in[T]}\sum_{s\in\mathcal{E}(t)} \Delta_{t,s},
\end{align}
where the feasible region $\mathcal{D}$ here is determined by Theorem \ref{thm:main} for the symmetric data replication and entanglement distribution maps. $C_{\alpha}^{(\beta)}$ immediately follows as $1/F^*$. Due to symmetry, the minimal value of $\sum_{t\in[T]}\sum_{s\in \mathcal{E}(t)}\Delta_{t,s}$ is achieved by $\Delta_{t,s} = \Delta_o \in\mathbb{R}_+, \forall t\in[T],s\in\mathcal{E}(t)$. Again by symmetry, the $K$ conditions in $\mathcal{D}$ are identical in the form $f(\alpha,\beta) \Delta_o \geq 1$, where $f(\alpha,\beta) \in \mathbb{Z}^+$ is a function of $(\alpha,\beta)$. Next we derive $f(\alpha,\beta)$. Consider the value $k$ such that $\mathcal{W}(k) = \{1,2,\cdots, \alpha\} = [\alpha]$. We have
\begin{align}
	f(\alpha,\beta) &= \sum_{t\in[T]} \min \big(|\mathcal{E}(t)|, ~2|\mathcal{E}(t)\cap [\alpha]|\big) \\
	&= \sum_{\mathcal{B} \in \binom{[S]}{\beta}} \min \big( \beta, ~2|\mathcal{B} \cap [\alpha]| \big) \\
	&= \sum_{\gamma = (\alpha+\beta-S)^+}^{\min(\alpha,\beta)} \min (\beta,~ 2\gamma) \cdot \Big|\Big\{\mathcal{B} \in \mbox{$\binom{[S]}{\beta}$} ~\big|~  |\mathcal{B}\cap [\alpha] | = \gamma  \Big\}\Big| \\
	&= \sum_{\gamma = (\alpha+\beta-S)^+}^{\min(\alpha,\beta)} \min (\beta,~ 2\gamma) \cdot \binom{\alpha}{\gamma}\cdot \binom{S-\alpha}{\beta-\gamma}.
\end{align}
It follows that $F^* = \frac{\beta T}{f(\alpha,\beta)}$ and therefore $C_{\alpha}^{(\beta)} = 1/F^* = \frac{f(\alpha,\beta)}{\beta T}$.

\subsection{Proof of \eqref{eq:symmetric_equiv_1}}
We can rewrite $C_{\alpha}^{(\beta)}$ as
{\small
\begin{align}
	C_{\alpha}^{(\beta)} &= \frac{1}{\beta T} \sum_{\gamma=(\alpha+\beta-S)^+}^{\min(\alpha, \beta)} \big( 2\gamma - (2\gamma-\beta)^+ \big) \cdot \binom{\alpha}{\gamma}\cdot \binom{S-\alpha}{\beta-\gamma} \\
	&= \frac{2}{\beta T} \sum_{\gamma=(\alpha+\beta-S)^+}^{\min(\alpha, \beta)} \gamma \cdot \binom{\alpha}{\gamma}\cdot \binom{S-\alpha}{\beta-\gamma} - \frac{1}{\beta T} \sum_{\gamma=(\alpha+\beta-S)^+}^{\min(\alpha, \beta)} (2\gamma-\beta)^+ \cdot \binom{\alpha}{\gamma}\cdot \binom{S-\alpha}{\beta-\gamma} \\
	&= \frac{2\alpha}{\beta T} \sum_{\gamma=(\alpha+\beta-S)^+}^{\min(\alpha, \beta)} \binom{\alpha-1}{\gamma-1}\cdot \binom{S-\alpha}{\beta-\gamma} - \frac{1}{\beta T} \sum_{\gamma=(\alpha+\beta-S)^+}^{\min(\alpha, \beta)} (2\gamma-\beta)^+ \cdot \binom{\alpha}{\gamma}\cdot \binom{S-\alpha}{\beta-\gamma} \\
	&= \frac{2\alpha}{\beta T} \cdot \binom{S-1}{\beta -1} - \frac{1}{\beta T} \sum_{\gamma=\max(\alpha+\beta-S, \lceil \beta/2 \rceil)}^{\min(\alpha, \beta)} (2\gamma-\beta) \cdot \binom{\alpha}{\gamma}\cdot \binom{S-\alpha}{\beta-\gamma} \\
	&= \frac{2\alpha}{S} - \frac{1}{\beta T} \sum_{\gamma=\max(\alpha+\beta-S, \lceil \beta/2 \rceil)}^{\min(\alpha, \beta)} (2\gamma-\beta) \cdot \binom{\alpha}{\gamma}\cdot \binom{S-\alpha}{\beta-\gamma}
\end{align}
}

\subsection{Proof of \eqref{eq:symmetric_equiv_2}}
Alternatively, we can rewrite $C_{\alpha}^{(\beta)}$ as
{\small
\begin{align}
	C_{\alpha}^{(\beta)} &= \frac{1}{\beta T}\sum_{\gamma=(\alpha+\beta-S)^+}^{\min(\alpha, \beta)} \big( \beta - (\beta-2\gamma)^+ \big) \cdot \binom{\alpha}{\gamma}\cdot \binom{S-\alpha}{\beta-\gamma} \\
	&=  \frac{1}{T} \sum_{\gamma=(\alpha+\beta-S)^+}^{\min(\alpha, \beta)}\binom{\alpha}{\gamma}\cdot \binom{S-\alpha}{\beta-\gamma} - \frac{1}{\beta T}\sum_{\gamma=(\alpha+\beta-S)^+}^{\min(\alpha, \beta)}(\beta-2\gamma)^+ \cdot \binom{\alpha}{\gamma}\cdot \binom{S-\alpha}{\beta-\gamma} \\
	&= \frac{1}{T} \binom{S}{\beta} - \frac{1}{\beta T} \sum_{\gamma = (\alpha+\beta-S)^+}^{\min(\alpha,\lfloor \beta/2 \rfloor)} (\beta -2\gamma) \cdot \binom{\alpha}{\gamma}\cdot \binom{S-\alpha}{\beta-\gamma} \\
	&= 1 - \frac{1}{\beta T} \sum_{\gamma = (\alpha+\beta-S)^+}^{\min(\alpha,\lfloor \beta/2 \rfloor)} (\beta -2\gamma) \cdot \binom{\alpha}{\gamma}\cdot \binom{S-\alpha}{\beta-\gamma}
\end{align}
}

\subsection{Proof of \eqref{eq:beta_star}}
The case when $\alpha = S$ is trivial. For $\alpha \leq \lfloor S/2 \rfloor$, let us make use of \eqref{eq:symmetric_equiv_1}. It is not difficult to obtain that $C_{\alpha}^{(\beta)} = \frac{2\alpha}{S}$ when $\beta \geq 2\alpha$, as $0 \leq 2\gamma -\beta \leq 2\alpha -\beta \leq 0$ for $\lceil \beta/2 \rceil \leq \gamma \leq \alpha$, so either $2\gamma -\beta = 0$ or $\gamma$ does not take any value in the summation term. On the other hand, if $\beta < 2\alpha$, we have $\alpha \geq \lceil \beta/2 \rceil$ and thus $\min(\alpha, \beta) \geq \max(\alpha+\beta-S, \lceil\beta/2\rceil)$, so $\gamma$ must take at least one value in the summation. It follows that $C_{\alpha}^{(\beta)} < \frac{2\alpha}{S}$ because $2\min(\alpha, \beta) - \beta > 0$. Since $C_{\alpha}^{(S)} = \frac{2\alpha}{S}$ for $\alpha \leq \lfloor S/2 \rfloor$, we obtain that $\beta^* = 2\alpha$ for these cases.

For $\lceil S/2 \rceil \leq \alpha \leq S-1$, let us make use of \eqref{eq:symmetric_equiv_2}. It is not difficult to obtain that $C_{\alpha}^{(\beta)} = 1$ when $\beta \geq 2(S-\alpha)$, as $0 \leq \beta-2\gamma \leq 2(S-\alpha) -\beta \leq 0$ for $\lfloor \beta /2 \rfloor \leq \gamma \leq 2(S-\alpha)$, so either $\beta -2\gamma =0$ or $\gamma$ does not take any value in the summation term. On the other hand, if $\beta < 2(S-\alpha)$, we have $\lceil \beta/2 \rceil \leq S-\alpha \implies \lfloor \beta/2 \rfloor \geq \alpha +\beta -S$ and thus $\min(\alpha,\lfloor \beta/2 \rfloor) \geq (\alpha+\beta-S)^+$, so $\gamma$ must take at least one value in the summation. It follows that $C_{\alpha}^{(\beta)} < 1$ because $\beta - 2(\alpha+\beta-S)^+ > 0$. Since $C_{\alpha}^{(S)} = 1$ for $\lceil S/2 \rceil \leq \alpha \leq S-1$, we obtain that $\beta^* = 2(S-\alpha)$ for these cases.

\section{Proof of Corollary \ref{cor:bipartite}}\label{proof:2sum}
First let us note that there are two $\Sigma$-QMAC problems involved in the corollary, summarized as follows.
\begin{enumerate}
	\item The original problem $\mathcal{P}$ has data replication map $\mathcal{W}$, $S$ servers and $K$ data streams. We refer to the $S$ servers in $\mathcal{P}$ by $\mathcal{S}_i$ for $i\in [S]$. We are interested in the $2$-party-entangled capacity $C^{(2)}(\mathcal{W})$.
	\item The hypothetical problem $\widetilde{\mathcal{P}}$ has data replication map $\widetilde{\mathcal{W}}$, $\binom{S}{2}$ servers and the same $K$ data streams as in $\mathcal{P}$. We refer to the $\binom{S}{2}$ servers in $\widetilde{\mathcal{P}}$ by $\mathcal{S}_{\{i,j\}}$ for $\{i,j\} \in \binom{[S]}{2}$. $\mathcal{S}_{\{i,j\}}$ has the access to the data streams that are available to either $\mathcal{S}_i$ or $\mathcal{S}_j$ in $\mathcal{P}$.  We are interested in the fully-unentangled capacity $C^{\unent}(\widetilde{\mathcal{W}})$.
\end{enumerate}

Our goal is to prove that $C^{(2)}(\mathcal{W}) = C^{\unent}(\widetilde{\mathcal{W}})$.
The proposition is comprised of two bounds, $C^{(2)}(\mathcal{W}) \leq \widetilde{C}^o(\widetilde{\mathcal{W}})$ and $C^{(2)}(\mathcal{W}) \geq C^{\unent}(\widetilde{\mathcal{W}})$.
To prove $C^{(2)}(\mathcal{W}) \leq C^{\unent}(\widetilde{\mathcal{W}})$, consider any rate $R$ that is achievable in the problem $\mathcal{P}$ with only bipartite entanglement, i.e., any clique contains at most $2$ servers. The output state corresponds to any clique $\{i,j\} \in \binom{[S]}{2}$ of this scheme can always be generated by the server $\mathcal{S}_{\{i,j\}}$ in the problem $\widetilde{\mathcal{P}}$, since this server has the access to all data streams that are available to either $\mathcal{S}_i$ or $\mathcal{S}_j$ in the problem $\mathcal{P}$. In the problem $\widetilde{\mathcal{P}}$, a scheme can let $\mathcal{S}_{\{i,j\}}$ directly transmit this state to Alice. Therefore, $R$ must be achievable in $\widetilde{\mathcal{P}}$ as well, which shows that $C^{(2)}(\mathcal{W}) \leq C^{\unent}(\widetilde{\mathcal{W}})$.

To prove $C^{(2)}(\mathcal{W}) \geq C^{\unent}(\widetilde{\mathcal{W}})$, we need an intermediate result of Theorem \ref{thm:main} that coding based on the $N$-sum box abstraction is optimal in every case. For the unentangled case, this means that each server is simply treating qudits as classical dits. Additionlly, Lemma \ref{lem:multicast} implies that the optimal scheme is also linear.
With this knowledge, given that $R$ is achievable in the problem $\widetilde{\mathcal{P}}$, let us consider a linear scheme that achieves this $R$, where the coders $V_{\{i,j\}}\in \mathbb{F}_d^{N_{\{i,j\}}\times L}$ at the server $\mathcal{S}_{\{i,j\}}$ takes as input $L$ symbols of all its accessible data streams, i.e., ${\sf W}_k, k\in \widetilde{\mathcal{W}}(\{i,j\})$, and outputs a vector $Y_{\{i,j\}} \in \mathbb{F}_d^{N_{\{i,j\}} \times 1}$. In this linear scheme, symbols are considered in the field $\mathbb{F}_d$. Say $\mathcal{S}_{\{i,j\}}$ transmits $N_{\{i,j\}}$ qudits to inform Alice of $Y_{\{i,j\}}$. Upon receiving the $\sum_{\{i,j\}\in \binom{[S]}{2}}N_{\{i,j\}} \triangleq N$ qudits, Alice computes $L$ instances of sum. The scheme satisfies that $\frac{L}{N} \geq R$ by definition. What we will do next is to convert this scheme to a scheme in the problem $\mathcal{P}$ with the use of $2$-sum protocols that achieves the same rate. Recall that the $2$-sum protocol behaves the same as an $\mathbb{F}_d$ additive channel, in the way that the receiver is able to get one sum of the two $\mathbb{F}_d$ inputs from the two transmitters with the cost of one qudit on average. We also point out that each use of the $2$-sum protocol is equivalent to use such an additive channel twice, i.e., the receiver gets two dimensions of the sums at a cost of $2$ qudits.

Now let us look at the original problem $\mathcal{P}$. Let us construct a scheme with batch size $2L$, so that Alice is able to compute $2L$ instances of the sum by using the $2$-sum protocol $N$ times, at a cost of $2N$ qudits,  thus achieving the same rate as in the problem $\mathcal{\widetilde{P}}$.  Recall that $Y_{\{i,j\}}$ denotes  the transmission from the server $\mathcal{S}_{\{i,j\}}$ in the problem $\widetilde{\mathcal{P}}$.  Since $Y_{\{i,j\}}\in \mathbb{F}_d^{N_{\{i,j\}} \times 1}$ is a linear function of the data streams that are known to either $\mathcal{S}_i$ or $\mathcal{S}_j$ in the problem $\mathcal{P}$, we can represent $Y_{\{i,j\}} = Y_{\{i,j\},i} + Y_{\{i,j\},j}$, where $Y_{\{i,j\},i} \in \mathbb{F}_d^{N_{\{i,j\}}\times 1}$ can be computed by $\mathcal{S}_i$, and $Y_{\{i,j\},j} \in \mathbb{F}_d^{N_{\{i,j\}}\times 1}$ can be computed by Server $\mathcal{S}_j$ in the problem $\mathcal{P}$. Therefore, with $N_{\{i,j\}}$ use of the binary additive channel in $\mathbb{F}_d$ between $\mathcal{S}_i$ and $\mathcal{S}_j$, the two servers can transmit the sum $Y_{\{i,j\}}$ to Alice. To apply the $2$-sum protocol, we only need to consider two parallel instances of $Y_{\{i,j\}}$, so that with $N_{\{i,j\}}$ uses of the $2$-sum protocol between $\mathcal{S}_i$ and $\mathcal{S}_j$,  Alice is able to get $2$ instances of $Y_{\{i,j\}}$. Taking all pairs of servers into account, for the converted scheme in the problem $\mathcal{P}$, the $2$-sum protocol (in $\mathbb{F}_d$) is used $\sum_{\{i,j\} \in \binom{[S]}{2}} N_{\{i,j\}} = N$ times in total, allowing Alice to compute $2N$ instances of the sum in $\mathbb{F}_d$.

\section{Proof of Corollary \ref{cor:tripartite}} \label{proof:tripartite}
We show that for any data replication map $\mathcal{W}$, and entanglement distribution map $\mathcal{E}$ where $|\mathcal{E}(t)| = 3$ for some $t$, the clique $\mathcal{E}(t)$ can be replaced by three bipartite cliques (that contain only two servers) without decreasing the capacity.  Formally, let $\mathcal{E}$ be such an entanglement distribution map and without loss of generality, $\mathcal{E}(1) = \{1,2,3\}$. Let $\widetilde{\mathcal{E}}$ be another entanglement distribution map such that $\widetilde{\mathcal{E}}(1) = \{2,3\}$, $\widetilde{\mathcal{E}}(2) = \{1,3\}$, $\widetilde{\mathcal{E}}(3) = \{1,2\}$, and $\widetilde{\mathcal{E}}(t) = \mathcal{E}(t-2)$ for $t\geq 4$.
Let $\mathcal{P}, \widetilde{\mathcal{P}}$ denote the respective $\Sigma$-QMAC problems with entanglement distribution maps $\mathcal{E}$, $\widetilde{\mathcal{E}}$, and both problems have data replication map $\mathcal{W}$. For the problem $\mathcal{P}$, we use $\Delta_{t,s}$ to denote the (normalized) download cost associated with Clique $\mathcal{E}(t)$ and Server $s\in \mathcal{E}(t)$, and we use ${\bm \Delta}$ to denote the download cost tuple.
To avoid confusion, in the problem $\widetilde{\mathcal{P}}$, we use $\widetilde{\Delta}_{t,s}$ to denote the download cost associated with Clique $\widetilde{\mathcal{E}}(t)$ and Server $s\in \widetilde{\mathcal{E}}(t)$, and we use $\widetilde{{\bm \Delta}}$ to denote the download cost tuple.
Since $\mathcal{E}(1)=\{1,2,3\}$ is our main focus, in the following we let $\Delta_{s} \triangleq \Delta_{1,s} ,s\in \{1,2,3\}$ for brevity. Our goal is to show that $C(\mathcal{W},\widetilde{\mathcal{E}}) \geq C(\mathcal{W},\mathcal{E})$ and thus $C(\mathcal{W},\widetilde{\mathcal{E}}) = C(\mathcal{W},\mathcal{E})$ because any coding scheme allowed in the problem $\widetilde{\mathcal{P}}$ is also allowed in the problem $\mathcal{P}$.

Consider any feasible download tuple ${\bm \Delta}$ (in the feasible region implied by Theorem \ref{thm:main}) for the problem $\mathcal{P}$. To focus on the clique $\mathcal{E}(1)$, Theorem \ref{thm:main} implies that ${\bm \Delta}$ is feasible if and only if,
\begin{align}\label{eq:3_way_feasibility}
	\begin{cases}
		\min\{\Delta_1+\Delta_2+\Delta_3, 2\Delta_1\} \geq c_1\\
		\min\{\Delta_1+\Delta_2+\Delta_3, 2\Delta_2\} \geq c_2\\
		\min\{\Delta_1+\Delta_2+\Delta_3, 2\Delta_3\} \geq c_3\\
		\min\{\Delta_1+\Delta_2+\Delta_3, 2\Delta_1+2\Delta_2\} \geq c_{12} \\
		\min\{\Delta_1+\Delta_2+\Delta_3, 2\Delta_1+2\Delta_3\} \geq c_{13} \\
		\min\{\Delta_1+\Delta_2+\Delta_3, 2\Delta_2+2\Delta_3\} \geq c_{23} \\
		\Delta_1+\Delta_2+\Delta_3 \geq c_{123}
	\end{cases}
\end{align}
where $c_1, c_2,\cdots, c_{123}$ are determined by $\mathcal{W}$ and $\big(\Delta_{t,s}\big)_{t\geq 2, s\in \mathcal{E}(t)}$, i.e., the (normalized) download costs associated with the other remaining cliques in $\mathcal{E}$.

Let us note that it is without loss of generality to consider such feasible tuples with $\Delta_{i} \leq \Delta_{j} + \Delta_{k}$ for $\{i,j,k\}\in \{\{1,2,3\},\{2,1,3\},\{3,1,2\}\}$ if we are only interested in their sum $\Delta_1+\Delta_2+\Delta_3$, because otherwise (say $\Delta_{1} > \Delta_{2} + \Delta_{3}$) we can let 
\begin{align}
	\begin{bmatrix}
		\Delta_1' \\ \Delta_2' \\ \Delta_3'
	\end{bmatrix}
	=
	\begin{bmatrix}
		(\Delta_1+\Delta_2+\Delta_3)/2 \\ (\Delta_1+\Delta_2-\Delta_3)/2 \\ \Delta_3
	\end{bmatrix}
\end{align}
so that \eqref{eq:3_way_feasibility} is also satisfied if we replace $(\Delta_1,\Delta_2,\Delta_3)$ with $(\Delta_1',\Delta_2',\Delta_3')$. Note that $\Delta_1'+\Delta_2'+\Delta_3' = \Delta_1+\Delta_2+\Delta_3$ but now $\Delta_1' = \Delta_2'+\Delta_3'$. 

Now let us study the problem $\widetilde{\mathcal{P}}$. Note that by definition, $\widetilde{\mathcal{E}}(t) = \mathcal{E}(t-2)$ for $t\geq 4$. Then by Theorem \ref{thm:main}, the download cost tuple $\widetilde{{\bm \Delta}}$ is feasible if and only if
\begin{align} \label{eq:2_way_feasibility}
	\begin{cases}
		\widetilde{\Delta}_{t,s} = \Delta_{t-2,s}, \forall t\geq 4, s\in \widetilde{\mathcal{E}}(t)\\
		\min\{\widetilde{\Delta}_{2,1}+\widetilde{\Delta}_{2,3}, 2\widetilde{\Delta}_{2,1}\} + \min\{\widetilde{\Delta}_{3,1}+\widetilde{\Delta}_{3,2}, 2\widetilde{\Delta}_{3,1}\} \geq c_1 \\
		\min\{\widetilde{\Delta}_{1,2}+\widetilde{\Delta}_{1,3}, 2\widetilde{\Delta}_{1,2}\} + \min\{\widetilde{\Delta}_{3,1}+\widetilde{\Delta}_{3,2}, 2\widetilde{\Delta}_{3,2}\} \geq c_2 \\
		\min\{\widetilde{\Delta}_{1,2}+\widetilde{\Delta}_{1,3}, 2\widetilde{\Delta}_{1,3}\} + \min\{\widetilde{\Delta}_{2,1}+\widetilde{\Delta}_{2,3}, 2\widetilde{\Delta}_{2,3}\} \geq c_3\\
		\min\{\widetilde{\Delta}_{1,2}+\widetilde{\Delta}_{1,3},2\widetilde{\Delta}_{1,2}\} + \min\{\widetilde{\Delta}_{2,1}+\widetilde{\Delta}_{2,3},2\widetilde{\Delta}_{2,1}\}+ \widetilde{\Delta}_{3,1}+\widetilde{\Delta}_{3,2} \geq c_{12} \\
		\min\{\widetilde{\Delta}_{1,2}+\widetilde{\Delta}_{1,3},2\widetilde{\Delta}_{1,3}\}+ \widetilde{\Delta}_{2,1}+\widetilde{\Delta}_{2,3} + \min\{\widetilde{\Delta}_{3,1} + \widetilde{\Delta}_{3,2}, 2\widetilde{\Delta}_{3,1}\} \geq c_{13} \\
		\widetilde{\Delta}_{1,2}+\widetilde{\Delta}_{1,3} + \min\{ \widetilde{\Delta}_{2,1}+\widetilde{\Delta}_{2,3}, 2\widetilde{\Delta}_{2,3} \} + \min\{\widetilde{\Delta}_{3,1} + \widetilde{\Delta}_{3,2}, 2\widetilde{\Delta}_{3,2}\} \geq c_{23} \\
		\widetilde{\Delta}_{1,2}+\widetilde{\Delta}_{1,3} + \widetilde{\Delta}_{2,1}+\widetilde{\Delta}_{2,3} + \widetilde{\Delta}_{3,1}+\widetilde{\Delta}_{3,2} \geq c_{123}
	\end{cases}
\end{align}
where $c_1,c_2,\cdots,c_{123}$ are the same as those in \eqref{eq:3_way_feasibility}. Now, consider the download cost tuple $\widetilde{{\bm \Delta}}$ in the problem $\widetilde{\mathcal{P}}$, such that 
\begin{align} 
	 &\widetilde{\Delta}_{t,s} = \Delta_{t-2,s}, \forall t\geq 4, s\in \widetilde{\mathcal{E}}(t) \notag \\
	 \mbox{and~~}  &\begin{bmatrix}
		\widetilde{\Delta}_{1,2} \\ \widetilde{\Delta}_{2,1} \\ \widetilde{\Delta}_{3,1}
	\end{bmatrix}
	=
	\begin{bmatrix}
		\widetilde{\Delta}_{1,3} \\ \widetilde{\Delta}_{2,3} \\ \widetilde{\Delta}_{3,2}
	\end{bmatrix}
	=
	\begin{bmatrix}
		(\Delta_2+\Delta_3-\Delta_1)/2 \\ (\Delta_1+\Delta_3-\Delta_2)/2 \\ (\Delta_1+\Delta_2-\Delta_3)/2
	\end{bmatrix}. \label{eq:2_way_3_way}
\end{align}
Then it can be verified that the download cost tuple $\widetilde{{\bm \Delta}}$ is feasible in $\widetilde{\mathcal{P}}$ if ${\bm \Delta}$ is feasible in $\mathcal{P}$, because \eqref{eq:2_way_feasibility} is satisfied if \eqref{eq:3_way_feasibility} is satisfied. Therefore, the feasibility of ${\bm \Delta}$ in the problem $\mathcal{P}$ implies the feasibility of $\widetilde{{\bm \Delta}}$ in the problem $\widetilde{\mathcal{P}}$. Since $\widetilde{\Delta}_{1,2}+\widetilde{\Delta}_{1,3} + \widetilde{\Delta}_{2,1}+\widetilde{\Delta}_{2,3} + \widetilde{\Delta}_{3,1}+\widetilde{\Delta}_{3,2} = \Delta_1+\Delta_2+\Delta_3$ and $\sum_{t\geq 4,s\in \widetilde{\mathcal{E}}(t)} \widetilde{\Delta}_{t,s} = \sum_{t\geq 2,s\in \mathcal{E}(t)}\Delta_{t,s}$ by \eqref{eq:2_way_3_way}, ${\bm \Delta}$ and $\tilde{{\bm \Delta}}$ have the same (normalized) sum download costs.
Since it holds for any feasible download cost tuple ${\bm \Delta}$ in the problem $\mathcal{P}$, it follows that  $C(\mathcal{W},\mathcal{E}') \geq C(\mathcal{W},\mathcal{E})$.

\section{Proof of Corollary \ref{cor:S_partite_necessary}}\label{proof:asymmetry}
When $S$ is even, this corollary can be easily shown with the symmetric data replication maps as defined in Corollary \ref{cor:symmetric}. Consider the $\Sigma$-QMAC with the symmetric data replication maps with $\alpha=S/2$. Then \eqref{eq:beta_star} says that $\beta^* = S \implies C_{S/2}^{(S)}>C_{S/2}^{(S-1)}$.

When $S$ is odd, let us consider the data replication map $\mathcal{W} = \binom{[S-1]}{S-2} \cup \{S\}$. Without loss of generality, say $\mathcal{W}(S) = \{S\}$. We first apply Corollary \ref{cor:unconstrained} to show that $C^{\fullent}(\mathcal{W}) \geq \frac{2S-4}{2S-3}$. Let $\Delta_s$ be the normalized download cost from Server $s$ for $s\in [S]$. 
Writing down the feasible region by Corollary \ref{cor:unconstrained} explicitly for this setting, we have
\begin{align} 
    \mathcal{D}^{\fullent}(\mathcal{W}) = \left\{ (\Delta_1,\cdots,\Delta_S) \in \mathbb{R}_+^S \left|
      \begin{array}{l}
	\Delta_1 + \Delta_2 +\cdots + \Delta_S \geq 1 \\
	2(\Delta_1 + \Delta_2 + \cdots + \Delta_{S-1} + \Delta_{S-2}) \geq 1 \\
	2(\Delta_1 + \Delta_2 + \cdots + \Delta_{S-3} + \Delta_{S-1}) \geq 1 \\
	 ~~\vdots \notag \\
	 2(\Delta_2 + \Delta_3 + \cdots + \Delta_{S-2} +\Delta_{S-1}) \geq 1 \\
	 2\Delta_S \geq 1
\end{array}
\right.\right\}.
\end{align}
It can be verified that
\begin{align}
	(\Delta_1,\cdots,\Delta_{S-1}, \Delta_S) = \left( \underbrace{\frac{1}{2(S-2)},\cdots,\frac{1}{2(S-2)}}_{S-1}, \frac{1}{2}  \right) \in \mathcal{D}^{\fullent}(\mathcal{W}).
\end{align}
We thus obtain that $\Delta_1 +\Delta_2+\cdots+\Delta_S = \frac{S-1}{2(S-2)}+\frac{1}{2} = \frac{2S-3}{2S-4}$. It then follows that $C^{\fullent}(\mathcal{W}) = \left(\min_{(\Delta_1,\cdots, \Delta_S) \in \mathcal{D}^{\fullent}(\mathcal{W})}\sum_{s\in [S]} \Delta_{s} \right)^{-1} \geq \frac{2S-4}{2S-3}$.

Next we show that $C^{(S-1)}(\mathcal{W}) \leq \frac{2S-5}{2S-4}$ by Theorem \ref{thm:main}. 
Note that for the entanglement distribution map $\mathcal{E} = \binom{[S]}{S-1}$, the region specified in Theorem \ref{thm:main} contains $\Gamma = S(S-1)$ variables. This is because there are $T = S$ cliques and the size of each clique is $S-1$.
Also note that for this setting we have $\mathcal{E}(i)\not=\mathcal{E}(j)$ for $i\not=j$. Therefore, whenever it is needed to explicitly identify the servers in a specified clique, we use $\Delta_{\mathcal{E}(t),s}$ to replace $\Delta_{t,s}$, so that it becomes clear which servers are contained in the clique. Also, let $\Delta_{\mathcal{E}(t)} \triangleq \sum_{s\in \mathcal{E}(t)}\Delta_{\mathcal{E}(t),s}$. By Theorem \ref{thm:main}, the feasibility of ${\bm \Delta}$ implies that
\begin{align}
	\Delta_{\{1,2,\cdots,S-1\}} + \Delta_{\mathcal{W}(k) \cup \{S\}} + \sum_{\mathcal{A} \in \binom{[S-1]}{S-2} \backslash \{\mathcal{W}(k)\}} \sum_{s\in (\mathcal{A}\cup \{S\}) \cap \mathcal{W}(k)} 2\Delta_{\mathcal{A}\cup \{S\},s} \geq 1, ~~\forall k\in [S-1],
\end{align}
which can be simplified as
\begin{align} \label{eq:asymmetry_step_1}
	\Delta_{\{1,2,\cdots,S-1\}} + \Delta_{\mathcal{B} \cup \{S\}} + \sum_{\mathcal{A} \in \binom{[S-1]}{S-2} \backslash \{\mathcal{B}\}} \sum_{s\in \mathcal{A} \cap \mathcal{B}} 2\Delta_{\mathcal{A}\cup \{S\},s} \geq 1, ~~ \forall \mathcal{B} \in \binom{[S-1]}{S-2},
\end{align}
where we use $\mathcal{B}$ to substitute $\mathcal{W}(k)$ and  note that $\mathcal{W}(k) \cap \{S\} = \emptyset$ for any $k\in [S-1]$.
For the last data stream, since $\mathcal{W}(S) = \{S\}$, data stream ${\sf W}_S$ is not available to any server in the clique $\{1,2,\cdots, S-1\}$ and thus the feasibility of ${\bm \Delta}$ implies that
\begin{align} \label{eq:asymmetry_step_2}
	\sum_{\mathcal{A}\in \binom{[S-1]}{S-2}} 2\Delta_{\mathcal{A} \cup \{S\},S} \geq 1.
\end{align}
Therefore,
{\small
\begin{align}
	&2S-4 \notag \\
	&= (S-1)+(S-3) \\
	&\leq \sum_{\mathcal{B}\in \binom{[S-1]}{S-2}} \left( \Delta_{\{1,2,\cdots,S-1\}} + \Delta_{\mathcal{B} \cup \{S\}} + \sum_{\mathcal{A} \in \binom{[S-1]}{S-2} \backslash \{\mathcal{B}\}} \sum_{s\in \mathcal{A} \cap \mathcal{B}} 2\Delta_{\mathcal{A}\cup \{S\},s} \right) + \underbrace{(S-3)\sum_{\mathcal{A}\in \binom{[S-1]}{S-2}} 2\Delta_{\mathcal{A} \cup \{S\},S}}_{\Xi_1} \label{eq:asymmetry_step_3}\\
	&=\underbrace{(S-1)\Delta_{\{1,2,\cdots,S-1\}} + \sum_{\mathcal{B}\in \binom{[S-1]}{S-2}} \Delta_{\mathcal{B}\cup \{S\}}}_{\Xi_2} + \sum_{\mathcal{B}\in \binom{[S-1]}{S-2}}\sum_{\mathcal{A} \in \binom{[S-1]}{S-2} \backslash \{\mathcal{B}\}} \sum_{s\in \mathcal{A} \cap \mathcal{B}} 2\Delta_{\mathcal{A}\cup \{S\},s}+\Xi_1  \\
	& = \Xi_1+ \Xi_2 + \sum_{\mathcal{B}\in \binom{[S-1]}{S-2}}\sum_{\mathcal{A}\in \binom{[S-1]}{S-2}} \sum_{s\in \mathcal{A}\cap \mathcal{B}} 2\Delta_{\mathcal{A}\cap \{S\},s} - \sum_{\mathcal{B}\in \binom{[S-2]}{S-1}} \sum_{s\in \mathcal{B}} 2\Delta_{\mathcal{B}\cup \{S\},s} \\
	& = \Xi_1+ \Xi_2 + \sum_{\mathcal{A}\in \binom{[S-1]}{S-2}}\sum_{\mathcal{B}\in \binom{[S-1]}{S-2}} \sum_{s\in \mathcal{A}\cap \mathcal{B}} 2\Delta_{\mathcal{A}\cap \{S\},s}- \sum_{\mathcal{B}\in \binom{[S-2]}{S-1}} \sum_{s\in \mathcal{B}} 2\Delta_{\mathcal{B}\cup \{S\},s} \\
	& = \Xi_1+ \Xi_2 + (S-2)\sum_{\mathcal{A}\in \binom{[S-1]}{S-2}}\sum_{s\in \mathcal{A}} 2\Delta_{\mathcal{A}\cup \{S\},s}- \sum_{\mathcal{B}\in \binom{[S-2]}{S-1}} \sum_{s\in \mathcal{B}} 2\Delta_{\mathcal{B}\cup \{S\},s} \label{eq:asymmetry_step_4} \\
	&=\Xi_1+\Xi_2 + (S-3)\sum_{\mathcal{A}\in \binom{[S-1]}{S-2}} \sum_{s\in \mathcal{A}}2\Delta_{\mathcal{A}\cup \{S\},s}\\
	&=\Xi_2 + (S-3)\sum_{\mathcal{A}\in \binom{[S-1]}{S-2}}\sum_{s\in \mathcal{A} \cup \{S\}}2\Delta_{\mathcal{A}\cup \{S\},s} \\
	&=\Xi_2 + (S-3) \sum_{\mathcal{A}\in \binom{[S-1]}{S-1}}2\Delta_{\mathcal{A}\cup \{S\}} \\
	&= (S-1)\Delta_{\{1,2,\cdots,S-1\}} + \sum_{\mathcal{A}\in \binom{[S-1]}{S-2}} \underbrace{\big(1+ 2(S-3)\big)}_{2S-5}\Delta_{\mathcal{A}\cup \{S\}}\\
	&\leq \Bigg( 2S-5 \Bigg) \left( \Delta_{\{1,2,\cdots, S-1\}} + \sum_{\mathcal{A} \in \binom{[S-1]}{S-2}}\Delta_{\mathcal{A}\cup \{S\}} \right) \label{eq:asymmetry_step_5}\\
	& = \left( 2S-5 \right)\sum_{t\in [T]}\Delta_{\mathcal{E}(t)}\\
	& = \left( 2S-5 \right) \sum_{t\in[T]}\sum_{s\in\mathcal{E}(t)} \Delta_{t,s}
\end{align}
}Step \eqref{eq:asymmetry_step_3} is by \eqref{eq:asymmetry_step_1} and \eqref{eq:asymmetry_step_2}. To see Step \eqref{eq:asymmetry_step_4}, the term $2\Delta_{\mathcal{A}\cup \{S\},s}$ is counted $S-2$ times for any specified $\mathcal{A}\in \binom{[S-1]}{S-2}$ and $s\in \mathcal{A}$, because there is exactly one $\mathcal{B}\in \binom{[S-1]}{S-2}$ such that $s \not\in \mathcal{A}\cap \mathcal{B}$, which is $\mathcal{B} = [S-1]\backslash\{s\}$. Therefore, there are $(S-2)$ $\mathcal{B}\in \binom{[S-1]}{S-2}$ for which $s \in \mathcal{A}\cap \mathcal{B}$. Step \eqref{eq:asymmetry_step_5} is because $(S-1)<2S-5$ and $\Delta_{\{1,2,\cdots, S-1\}}\geq 0$.

Since $C^{(S-1)}(\mathcal{W}) = \left(\min_{{\bm \Delta}\in \mathcal{D}}\sum_{t\in[T],s\in\mathcal{E}(t)} \Delta_{t,s} \right)^{-1}$, we conclude that $C^{(S-1)}(\mathcal{W}) \leq \frac{2S-5}{2S-4}$.

\section{Proof of Lemma \ref{lem:multicast}} \label{proof:multicast}
The converse is obvious, as when there is only one receiver, the capacity cannot exceed $\rk({\bf H}_1)$. Next let us consider the achievability. We want to design a scheme with batch size $L$ and $N$ channel uses such that $L/N = \min_{k\in [K]}\rk_q({\bf H}_k)$. Note that since we use the channel $N$ times, we can consider the input $\widetilde{X} \in \mathbb{F}_{q^N}^{n\times 1}$ and $\widetilde{Y} \in \mathbb{F}_{q^N}^{m_k \times 1}$. 
 Let $L,N$ be such integers that $q^N > K\min_{k\in[K]}\rk({\bf H}_k)$ and $L = N\min_{k\in[K]}\rk({\bf H}_k)$, i.e., $L/N = \min_{k\in[K]}\rk({\bf H}_k)$. Take $L$ symbols from the data stream ${\sf W}$ and regard it as a vector ${\bf W} \in \mathbb{F}_{q^N}^{L/N\times 1}$. For each $k\in [K]$, since $\rk({\bf H}_k) \geq L/N$, there exist matrices $\overline{\bf U}_k\in \mathbb{F}_{q^N}^{L/N \times m_k}$, $\overline{\bf V}_k \in \mathbb{F}_{q^N}^{n\times L/N}$ such that $\overline{\bf U}_k {\bf H}_k^T \overline{\bf V}_k = {\bf I}_{L/N}$. 
Now consider a matrix ${\bf V} \in \mathbb{F}_{q^N}^{n\times L/N}$ whose elements are variables in $\mathbb{F}_{q^N}$ with values yet to be determined. Note that $P_k \triangleq \det(\overline{{\bf U}}_k {\bf M}_k {\bf V})$ is a polynomial of degree $L/N$ in these variables, and it is not a zero polynomial because setting ${\bf V} = \overline{\bf V}_k$ yields the valuation $P_k = \det({\bf I}_{L/N}) = 1$. It follows that $P \triangleq \prod_{k\in[K]} P_k$ is a non-zero polynomial with degree $KL/N$. By Schwartz-Zippel Lemma, the probability of $P$ evaluating to zero is not more than $\frac{KL/N}{q^N} = \frac{K\min_{k\in[K]}\rk({\bf H}_k)}{q^N}<1$. Therefore, there exists a realization of ${\bf V}$ for which the evaluation of $P$ is non-zero $\implies$ $\overline{{\bf U}}_k{\bf H}_k^T {\bf V}$ is invertible for all $k\in [K]$ for this realization of ${\bf V}$. Now let ${\bf U}_k \triangleq (\overline{\bf U}_k{\bf H}_k^T {\bf V})^{-1} \overline{{\bf U}}_k$. We obtain that ${\bf U}_k{\bf H}_k^T{\bf V} = {\bf I}_{L/N}$ for all $k\in [K]$. Now, let the input at the transmitter be $\widetilde{X}= {\bf V} {\bf W}$. Receiver $k\in[K]$ then hears $\widetilde{Y}_k = {\bf H}_k^T \widetilde{X} = {\bf H}_k^T {\bf V} {\bf W}$. The decoding at Receiver $k$ is then ${\bf U}_k \widetilde{Y} _k= {\bf U}_k {\bf H}_k^T {\bf V} {\bf W} = {\bf W}$, which is $L$ symbols (considered in $\mathbb{F}_q$) of the data stream. Thus, the scheme achieves $L/N = \min_k \rk({\bf H}_k)$.

\section{Proof of Lemma \ref{lem:halfMDS}} \label{sec:proof_halfMDS}
The proof is by construction. We make use of the Generalized Reed Solomon (GRS) code. Let $\mathbb{F}_q$ be a field. $n,k\in\mathbb{N}$ such that $k\leq n$. ${\bm \alpha} = (\alpha_1,\cdots,\alpha_n)\in \mathbb{F}_q^{n}$, ${\bm u} = (u_1,\cdots,u_n)\in \mathbb{F}_q^{n}$, such that $\alpha_i\not=\alpha_j$ for $i\not=j$ and $u_i\not=0$ for $i\in [n]$. This requires that $q \geq n$.
Define
{\small
\begin{align}
  \mbox{GRS}_{k,n}^q({\bm \alpha}, {\bm u}) \triangleq
  \begin{bmatrix}
    u_1 & u_2 & u_3 & \cdots & u_n\\
    u_1\alpha_1 & u_2\alpha_2 & u_3\alpha_3 & \cdots & u_n\alpha_n \\
    u_1\alpha_1^2 & u_2\alpha_2^2 & u_3\alpha_3^2 & \cdots & u_n\alpha_n^2 \\
    \vdots & \vdots & \vdots & \vdots & \vdots & \\
    u_1 \alpha_1^{k-1} & u_2 \alpha_2^{k-1} & u_3 \alpha_3^{k-1} & \cdots & u_n \alpha_n^{k-1}
  \end{bmatrix} \in \mathbb{F}_q^{k \times n}
\end{align}
}as the generator matrix of an $[n,k]$ GRS code over $\mathbb{F}_q$. GRS codes have the following properties \cite{macwilliams1977theory}.
\begin{enumerate}
  \item GRS codes are MDS. Any $k$ columns of the matrix $\mbox{GRS}_{k,n}^q({\bm \alpha}, {\bm  u})$ are $\mathbb{F}_q$ linearly independent.
  \item The dual code of a GRS code is also a GRS code. In particular, there exists ${\bm v} = (v_1,v_2,\cdots,v_n)\in\mathbb{F}_q^n$, $v_i\not=0$ for $i\in [n]$ such that
  \begin{align}
    \mbox{GRS}_{k,n}^q({\bm \alpha}, {\bm u}) \cdot \mbox{GRS}_{n-k,n}^q({\bm \alpha}, {\bm  v})^T = {\bf 0}_{k \times (n-k)}.
  \end{align}
\end{enumerate}
Note that $\lceil N/2 \rceil + \lfloor N/2 \rfloor = N$.
Define
{\small 
\begin{align} \label{eq:half_MDS_construction}
  {\bf M} = \begin{bmatrix}
    \mbox{GRS}_{\lceil N/2 \rceil,N}^q({\bm \alpha}, {\bm  u}) & {\bf 0}_{\lceil N/2 \rceil\times N} \\
    {\bf 0}_{\lfloor N/2 \rfloor\times N} & \mbox{GRS}_{\lfloor N/2 \rfloor,N}^q({\bm \alpha}, {\bm  v})
  \end{bmatrix}  ~~\in \mathbb{F}_q^{N\times 2N}.
\end{align}
}We claim that this ${\bf M}$ is half-MDS and it is a valid transfer matrix of an $N$-sum box. Note that the idea of placing the generator matrices of two codes that are dual to each other on the diagonal to construct a SSO matrix follows the CSS construction \cite{Calderbank_Shor_CSS_code, Steane_CSS_code}. Now we have,
{\small
\begin{align}
  &({\bf M} {\bf J}_{2N}) {\bf M}^T \notag \\
  &=\begin{bmatrix}
     {\bf 0}_{\lceil N/2 \rceil\times N} & -\mbox{GRS}_{\lceil N/2 \rceil,N}^q({\bm \alpha}, {\bm  u}) \\
    \mbox{GRS}_{\lfloor N/2 \rfloor,N}^q({\bm \alpha}, {\bm  v}) & {\bf 0}_{\lfloor N/2 \rfloor\times N}
  \end{bmatrix}
  \begin{bmatrix}
    \mbox{GRS}_{\lceil N/2 \rceil,N}^q({\bm \alpha}, {\bm  u}) & {\bf 0}_{\lceil N/2 \rceil\times N} \\
    {\bf 0}_{\lfloor N/2 \rfloor\times N} & \mbox{GRS}_{\lfloor N/2 \rfloor,N}^q({\bm \alpha}, {\bm  v})
  \end{bmatrix}^T \\
  &= \begin{bmatrix}
    {\bf 0}_{\lceil N/2 \rceil \times \lceil N/2 \rceil} & -\mbox{GRS}_{\lceil N/2 \rceil,N}^q({\bm \alpha}, {\bm  u})\cdot \mbox{GRS}_{\lfloor N/2 \rfloor,N}^q({\bm \alpha}, {\bm  v})^T \\
    \mbox{GRS}_{\lfloor N/2 \rfloor,N}^q({\bm \alpha}, {\bm  v})\cdot \mbox{GRS}_{\lceil N/2 \rceil,N}^q({\bm \alpha}, {\bm  u})^T & {\bf 0}_{\lfloor N/2 \rfloor \times \lfloor N/2 \rfloor}
  \end{bmatrix} \\
  &= {\bf 0}_{N\times N} 
\end{align}
}Finally, since GRS codes are MDS, it follows that the ${\bf M}$ constructed in \eqref{eq:half_MDS_construction} is half-MDS. Therefore, if the field size $q \geq N$, there exists an $N$-sum box operating in $\mathbb{F}_q$ that has a half-MDS transfer matrix. \hfill \qed

\bibliographystyle{IEEEtran}
\bibliography{../yy.bib}
\end{document}